\documentclass[%
 aip,
 amsmath,amssymb,
 reprint,%
]{revtex4-2}

\usepackage{graphicx}
\usepackage{dcolumn}
\usepackage{bm}

\usepackage[utf8]{inputenc}
\usepackage[T1]{fontenc}
\usepackage{mathptmx}
\usepackage{etoolbox}
\usepackage{geometry}
\usepackage{color}
\usepackage{float}
\usepackage{mathrsfs}
\usepackage{mathtools}
\usepackage{amsthm}
\usepackage{graphicx}
\usepackage{amssymb}
\usepackage{amsfonts}
\usepackage{amstext}
\usepackage{dsfont}
\usepackage{mathdots}
\usepackage{xcolor}
\usepackage{lipsum}
\usepackage{array}
\usepackage{multirow}
\usepackage{babel}
\usepackage{subfig}
\usepackage[utf8]{inputenc}
\usepackage[linesnumbered,ruled,vlined]{algorithm2e}
\usepackage[draft]{changes}

\newtheorem{theorem}{Theorem}

\newtheorem{lemma}{Lemma}[theorem]
\newtheorem*{remark}{Remark}

\renewcommand{\arraystretch}{1.1}

\makeatletter
\def\@email#1#2{%
 \endgroup
 \patchcmd{\titleblock@produce}
  {\frontmatter@RRAPformat}
  {\frontmatter@RRAPformat{\produce@RRAP{*#1\href{mailto:#2}{#2}}}\frontmatter@RRAPformat}
  {}{}
}%
\makeatother

\begin{document}

\preprint{AIP/123-QED}

\title[A Novel Self-Adaptive SIS Model]{A Novel Self-Adaptive SIS Model Based on the Mutual Interaction between a Graph and its Line Graph}

\author{Paolo Bartesaghi}
\email{paolo.bartesaghi@unimib.it}
\affiliation{University of Milano - Bicocca, Via Bicocca degli Arcimboldi 8, 20126 Milano, Italy.}

\author{Gian Paolo Clemente}
\affiliation{Universit\`{a} Cattolica del Sacro Cuore di Milano, Largo Gemelli 1, 20123 Milano, Italy}

\author{Rosanna Grassi}
\affiliation{University of Milano - Bicocca, Via Bicocca degli Arcimboldi 8, 20126 Milano, Italy.}

\date{\today}

\begin{abstract}
We propose a new paradigm to design a network-based self-adaptive epidemic model that relies on the interplay between the network and its line graph. We implement this proposal on a Susceptible-Infected-Susceptible model in which both nodes and edges are considered susceptible and their respective probabilities of being infected result in a real-time re-modulation of the weights of both the graph and its line graph.
The new model can be considered as an appropriate perturbation of the standard Susceptible-Infected-Susceptible model, and the coupling between the graph and its line graph is interpreted as a reinforcement factor that fosters diffusion through a continuous adjustment of the parameters involved. We study the existence and stability conditions of the endemic and disease-free states for general network topologies. Moreover, we introduce, through the asymptotic values in the endemic steady states, a new type of eigenvector centrality where the score of a node depends on both the neighboring nodes and the edges connected to it. We also investigate the properties of this new model on some specific synthetic graphs, such as cycle, regular, and star graphs. Finally, we perform a series of numerical simulations and prove their effectiveness in capturing some empirical evidence on behavioral adoption mechanisms.
\end{abstract}

\maketitle

\begin{quotation}
The spread of a disease within a population, the propagation of a shock among financial institutions, the diffusion of opinions in online social networks, or the adoption of a behavior by members of a community, are all examples of diffusive phenomena within a network of interacting individuals.
Despite their similarities, these processes can be very different and cannot always be reduced to simple models.
For example, the phenomenon by which a repeated message becomes a personal belief and is adopted by an individual is not the same as the spread of a cold through contact between individuals in the same social network.
It is known that opinions and behaviors require reinforcement, and only when the individual is reached by multiple messages does he or she adopt them, whereas a single contact could be sufficient to transmit a sexual disease.
Similarly, the propagation of shocks through financial or transportation networks is hardly captured by standard contagion models. 
In this paper, we propose a new diffusion model in networks that exploits the mutual interaction between spread processes over nodes and edges. This mutual reinforcement is able to explain some well-known empirical evidence about adoption mechanisms and how they differ from other contagion processes.
\end{quotation}

\section{Introduction}
\label{Introduction}

In the last few decades, massive research efforts have focused on evolutionary and dynamical models in complex networks. The spread of a disease within a population, the propagation of a financial shock among banks, the ripple effect of an accident on traffic and transportation networks or an attack on cybersecurity networks, the spread of trends in online social networks, or the adoption of a behavior by members of a given community are all examples of diffusive phenomena in networks of interconnected entities.\cite{Vespignani2008} Despite their similarities, these processes can be very different and can hardly be reduced to simple contagion models.

The pioneering work by \citet{Kermack1927} in 1927 embedded for the first time an epidemic process in a closed population with homogeneous mixing. Since then, it has been pointed out that compartmental models such as the Susceptible-Infected-Susceptible (SIS) and  Susceptible-Infected-Recovered (SIR) models fail to describe many types of propagation phenomena and several variants have been proposed to provide a more realistic representation of the spread dynamics in different contexts. \cite{Ball2015,bartesaghi2021b,Liang2023,Zixiang2023,Leng2022,Li2018,Schreiber2021,Yakubu2006}

For example, \citet{Mieghem2015} propose a generalization to the so-called $\varepsilon$-SIS model by adding a source of self-infection in a cybersecurity network due, for instance, to accessing malicious websites, opening emails with worms, or downloading files containing malware. The heterogeneous SIS (H-SIS) model proposed by \citet{Mieghem2017} allows the infection rate along each link to be different, and makes these rates dependent on the type of connection between the two nodes. \citet{Yeftanus2021}, starting from a market basket analysis, construct a weighted communication network of different computers of a given company, and propose a HG-SIS model as a generalization of the H-SIS model, in which the infection rate is a function of the communication weight and a self-infection is allowed.

In general, the state of a node can influence the infection rate by altering the flow along a given edge. For example, in a transportation network, nodes affected by a shock, such as an accident, can induce changes in movement patterns, thereby affecting the likelihood of shock transmission in the network. Addressing this issue, \citet{Punzo2022} proposes a flow-regulated infection rate which accounts for the tendency of infection carriers to prefer healthy nodes over infected ones.

The limitations of the SIS and SIR models appear most clearly in the context of social interactions, opinion dissemination, and behavioral adoption, where it is well known that a single exposure to a piece of information is not sufficient for an individual to adopt that opinion or behavior.

The process through which a repeated message transforms into a personal belief and is adopted by an individual within a social network differs significantly from the process of spreading a cold through direct contact between individuals in the same network. Opinions and behaviors require reinforcement, and only when the individual is exposed to multiple messages does he or she adopt them. \cite{Centola2010} This leads to unexpected interactions with the topological structure of the network, which responds differently depending on the type of diffusion process it hosts. For instance, a disease and an opinion spread very differently in regular networks such as lattices compared to random networks.\cite{Zheng2013}

Our proposal stems from the search for a model that is flexible enough to potentially adapt to different contexts. The key idea is to design a process in which the weights on the edges in the network adapt to the actual epidemiological state of the nodes, and vice versa. These weights are not statically assigned at the beginning of the process; instead, they naturally emerge as the outcome of a coupled secondary process. We refer to this process as the \lq\lq dual process\rq\rq, and its interplay with the primary one results in a unique, brand-new diffusion process that we call the self-Adaptive SIS (ASIS) model.

In other words, we avoid suggesting any extrinsic dynamic process that superimposes on the epidemiological model by modifying the infection rates along the edges according to arbitrary criteria. Rather, the system autonomously adapts to the actual epidemiological state of the network. This results in a tunable coupling between the primary and the dual process, which can be interpreted as a reinforcement effect in message transmission.

The concept of a reciprocal action in which node and edge attributes are mutually dependent has recently been used to propose a nonlinear eigenvector centrality for both nodes and edges. \cite{Tudisco2021} The purpose of the authors is to define a mutually reinforcing static centrality measure, in which the node's score inherits that of its connecting edges and the edge's score that of its extreme nodes. Instead, our goal is to employ a dynamic approach that leverages a similar but distinct mutual reinforcement between the attributes of nodes and edges. To this end, in the ASIS model, the node score is associated with the asymptotic steady state probability of that node in the primary process and is influenced by the score of the connected edges. Similarly, the edge score is associated with its steady state probability in the dual process and is contingent on the scores of the nodes at its ends. These scores evolve simultaneously, interacting with each other over time.

One of the implications of the ASIS model is that it induces a brand new definition of self-adaptive eigenvector centrality. Traditionally, eigenvector centrality assigns importance to nodes based on the importance of their neighbors.
Our model allows to consider jointly nodes and edges relevance. Node centrality is indeed  proportional to the product of the scores of its neighboring nodes by that of the corresponding edges connecting that node to its neighbors. Differently from  \citet{Tudisco2021}, these scores emerge at the end of an iterative process that gradually updates them to stationary values.

The paper is structured as follows. In Section \ref{Motivation and Model overview}, we provide the motivations and the intuition behind the ASIS model. The main background and preliminaries are introduced in Section \ref{Background}. In Section \ref{The Self-Adaptive SIS Model}, we describe the details of the ASIS model and the analytical results for the cycle and complete graph. Section \ref{General steady states analysis} focuses on steady states and the related nonlinear eigenproblem. The self-adaptive eigenvector centrality is defined in Section \ref{Self-adaptive eigenvector centrality}. The model is tested on  an illustrative example in Section \ref{Illustrative example} and a variety of numerical simulations in Section \ref{Numerical Experiments}, while its effectiveness in online social networks is discussed in Section \ref{Social reinforcement in lattice and random networks}. Conclusions follow.

\section{Motivation and Model overview}
\label{Motivation and Model overview}

In 2010, \citet{Centola2010} conducted an influential experiment on the spread of behavior in online social networks, highlighting the pivotal role of social reinforcement in the adoption process. Social reinforcement pertains to the common scenario where an individual requires multiple cues from peers before adopting a particular opinion or behavior. \cite{Peyton2009,Onnela2010} Indeed, the experiment showed that a single signal exerts minimal influence on individuals' decision making, while redundant signals can improve the probability of approval and behavior adoption. It is only when a node receives a reinforced message that it may transition to adopting the opinion or behavior it carries.

The reinforcement effect in the spread of information, opinions, and behaviors within social networks, particularly in online contexts where face-to-face interactions are absent, is known to radically alter diffusion dynamics compared to the case of biological diseases. For example, it has been observed that in cases where the infection rate is not too high, reinforcement favors diffusion in regular networks over random networks. \cite{Zheng2013}

It is, therefore, important to devise a mechanism that takes into account the intensity with which a given node is able to transmit a message.
In the standard SIS epidemic model on networks, initial infection probabilities for nodes evolve over time according to a dynamics that depends on the infection rate $\beta$, the recovery rate $\gamma$ and, assuming a weighted network, on a static assignment of weights to the edges.  The weight of the edge conveys how likely that edge is to be a channel for the spread of the infection. Hence, the potential of an edge to transmit the infection may be different from edge to edge due to the intrinsic and topological features of the network. However, this capability may vary over time as a result of the diffusion itself.\cite{Roberts2015}

For opinions and behaviors, the more information individuals receive, the more inclined they are to accept them. However, this information is obtained from other individuals who are engaged in the same process and who may be more or less convinced or at a more or less advanced stage in the adoption process. The presence of an edge and its initial weight are not sufficient to explain this phenomenon, unless the weight is adjusted over time based on the level of actual infection/adoption of neighboring nodes.

The evolution of these weights over time can be described by a similar contagion process. In fact, the intensity of the message transmitted along an edge is ultimately governed by the probabilities that nodes at its ends are at varying stages of the adoption process and evolve accordingly.

We then introduce an adaptive reinforcement mechanism in the signal transmission from one node to another that accounts for the graded nature typical of social responses in contrast to the \textit{all-or-nothing} nature which is more typical of infectious diseases spread. Furthermore, since we leave open the possibility that an individual may suddenly abandon the idea or behavior for various reasons and return to the susceptible state, we turn to an SIS-type model.

To further support this idea, let us consider this analogy. In a traffic network, nodes represent locations, such as squares, intersections, or prominent sites, while edges denote streets, roads, or connections between them. Consider a shock propagating across the network. When a location is affected by an accident, the edges linked to that node experience traffic blockage, regardless of their weights in terms of traffic volume. A realistic model should therefore update the weights of those edges, to reflect the heightened probability of the shock being transmitted along a road originating from that location. In other words, the probability that an edge is a channel for the transmission of an infection \textit{is not independent} of the probability that its end points are infected. This mirrors the fact that the probability that a node is infected at time $t$ is not independent of the probability that an edge would transmit the epidemic, which is typically expressed by its weight.

A natural way to implement this idea is to run	two parallel SIS processes over nodes and over edges. More precisely, to consider an auxiliary, or \textit{dual}, process in which the information propagates \textit{among edges through the nodes}; that is, a process occurring in a new network in which edges become nodes and nodes become edges. This network is usually defined in the literature as \textit{line graph}. \cite{Gross2013} Specifically, we consider two SIS processes, one on the original network $G_P$ (primary process) and one on its line graph $G_D$ (dual process). The updated values of the edge weights are computed as outcomes of the dual process on the line graph. Let ${G_P}=({V_P},{E_P})$ be the primary network and ${G_D}=({V_D},{E_D})$ the corresponding line graph and denote by $x_i(t)$ the probability that node $i\in {V_P}$ is infected at time $t$, and $y_j(t)$  the probability that node $j\in {V_D}$ is infected at time $t$. The probabilities $y_{j}(t)$ will serve as weight attributes for the edges in ${E_P}$, while the probabilities $x_{i}(t)$ will be used as weight attributes for the edges in ${E_D}$. In this way, we generate a pair of intertwined processes that evolve simultaneously over time  utilizing the probability derived by each other. The model works jointly on both networks, leveraging the interrelated properties of nodes and edges.


To further illustrate the intuition behind the proposed mechanism, consider the binary network $G_P$ shown in Fig. 1, panel (a). Edges are labeled by letters $a$, $b$, $c$ and $d$. Panel (b) shows the corresponding line graph $G_D$, in which the nodes adopt the labels of the corresponding edges and the edges retain the colors of the corresponding nodes in $G_P$. At each step, the probabilities obtained through the evolution of an SIS process on $G_P$ are assigned as edge weights of the dual network $G_D$, as shown in panel (c). Similarly, the probabilities obtained from the SIS process on $G_D$ are reassigned to the primary network $G_P$ in the form of updated edge weights, as shown in panel (d). This simple example will be analyzed in more depth in Section \ref{Illustrative example}, after discussing the details of the model.
\begin{widetext}
\begin{center}
\begin{figure}[H]
	\centering
	\subfloat[]{\includegraphics[width=0.36\textwidth]{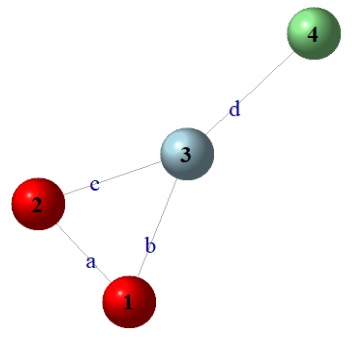}}\hspace{10mm}
	\subfloat[]{\includegraphics[width=0.36\textwidth]{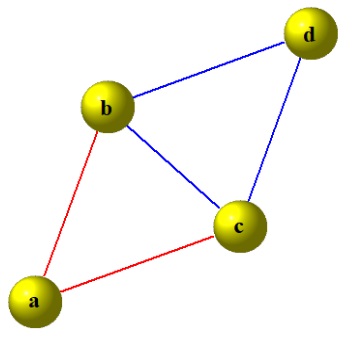}}\\
	\subfloat[]{\includegraphics[width=0.48\textwidth]{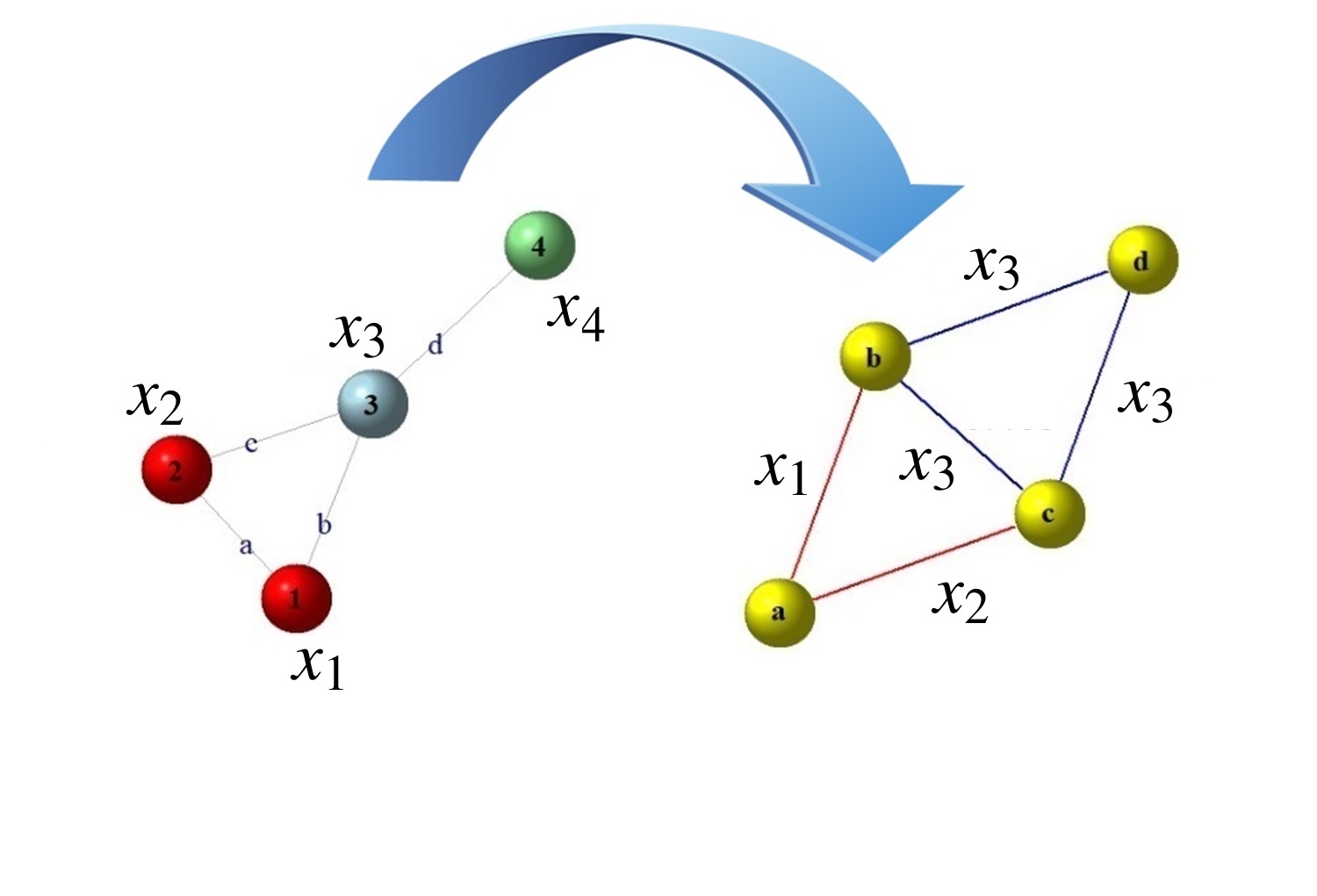}}
	\subfloat[]{\includegraphics[width=0.48\textwidth]{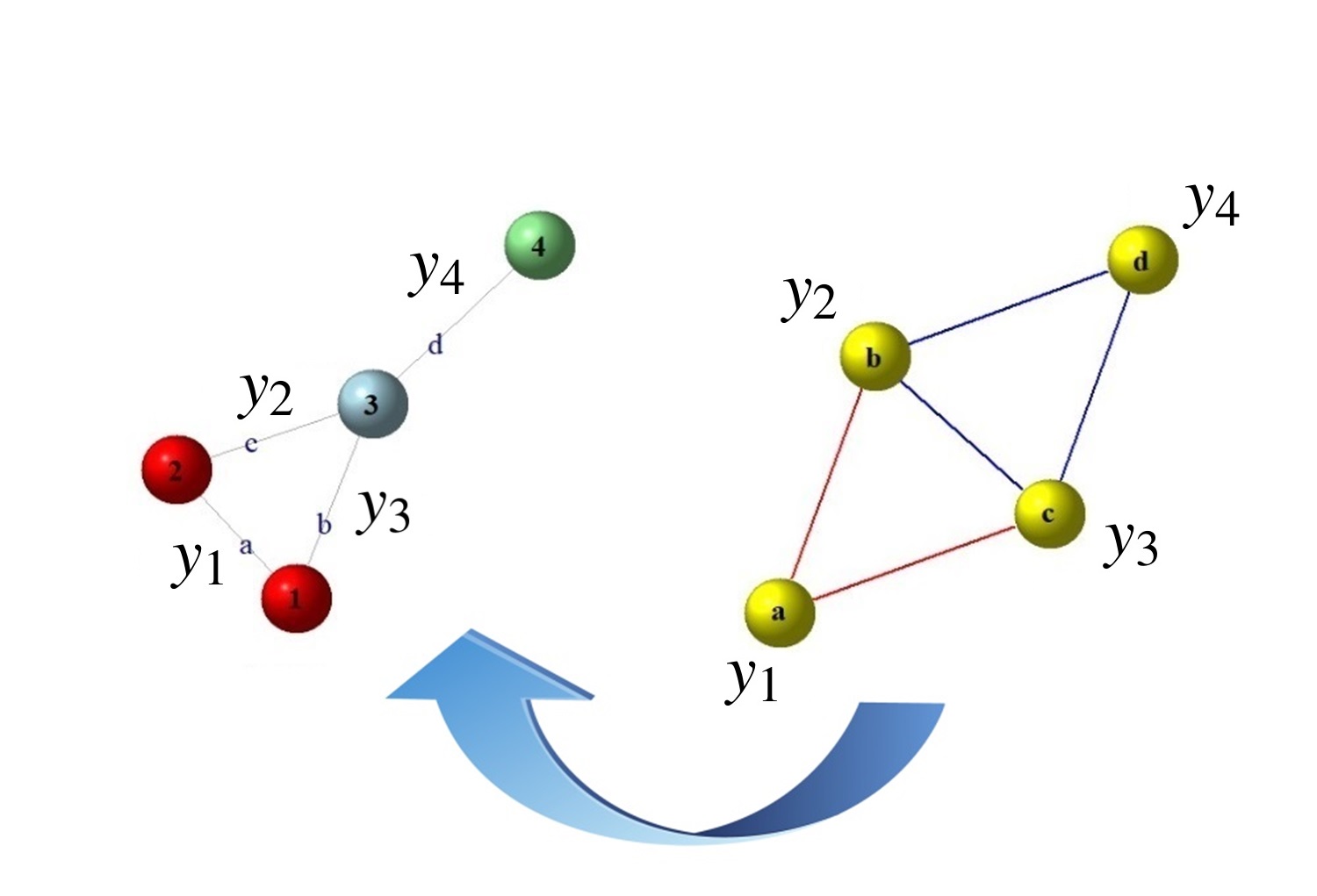}}
	\caption{Primary network model, panel (a), and its line graph, panel (b). Panels (c) and (d) illustrate the interplay between the two processes and the mechanism that models the interdependence between the probabilities $x_{i}(t)$ and $y_{i}(t)$.}
	\label{Cartoon} 
\end{figure}
\end{center}
\end{widetext}

\section{Background}
\label{Background}
	We give a brief overview of standard SIS models, mainly to recall some background ideas and to introduce notations that will be extended in the rest of the paper. The scalar SIS model is described by the following differential equation
	\begin{equation}
		\left\{ 
		\begin{array}{l}
			\dot{x}(t)=\beta[1-x(t)]x(t) -\gamma x(t) \\
			x(0)=p \\
		\end{array}
		\right.
		\label{SIS1}
	\end{equation}
	where $x(t)$ denotes the prevalence of infected individuals at time $t$, $\beta\geq 0$ is the infection rate, $\gamma\geq 0$ the recovery rate, and $0<p<1$ the initial prevalence of infected individuals at time $t=0$. If $n$ is the population size, then $nx(t)$ is the size of the infected compartment, and $\beta nx(t)$ is the total rate of infectious contacts. Conversely, $n[1-x(t)]$ is the size of the susceptible compartment. A closed solution for Eq. (\ref{SIS1}) is
	\begin{equation}
		x(t)=\frac{\left(1-\frac{\gamma}{\beta} \right)p}{p+(1-\frac{\gamma}{\beta}-p)e^{-\beta \left(1-\frac{\gamma}{\beta} \right)t}}
		\label{SISsolution1}
	\end{equation}
	where ${\mathcal{R}}\coloneqq\frac{\beta}{\gamma}$ is called \textit{basic reproductive ratio}. It is known from the related literature that, if ${\mathcal{R}}<1$, all trajectories converge to the unique disease-free steady state $x^{\star}=0$ and the epidemic disappears. If ${\mathcal{R}}>1$, each trajectory from initial condition $0<p<1$ converges to the exponentially stable endemic steady state $x^{\star}=1-\frac{\gamma}{\beta}$, and the disease-free steady state is unstable. Therefore, a transcritical bifurcation occurs at ${\mathcal{R}}=1$ (see \citet{Kiss2017} for an in-depth discussion).
	
	The first step toward an SIS model on networks is the Kermack-McKendrick model. \cite{Kermack1927} This model is based on the homogeneous mean-field assumption that nodes have an average number of neighbors $\left\langle k \right\rangle$ and that their
	degrees have only small fluctuations around this mean value. With the same meanings as before, the prevalence evolution equation is given by
	\begin{equation}
		\left\{ 
		\begin{array}{l}
			\dot{x}(t)=\beta \left\langle k \right\rangle [1-x(t)] x(t) -\gamma x(t) \\
			x(0)=p \\
		\end{array}
		\right.
		.
		\label{ Kermack-McKendrick}
	\end{equation}
	By setting $x^{\star} \cdot \left[ \beta \left\langle k \right\rangle (1-x^{\star})-\gamma \right]=0$, the equilibrium states are obtained: $x^{\star}=0$ and $x^{\star}=1-\frac{\gamma}{\beta \left\langle k \right\rangle}$. For ${\mathcal{R}}< \frac{1}{\left\langle k \right\rangle}$, $x^{\star}=0$ is asymptotically stable; whereas, for ${\mathcal{R}}> \frac{1}{\left\langle k \right\rangle}$, $x^{\star}=1-\frac{\gamma}{\beta \left\langle k \right\rangle}$ is asymptotically stable and $x^{\star}=0$ is unstable. 
	This means that $\tau\coloneqq \frac{1}{ \left\langle k \right\rangle}$ represents the threshold below which the epidemic cannot spread, since more nodes are recovered by $\gamma$ than are infected by $\beta$.
	
	Let us now turn to the network-based model. We consider an undirected weighted network $G=(V,E)$ with $n\times n$ adjacency matrix ${\bf A}=[A_{ij}]$. We denote henceforth by ${\bf k}=[k_1, \dots ,k_n]^T$ the degree vector of $G$, by $\lambda_{i}, \ i=1, \dots, n$, the eigenvalues of ${\bf A}$ with $\lambda_1\geq \lambda_2\geq \dots \geq\lambda_n$ and $\psi_{i}, \ i=1, \dots, n$, the corresponding eigenvectors. A weight $w_{ij} \in [0,1]$ is associated with each edge $(i,j)$. The weight represents the probability that the disease is transmitted along that edge, or, in other words, how likely that edge is to be a channel for the spread of infection. The SIS model on network is then described by the $n$ differential equations represented, in matrix form, by the following system
	\begin{equation}
		\left\{ 
		\begin{array}{l}
			\dot{{\bf x}}(t)=\beta\left[ {\bf I}_{n} - {\rm diag}\, {\bf x}(t) \right] {\bf A}\, {\bf x}(t)-\gamma {\bf x}(t)\\
			{\bf x}(0)={\bf p}\\
		\end{array}
		\right.
		\label{SIS3}
	\end{equation}
	where ${\bf I}_{n}$ is the $n\times n$ identity matrix and ${\rm diag}\, {\bf x}(t)$ is the diagonal matrix whose diagonal entries are $x_i(t),i=1,...,n.$ 
	Let us remark that we make some assumptions that we will preserve throughout the paper: first, $\beta$ and $\gamma$ are the same for all nodes; second, the initial infection probabilities $p_i$ are identical on all nodes, namely ${\bf p}=p {\bf u}_{n}$ where $p\in(0,1]$ and ${\bf u}_{n}=[1,1,\dots,1]^{T}\in{\mathbb R}^{n}$.
	
	\noindent \textit{Steady states}\\
	Although a closed solution of the non-linear problem in Eq. (\ref{SIS3}) cannot be provided, we can obtain information about its steady states.
	A steady state is achieved when $\dot{{\bf x}}(t)=\bf 0$ as $t \to +\infty$,  that is 
	\begin{equation*}
		\beta\left[ {\bf I}_{n} - {\rm diag}\, {\bf x} \right] {\bf A}\, {\bf x}-\gamma {\bf x}=\bf 0.
	\end{equation*}
	
	The disease-free steady state is given by the trivial solution $x_{i}=0$, $\forall i$. If we consider the linearization of Eq. (\ref{SIS3}) around the disease-free steady state, that is 
	\begin{equation*}
		\left\{ 	
		\begin{array}{l}
			\dot{\bf x}(t)=\beta{\bf A}\, {\bf x(t)}-\gamma {\bf x(t)}\\
			{\bf x}(0)={\bf p}\\
		\end{array}
		\right.
		,
	\end{equation*}
	then the study of the steady state involves the resolution of the eigenvalue problem
	\begin{equation*}
	\beta{\bf A}\, {\bf x}-\gamma {\bf x}=0.
	\end{equation*}
	Recalling that $\lambda$ is an eigenvalue of $\bf A$  if and only if $\beta\lambda-\gamma$ is an eigenvalue of $\beta{\bf A}-\gamma {\bf I}_n$, the threshold that ensures the stability of the null solution is given by $\beta \lambda_{1}-\gamma < 0$, or, equivalently, ${\mathcal{R}}< \frac{1}{\lambda_{1}}$. Therefore, if the reproductive ratio is less than $\frac{1}{\lambda_{1}}$, this state is stable and the process dies out.
	
	Conversely, if ${x_i}>0$ for at least one $i$, the system evolves into an endemic state. In this case, the
	steady states can be viewed as solutions of the implicit recurrence relation
	\begin{equation}
		{x_i}=\frac{\beta \sum_{j=1}^{n}A_{ij}x_{j} }{\gamma + \beta \sum_{j=1}^{n}A_{ij}x_{j}}=
		1-\frac{1}{1+{\mathcal{R}} \sum_{j=1}^{n}A_{ij}x_{j}}.
		\label{recurrence1}
	\end{equation}
	A sufficient condition for the existence of the endemic state is ${\mathcal{R}}>\max_{i} \frac{1}{k_i}$,  where $k_i$, $\forall i=1,\dots, n$, is the degree of node $i$. In this case, there exists a constant $c\leq 1-\frac{\gamma}{\beta {k_i}}$, $\forall i=1,\dots, n$, such that $c\leq {x_i}\leq 1$ holds for all $i$ (see \citet{Kiss2017}). It has been shown that this condition can be relaxed. Indeed, a refined sufficient condition for the existence and stability of the endemic steady-state solution is ${\mathcal{R}}>\frac{1}{\lambda_{1}}$. It can be proved that, under this condition, if ${\bf x}(0)\in [0,1]^n$ then ${\bf x}(t)\in [0,1]^n$ for all $t>0$ and if ${\bf x}(0)>{\bf 0}$ then ${\bf x}(t)>{\bf 0}$ for all $t>0$.\cite{Kiss2017} Moreover, there exists an equilibrium point ${\bf x}^{\star}={\bf 0}$, the epidemic outbreak, which is exponentially unstable, and an endemic state ${\bf x}^{\star}\neq {\bf 0}$, which is  exponentially stable. Something more can be said about the behavior of the endemic solution at the ends of the interval $\left( \frac{1}{\lambda_{1}}, +\infty \right)$ in ${\mathcal{R}}$: if ${\mathcal{R}}\to\left( \frac{1}{{\lambda_{1}}} \right)^+$ then ${\bf x}^{\star} \to a\left( {\mathcal{R}}\lambda_{1}-1 \right) {{{\psi}_{1}}}\quad {\rm with} \quad a=\frac{||{{{\psi}_{1}}}||^2}{{{{\psi}_{1}^{T}}{\rm diag}({{{\psi}_{1}}}){{\psi}_{1}}}}$, while if ${\mathcal{R}}\to +\infty$ then ${\bf x}^{\star}\to {\bf u}_{n}-\frac{1}{{\mathcal{R}}}\, {\rm diag}\, {\bf k}^{-1}$.
	
	It should be emphasized that the value $ \frac{1}{\lambda_{1}}$ represents a lower bound for the actual threshold $\tau$ of the process \textit{in networks}, $\tau\geq  \frac{1}{\lambda_{1}}$ (while it as an exact value for the N-intertwined mean-field approximation, see \citet{Mieghem2015}). For some graphs, such as the complete graph, this value is a good approximation of the actual threshold, while for other graphs, such as the star, it is less accurate. In general, the larger the heterogeneity in the degree distribution, the larger the deviation from the first-order mean-field approximation. For $d$-regular graphs, where all degrees are equal to $d$, the lower bound is $ \frac{1}{\lambda_{1}}=\frac{1}{d}$.
	
\section{The Self-Adaptive SIS Model}
\label{The Self-Adaptive SIS Model}
	
	\subsection{Primary network epidemic model and its dual}
	To facilitate the understanding of the model architecture, we initially assume that the primary network and its dual are unweighted. Hence, let us suppose 
	that the original network $G_{P}$ is represented by a binary undirected graph with adjacency matrix ${\bf B}_{P}\in {\mathbb R}^{n\times n}$ and incidence matrix ${\bf E}\in {\mathbb R}^{n\times m}$, and the dual binary network $G_D$ has adjacency matrix ${\bf B}_{D}\in {\mathbb R}^{m\times m}$.
	By graph theory, it is known that 
	${\bf B}_{P}={\bf E} {\bf E}^{T}-{\rm diag(\bf k}_P)$ and ${\bf B}_{D}={\bf E}^{T}{\bf E}-{\rm diag(\bf k}_D)$, where ${\rm diag}({\bf k}_P)$ is the diagonal matrix with diagonal entries given by the node degrees of the network $G_P$, and ${\rm diag}({\bf k}_D)$ is the analog diagonal matrix of the dual network $G_D$. Note that, in the latter case, the diagonal entries count the number of nodes each edge contains, hence ${\rm diag}({\bf k}_D)=2{\bf I}_{m}$.
	
	Now suppose that both the nodes and the edges of the network $G_{P}$ are assigned numerical attributes represented by vectors ${\bf x}=[x_{1}, \dots , x_{n}]^{T}$ and ${\bf y}=[y_{1}, \dots , y_{m}]^{T}$, respectively.
	
	The attributes ${\bf y}$ of the edges in the network $G_{P}$ can be naturally and uniquely assigned to the nodes of the dual network $G_{D}$ in a one-to-one correspondence. Conversely, to assign attributes to the edges of the dual network $G_{D}$ from those assigned to the nodes of $G_{P}$, we proceed as follows. 
	An edge in $G_D$ is the bridge between two vertices in $G_D$ and corresponds to a specific node in $G_P$. This node is the common end of the two corresponding edges in $G_P$.  Therefore, we assign to an edge in $G_D$ the same attribute $x_i$ as the common node between the two edges in $G_P$. Of course, the same attribute $x_i$ can be used multiple times.
	
	The adjacency matrices of the networks $G_{P}$ and $G_{D}$ are then modified as follows
	\begin{equation}
		\left\{ 
		\begin{array}{l}
			{\bf A}_{P}={\bf E}\, {\rm diag}\, {\bf y}\, {\bf E}^{T}-{\rm diag}\, {\bf k}_P\\
			\hfill \\
			{\bf A}_{D}={\bf E}^{T} {\rm diag}\, {\bf x}\, {\bf E}-{\rm diag}\, {\bf k}_D\\
		\end{array}
		\right.
		\label{ApAd}
	\end{equation}
	where ${\bf k}_P={\bf E}{\bf y}$ and ${\bf k}_D={\bf E}^{T}{\bf x}$. These relations play a central role because they link the entries of the adjacency matrices of one network with the attributes of the nodes of the other.
	
	Setting the initial conditions  ${\bf x}(0)={\bf x}_{0}$ and ${\bf y}(0)={\bf y}_{0}$ on the nodes of $G_P$ and $G_D$, respectively, by Eq. \eqref{SIS3}, the two parallel SIS processes on the $G_P$ and $G_D$ networks are described by
	\begin{equation}
		\label{continuos_eqs}
		\left\{ 
		\begin{array}{l}
			\dot{{\bf x}}(t)=\beta_{P}\left[ {\bf I}_{n} - {\rm diag}\, {\bf x}(t) \right] {\bf A}_{P}\, {\bf x}(t)-\gamma_{P} {\bf x}(t) \\
			\hfill \\
			\dot{{\bf y}}(t)=\beta_{D}\left[ {\bf I}_{m} - {\rm diag}\, {\bf y}(t) \right] {\bf A}_{D}\, {\bf y}(t)-\gamma_{D} {\bf y}(t)\\
		\end{array}
		\right.
	\end{equation}
	where $\beta_{P}$ ($\beta_{D}$) and $\gamma_{P}$ ($\gamma_{D}$) are the infection and recovery rates on the primary (dual) network.
	
	What we aim to do is to consider the non-autonomous version of system (\ref{continuos_eqs}). Specifically, the time dependence of the two matrices ${\bf A}_{P}$ and ${\bf A}_{D}$ can be introduced by setting 
	${\bf k}_P(t)={\bf E}{\bf y}(t)$ and ${\bf k}_D(t)={\bf E}^{T}{\bf x}(t)$,	
	where ${\bf E}{\bf y}(t)$ returns, for each node in the network $G_P$, a weight equal to the sum of the attributes of the edges connected to that node, and ${\bf E}^{T}{\bf x}(t)$ returns, for each edge in the network $G_P$, a weight equal to the sum of the attributes of its two end nodes. Hence, the two adjacency matrices of the network $G_{P}$ and $G_{D}$ become, respectively
	\begin{equation}
		\left\{ 
		\begin{array}{l}
			{\bf A}_{P}({\bf y}(t))={\bf E}\, {\rm diag}\, {\bf y}(t) {\bf E}^{T}-{\rm diag}({\bf E}{\bf y}(t))\\
			\hfill \\
			{\bf A}_{D}({\bf x}(t))={\bf E}^{T}\, {\rm diag}\, {\bf x}(t) {\bf E}-{\rm diag}({\bf E}^{T}{\bf x}(t))\\
		\end{array}
		\right.
		.
		\label{ApAd_continuos}
	\end{equation}
	
	Let us emphasize that, by formula (\ref{ApAd_continuos}), edges in network $G_P$ inherit the weights from the node probabilities in network $G_D$ to produce an updated version of the adjacency matrix\footnote{In the following, for ease of reading, we will alternately use the equivalent notations ${\bf A}_{P}(t)$ and ${\bf A}_{P}({\bf y})$, instead of ${\bf A}_{P}({\bf y}(t))$. Similarly, ${\bf A}_{D}(t)$ and ${\bf A}_{D}({\bf x})$ instead of ${\bf A}_{D}({\bf x}(t))$.} ${\bf A}_{P}(t)$ at time $t$. Similarly, ${\bf A}_{D}(t)$ inherits the weights from the node probabilities in network $G_P$, by assigning to the edges of the dual network $G_D$ the probabilities of the corresponding nodes in network $G_P$ at time $t$ in a non-one-to-one correspondence.
	Expressions in Eq. (\ref{ApAd_continuos}) make clear that the adjacency matrix controlling the SIS evolution on the network $G_P$ depends on the attributes ${\bf y}(t)$ and the adjacency matrix controlling the SIS evolution on the network $G_D$ depends on the attributes ${\bf x}(t)$. It is worth noting that, in our model, originally binary networks become weighted networks in a natural way, through the introduction of node attributes and edge attributes.
	
	Equations (\ref{continuos_eqs}) can be conveniently expressed in a more compact form as
		\begin{equation}
		\resizebox{.88\hsize}{!}{$
		\left[\begin{array}{c}
			{\bf \dot{x}} \\
			{\bf \dot{y}} \\
		\end{array} \right]=\\
		\left[ \begin{array}{cc} 
			\beta_{P}\left[ {\bf I}_{n} - {\rm diag}\, {\bf x} \right] {\bf A}_{P}({\bf y})-\gamma_{P} {\bf I}_{n} & {\bf 0}_{n\times m} \\
			{\bf 0}_{m\times n} & \beta_{D}\left[ {\bf I}_{m} - {\rm diag}\, {\bf y} \right] {\bf A}_{D}({\bf x})-\gamma_{D}{\bf I}_{m} \\
		\end{array} \right]
		\left[\begin{array}{c}
			{\bf x} \\
			{\bf y} \\
		\end{array} \right].$}
		\label{system1}
	\end{equation}
	Although the model allows working with different parameter values on the networks $G_P$ and $G_D$, we will only consider analytically the case  $\beta_{P}=\beta_{D}=\beta$ and $\gamma_{P}=\gamma_{D} =\gamma$. By introducing the new variable ${\bf z}\coloneqq
	\left[\begin{array}{c}
		{\bf x} \\
		{\bf y} \\
	\end{array} \right] \in {\mathbb R}^{n+m}$,
		Eq. (\ref{system1}) can be expressed as
		\begin{equation}
		\resizebox{.88\hsize}{!}{$
			{\bf \dot{z}}=
			\left[ \begin{array}{cc} 
				\beta\left[ {\bf I}_{n} - {\rm diag}\, {\bf x} \right] {\bf A}_{P}({\bf y})-\gamma {\bf I}_{n} & {\bf 0}_{n\times m} \\
				{\bf 0}_{m\times n} & \beta\left[ {\bf I}_{m} - {\rm diag}\, {\bf y} \right] {\bf A}_{D}({\bf x})-\gamma{\bf I}_{m} \\
			\end{array} \right]
			{\bf z}.$}
			\label{system2}
		\end{equation}
		The variables ${\bf x}$ and ${\bf y}$ can be regained from ${\bf z}$ by means of the following two relations
		${\bf x}={\bf P}_{n}{\bf z}\coloneqq \left[{\bf I}_{n}|{\bf 0}_{n \times m}\right]{\bf z}$ and ${\bf y}={\bf Q}_{m}{\bf z}\coloneqq \left[{\bf 0}_{m \times n}|{\bf I}_{m}\right]{\bf z}$, so that we can write
		\begin{equation}
			\footnotesize
			\begin{split}
				{\bf \dot{z}}&=
				\left[ \begin{array}{cc} 
					\beta\left[ {\bf I}_{n} - {\rm diag} ( {\bf {\bf P}_{n}{\bf z}}) \right] {\bf A}_{P}({\bf z})-\gamma {\bf I}_{n} & {\bf 0}_{n\times m} \\
					{\bf 0}_{m\times n} & \beta\left[ {\bf I}_{m} - {\rm diag}({\bf {\bf Q}_{m}{\bf z}}) \right] {\bf A}_{D}({\bf z})-\gamma{\bf I}_{m} \\
				\end{array} \right]
				{\bf z}\\
				&=\beta
				\left[ \begin{array}{cc} 
					\left[ {\bf I}_{n} - {\rm diag} ( {\bf {\bf P}_{n}{\bf z}}) \right]  & {\bf 0}_{n\times m} \\
					{\bf 0}_{m\times n} & \left[ {\bf I}_{m} - {\rm diag}({\bf {\bf Q}_{m}{\bf z}}) \right] \\
				\end{array} \right]
				\cdot
				\left[ \begin{array}{cc} 
					{\bf A}_{P}({\bf z}) & {\bf 0}_{n\times m} \\
					{\bf 0}_{m\times n} & {\bf A}_{D}({\bf z}) \\
				\end{array} \right]{\bf z}-\gamma {\bf z}.
			\end{split}
			\label{system3}
		\end{equation}
		Now, let us define the two matrices 
		\begin{equation}
			{\bf G}({\bf z})\coloneqq
			\left[ \begin{array}{cc} 
				{\bf A}_{P}({\bf z}) & {\bf 0}_{n\times m} \\
				{\bf 0}_{m\times n} & {\bf A}_{D}({\bf z}) \\
			\end{array} \right]
			\label{matrixG}
		\end{equation}
		and
		\begin{equation}
			{\bf H}({\bf z})\coloneqq
			\beta
			\left[ {\bf I}_{n+m} - {\rm diag}\, {\bf z} \right]
			{\bf G}({\bf z})
			-\gamma{\bf I}_{n+m}.
			\label{matrixH}
		\end{equation}
		The self-adaptive SIS model is finally expressed by the ordinary differential equation
		\begin{equation}
			{\bf \dot{z}}={\bf H}({\bf z}){\bf z}.
			\label{ODE1}
		\end{equation}
		
		\begin{remark}
			It is worth focusing on the initial values of the adjacency matrices in Eq. (\ref{ApAd_continuos}).
			At time $t=0$, we set the initial attributes ${\bf x}(0)={\bf x}_{0}=p{\bf u}_{n}$ and ${\bf y}(0)={\bf y}_{0}=p{\bf u}_{m}$, where $p\in {\mathbb R}, p\in (0,1]$ represents the initial probability of being infected, uniformly distributed across nodes in network $G_P$ and nodes in network $G_D$.\footnote{The model allows for more general assumptions about initial probabilities. Here we assume that the initial probabilities are the same on the primary and dual networks. This allows us to obtain closed solutions in the case of some synthetic graphs. However, in numerical simulations nothing prevents the use of different values on the two networks  $G_{P}$ and $G_{D}$ or even different values on individual nodes. Similarly, in the numerical simulations, it is possible to implement different values of $\beta$ and $\gamma$ on the two networks $G_{P}$ and $G_{D}$.} 
			We denote by $q$ the initial probability of being susceptible: $q=1-p$. The initial values of the two matrices ${\bf A}_{P}(t)$ and ${\bf A}_{D}(t)$ in Eq. (\ref{ApAd_continuos}) are then
			\begin{equation}
				\left\{ 
				\begin{array}{l}
					{\bf A}_{P}(0)=p{\bf E} {\bf E}^{T}-p\, {\rm diag}({\bf E}{\bf u}_{m})=p {\bf B}_{P}\\
					\hfill \\
					{\bf A}_{D}(0)=p{\bf E}^{T} {\bf E}-p\, {\rm diag}({\bf E}{\bf u}_{n})=p {\bf B}_{D}\\
				\end{array}
				\right.
				\label{initialmatrices}
			\end{equation}
			where ${\bf B}_{P}$ and ${\bf B}_{D}$ are the original binary adjacency matrices of the two networks, containing the information about their topological structure.
		\end{remark}
		
		Since a real network $G_{P}$ is often originally edge-weighted, we now discuss how to incorporate the original weights in the process described earlier.
		
		Let ${\bf W}_{P}$ be the weighted adjacency matrix of the primary network $G_{P}$ obtained by ${\bf B}_{P}$ by adding weights to edges. Our aim is to re-modulate this matrix with the probabilities produced as the process evolves. To do this, we modify only the matrix ${\bf A}_{P}(t)$ in Eq. (\ref{ApAd_continuos}) as follows
		\begin{equation}
			{\bf A}_{P}(t)={\bf W}_{P} \odot \left( {\bf E}\, {\rm diag}\, {\bf y}(t) {\bf E}^{T}-{\rm diag}({\bf E}{\bf y}(t))\right)
		\end{equation}
		where $\odot$ is the Hadamard (i.e. element by element) product between the two matrices. Now the matrix ${\bf A}_{P}(t)$ can be understood as a weighted matrix that encompasses in itself both the original topological properties of the network (through ${\bf W}_{P}$) and the probabilities induced by the evolution of the process (through ${\bf y}$).
		It is worth stressing that we do not need to modify the expression of the matrix ${\bf A}_{D}(t)$. The construction of a line graph from an edge-weighted graph does not produce an edge-weighted dual graph, therefore the dual network is always structurally conceived as a binary network (that is ${\bf W}_{D}={\bf B}_{D}$), and weights on the edges in $G_D$ are only due to the effect of the evolving process. This raises no issue for the model since this network serves exclusively as an auxiliary network to trigger the process.
		
		\subsection{Reinforcement factor}
		\label{Reinforcement factor}
		We introduce here a parameter that allows a smooth transition from the standard SIS model in Eq. \eqref{SIS3} to the ASIS model in Eq. \eqref{continuos_eqs}. In particular, we can modulate the weights of the adjacency matrices updated at each time, weighing the level of \textit{self-adaptivity} that we want to apply. Let $e\in [0,1]$ and let us define
		\begin{align}
			{\bf x}_{e}(t) &= e{\bf x}(t) + (1-e)p{\bf u}_{n}, \label{xstrenght} \\
			{\bf y}_{e}(t) &= e{\bf y}(t) + (1-e)p{\bf u}_{m}. \label{ystrenght}
		\end{align}
		
		If we replace, in Eq. (\ref{system2}), ${\bf A}_{P}({\bf y}(t))$ by ${\bf A}_{P}({\bf y}_{e}(t))$ and ${\bf A}_{D}({\bf x}(t))$ by ${\bf A}_{D}({\bf x}_{e}(t))$, we reshape the weights of the adjacency matrices by quantities varying between the initial fixed probabilities of the model $p$ ($e=0$) and the actual probabilities of the nodes and edges at time $t$ ($e=1$).
		For $e=0$, we get two parallel and disentangled SIS processes on the primary and dual networks. For $e=1$, we get the \textit{fully} self-adaptive SIS model, described above. Thus, for any value $0<e<1$, we obtain a general model that includes the standard and the fully self-adaptive model as extremal and special cases. As a consequence, this general model yields perturbed solutions between the two extreme ones. We call the scalar parameter $e$ \textit{reinforcement factor} because it conveys the intensity of the mutual reinforcement between the primary and dual processes. This parameter, which is a measure of the level of self-adaptivity of the epidemic model, will be used to calibrate the reciprocal reinforcement action that typically takes place in social networks and discussed in the introduction. In the following, where not explicitly specified, by ASIS model we will mean the case $e=1$.
				
		\subsection{Application to synthetic graphs}
		
		In this section we present some analytical results about the steady state solutions of the ASIS model for some specific classes of binary networks (cycle, regular, complete and star networks). We report here only results concerning the cycle and the complete graph. We refer to Appendix \ref{appendixA} for the proofs of the theorems, and for the general case of regular graphs and star graphs. 
		
		Let $G_P$ be a cycle with $n$ nodes, $n$ edges and adjacency matrix ${\bf B}$. In this case, $G_D$ is also a cycle with $n$ nodes, $n$ edges and same adjacency matrix ${\bf B}$, then $G_P=G_D=C_n$, and the probabilities of all nodes in both graphs are identical. Since
		${\rm diag}\, {\bf x}(t)=x(t){\bf I}_{n}$,
		${\rm diag}\, {\bf y}(t)=y(t){\bf I}_{n}$, and
		${\rm diag}({\bf E}{\bf u}_{n})={\rm diag}({\bf E}^{T}{\bf u}_{n})=2{\bf I}_{n}$, Eq. (\ref{ApAd_continuos}) reduces to
		\begin{equation}
		\left\{ 
			\begin{array}{l}
				{\bf A}_{P}(t)=y(t){\bf B}\\
				\hfill \\
				{\bf A}_{D}(t)=x(t){\bf B}\\
			\end{array}
			\right.
			.
		\end{equation}
		Moreover, $x(t)=y(t)$ as we focus on the case in which the infectivity and recovery rates are the same for both $G_P$ and $G_D$.
		Therefore, for every node and edge in the cycle, Eq. (\ref{continuos_eqs}) reduces to a single equation:
		\begin{equation}
		\resizebox{.88\hsize}{!}{$
			\begin{split}
				\dot{x}(t)=&\ \beta\left[1-x(t) \right]\sum_{h=1}^{n} ({\bf A}_{P})_{ih}\, x_{h}(t)-\gamma x(t)\\
				=&\ \beta\left[1-x(t) \right] 2 y(t) x(t)-\gamma x(t)
				= -2\beta x^{3}(t)+2\beta x^{2}(t)-\gamma x(t).
				\label{eqs_cycles}
			\end{split}$}
		\end{equation}
		
		The steady states of the nonlinear mapping in Eq. (\ref{eqs_cycles}) are characterized by the following
		
		\begin{theorem}
			\label{theorem1_cycle}
			The stable equilibrium points of the ASIS model on the cycle $C_n$ and its dual network, described by Eq. (\ref{eqs_cycles}), are given by
			\begin{equation}
				\label{equilibrium_cycle_graph}
				\left\{ 
				\begin{array}{lll}
					x^{\star}=0 & {\rm if} & \mathcal{R}< \tau_{c}\\
					\hfill \\
					x^{\star}=\frac{1}{2}\left( 1+ \sqrt{1-\frac{2}{\mathcal{R}}} \right) & {\rm if} & \mathcal{R}> \tau_{c}
				\end{array}
				\right.
			\end{equation}
			where
			\begin{equation}
				\label{threshold_cycle_graph}
				\tau_{c}=
				\left\{ 
				\begin{array}{lll}
					\frac{1}{2p(1-p)} & {\rm if} & 0<p<\frac{1}{2}\\
					2 & {\rm if} & \frac{1}{2}\leq p<1
				\end{array}
				\right.
			\end{equation}
			is the threshold of the epidemic dynamics on cycles.
		\end{theorem}
		
	Notice that the threshold $\tau_{c}(p)$ is a nonincreasing function of $0<p<1$; in particular, it is always greater than or equal to $2$ and such that $\tau_{c}(p)\to +\infty$ when $p\to 0$. Moreover, when ${\mathcal R}\to \tau_{c}^{+}$ and $0<p<\frac{1}{2}$, it is easy to show by simple calculations that $x^{\star}\to 1-p$, i.e., the asymptotic probability of the endemic state is equal to the initial probability that a node is susceptible. If ${\mathcal R}\to \tau_{c}^{+}$ and $\frac{1}{2}<p<1$, then $x^{\star}\to \frac{1}{2}$, i.e., the endemic asymptotic state stabilizes on an equal distribution of infected and susceptible cases.\footnote{The threshold of the standard SIS model and its asymptotic endemic state on the cycle, with the same initial conditions, are equal to $\tau=\frac{1}{2p}$ and $x^{\star}=1-\frac{1}{2p\mathcal{R}}$.} For $\mathcal{R}=\tau_{c}(p)$, $x(t)=p,\ \forall t$. Observe that, when $\gamma =0$, we are in the case of a self-adaptive SI model; being $\beta>0$ and $p>0$, the only stable asymptotic solution reduces to $x^{\star}=1$, as in any SI model.
	The result obtained for the cycle $C_n$ can be generalized to regular graphs $K_n^d$ with $n$ vertices of degree $k_i=d,\  \forall i=1,...,n$. We report here the result for the complete graph $K_n$ and we refer the reader to appendix \ref{appendixA} for the detailed proof of the case $K_n^d$.
	If $G_P$ is the complete graph $K_n$ of $n$ nodes, its dual $G_D$ is a regular graph of degree $2(n-2)$.\footnote{$K_{n}$ has $n$ vertices, $m=\frac{1}{2}n(n-1)$ edges, and degree $d=n-1$. The line graph of $K_n$ has $m=\frac{1}{2}n(n-1)$ vertices and $\frac{1}{2}n(n-1)(n-2)$ edges}  The following theorem provides the values of the steady states and the expression of the corresponding threshold on the complete graph.
		\begin{theorem}
			\label{theorem2_complete}
			The stable equilibrium points of the ASIS model with reproductive ratio $\mathcal{R}$ on the complete graph $K_{n}$  
			are given by
			\begin{equation}
				\resizebox{.88\hsize}{!}{$
				\label{solutions_complete_graph}
				\left\{ 
				\begin{array}{lll}
					x^{\star}=0 & {\rm if} & \mathcal{R}<\tau_{\rm compl}\\
					\hfill \\
					x^{\star}=\frac{1}{2}\left(1-\frac{n-3}{2(n-1)(n-2)\mathcal{R}}+ \frac{\sqrt{\xi}}{2(n-1)(n-2)\mathcal{R}}\right) & {\rm if} & \mathcal{R}>\tau_{\rm compl}
				\end{array}
				\right.$}
			\end{equation}
			where \scalebox{0.88}{$\xi=\left[ (n-3)-2(n-1)(n-2)\mathcal{R} \right]^{2}-8(n-1)^2(n-2)\mathcal{R}$} and
			\begin{equation}
				\resizebox{.88\hsize}{!}{$
				\label{threshold_complete_graph}
				\tau_{\rm compl}=
				\left\{ 
				\begin{array}{lll}
					\frac{(n-1)+(n-3)p}{2(n-1)(n-2)}\cdot \frac{1}{p(1-p)} & {\rm if} & 0<p<\frac{1}{1+\sqrt{\frac{2(n-2)}{n-1}}}\\
					\left[ \frac{1}{\sqrt{n-1}}+\frac{1}{\sqrt{2(n-2)}}\right]^{2} & {\rm if} & \frac{1}{1+\sqrt{\frac{2(n-2)}{n-1}}}\leq p<1
				\end{array}
				\right.$}
			\end{equation}
			is the epidemic threshold on complete graphs.
		\end{theorem}
		In Fig. \ref{fig1new}, panels (a) and (b), we illustrate the evolution of the ASIS model above and below the threshold $\tau_{c}$, compared with the standard SIS model for a cycle with $n=5$ and $p=0.2$. In both panels, the threshold is $\tau_{c}(p)=3.125$. In panel (a), ${\mathcal R}=5$, and, in panel (b), ${\mathcal R}=1.333$. The stable asymptotic solution above the threshold, in panel (a), for the ASIS model is $x^{\star}=\frac{1}{2}\left( 1+ \sqrt{1-\frac{2}{\mathcal{R}}} \right)=0.8872983$, while for the standard SIS model is $x^{\star}=1-\frac{1}{2p{\mathcal R}}=0.5$. In Fig. \ref{fig1new}, panel (c), we represent the same evolution for the complete graph $K_n$ with $n=6$ and $p=\frac{1}{6}$, and with threshold $\tau_{\rm compl}(p)=0.990$. Then, for $\mathcal{R}=2$, the endemic state in Eq. (\ref{solutions_complete_graph}) reduces to $x^{\star}=0.892$.
		\begin{figure}[H]
			\centering
			\subfloat[]{\includegraphics[width=0.40\textwidth]{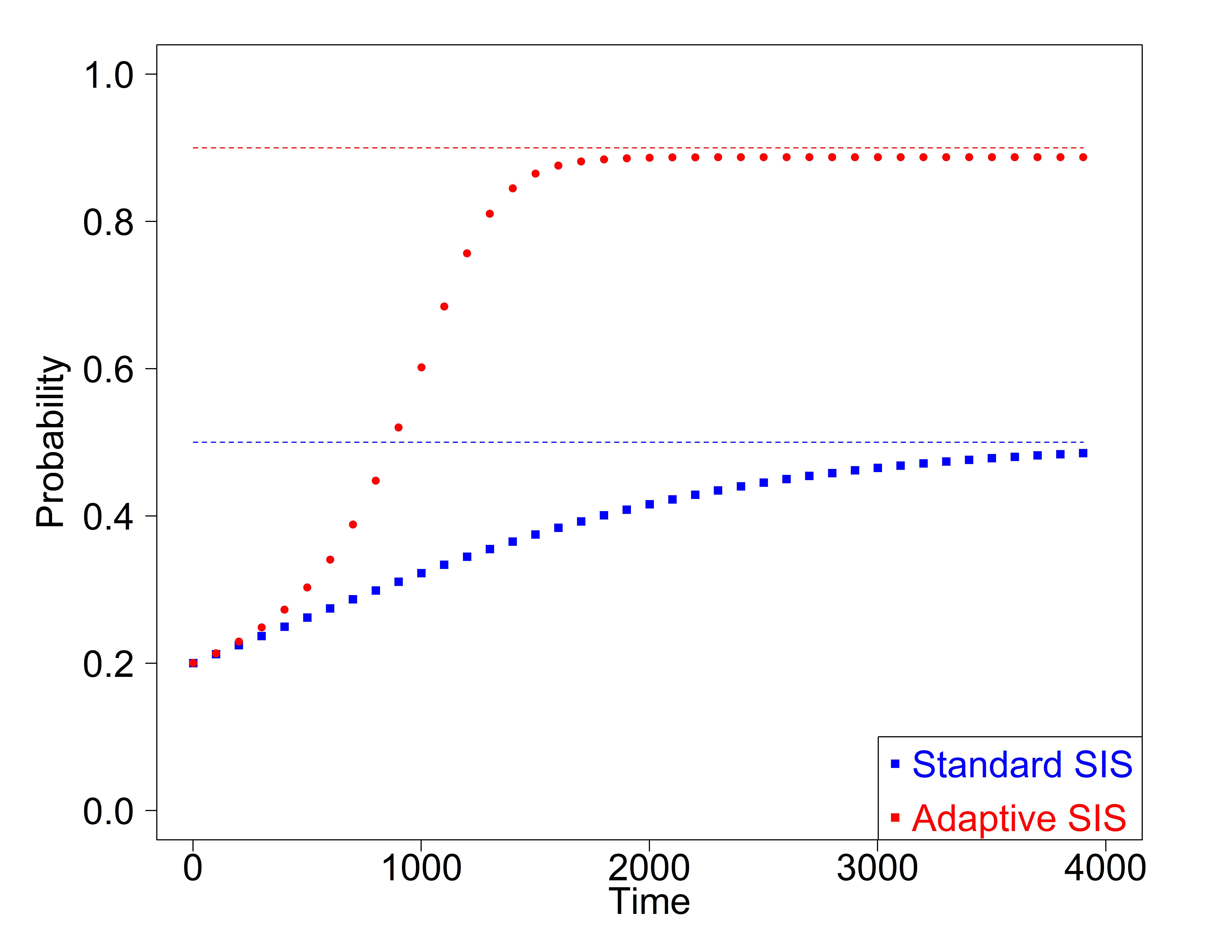}}
			\\ \vspace{-2mm}
			\subfloat[]{\includegraphics[width=0.40\textwidth]{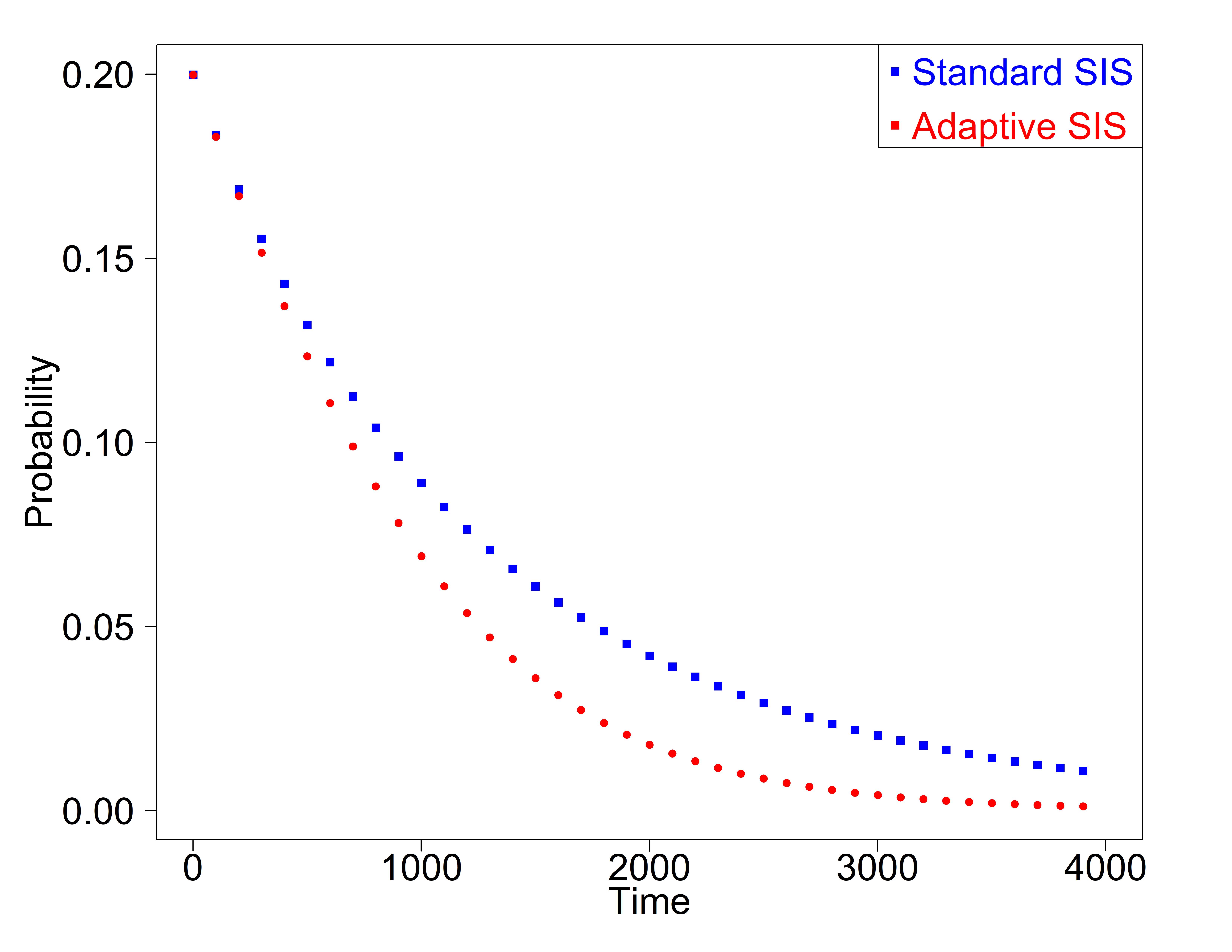}}
			\\ \vspace{-2mm}
			\subfloat[]{\includegraphics[width=0.40\textwidth]{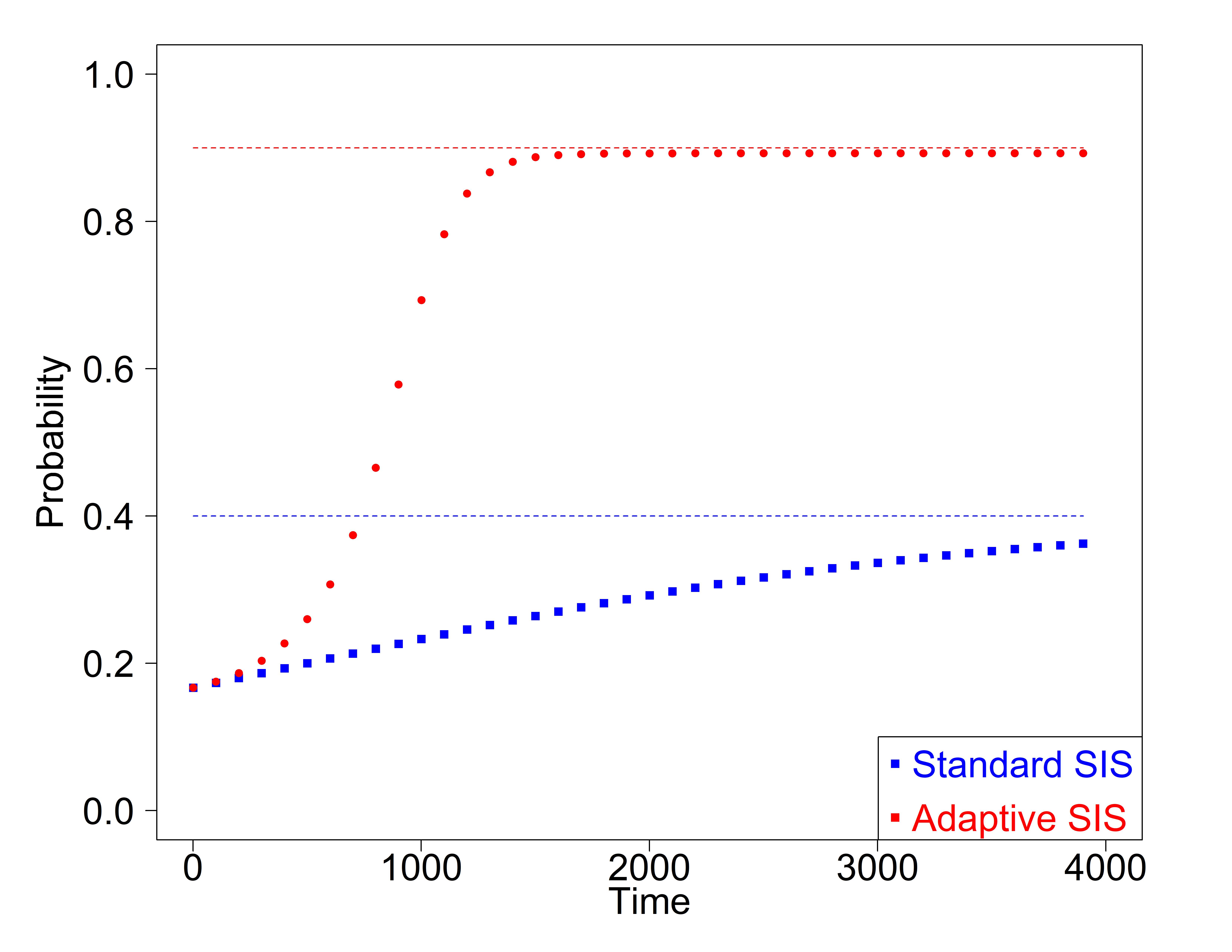}}
			\caption{Probability $x(t)$ for the self-adaptive SIS model (in red circle points) and for the standard SIS model (in blue square points) for (a) cycle graph with $n=5$, $p=1/5$, ${\mathcal R}=5$ ($\beta=0.005$ and $\gamma=0.001$); (b) cycle graph with $n=5$, $p=1/5$, ${\mathcal R}=1.333$ ($\beta=0.002$ and $\gamma=0.0015$); (c) complete graph with $n=6$, $p=1/6$, ${\mathcal R}=2$ ($\beta=0.002$ and $\gamma=0.001$).}
			\label{fig1new} 
		\end{figure}
				
\section{General steady states analysis}
\label{General steady states analysis}
		Our goal now is to study the general steady states of the ASIS model described by Eq. (\ref{ODE1}), 
		${\bf \dot{z}}(t)={\bf H}({\bf z}(t)){\bf z}(t)$, where ${\bf G}\left({\bf z}(t)\right)$ and ${\bf H}\left( \mathbf{z}(t)\right)$ are defined in Eq. (\ref{matrixG}) and (\ref{matrixH}).
		
		\subsection{Nonlinear eigenproblem}
		\label{Nonlineareigenproblem}
		First, we observe that, for all $t$, $\left[ {\bf I} - {\rm diag}\, {\bf z}(t) \right]{\bf G}({\bf z}(t))$ is the product of two symmetric matrices, the first of which is a diagonal matrix with nonnegative entries.
		Therefore, although the product is not a necessarily symmetric matrix, it has only real eigenvalues and its eigenvectors can always be chosen with real components.\footnote{Recall that for any symmetric matrix ${\bf G}$, ${\bf D}{\bf G}$ has the same eigenvalues as ${\bf D}^{1/2}{\bf G}{\bf D}^{1/2}$ for any diagonal matrix ${\bf D}$ with positive entries.}
		
		The identification of the endemic steady states of Eq. (\ref{ODE1}) can be interpreted as a nonlinear eigenproblem. In fact, by setting ${\bf H}({\bf z}^{\star}){\bf z}^{\star}= {\bf 0}$,
	we get\footnote{From now on, we set ${\bf I}_{n+m}={\bf I}$}:
	\begin{equation}
		{\bf z}^{\star}=
		{\mathcal R} \left[ {\bf I} - {\rm diag}\, {\bf z}^{\star} \right] {\bf G}\left({\bf z}^{\star}\right){\bf z}^{\star}.
		\label{eigenvalue_problem}
	\end{equation}
	Therefore, the vector representing the steady state is an eigenvector of the non-symmetric matrix ${\bf M}({\bf z}^{\star})\coloneqq {\mathcal R} \left[ {\bf I} - {\rm diag}\, {\bf z}^{\star} \right] {\bf G}\left({\bf z}^{\star}\right)$. The problem belongs to a peculiar class of eigenvalue problems in which the nonlinearity is produced by the matrix itself depending on and containing the eigenvector being pursued. Although problems of this type have received less attention in the literature than nonlinear eigenproblems, where the nonlinearity is only related to the eigenvalues, some iterative methods for obtaining the dominant eigenvector have been proposed. However, they rely heavily on specific assumptions required on the matrix ${\bf M}({\bf z})$ (see \citet{Meyer1997} and \citet{Jarlebring2014}).
	
	To the best of our knowledge, no effective algorithm has been proposed to find the dominant eigenvector of non-symmetric nonlinear problems like the one in Eq. (\ref{eigenvalue_problem}). 
	The approach we propose to fill this gap is inspired by the two above-mentioned contributions existing in the literature. In particular, \citet{Meyer1997} proposes a nonlinear eigenvector algorithm to show the global convergence for problems of the form ${\bf R}({\bf z}){\bf z}=\lambda {\bf S}({\bf z}){\bf z}$ where ${\bf R}({\bf z})$ and ${\bf S}({\bf z})$ are real symmetric block-diagonal matrices. The basic idea is to start with some arbitrary vector ${\bf z}_{0}$, and fixed matrices ${\bf R}({\bf z}_{0})$ and ${\bf S}({\bf z}_{0})$, and solve an ordinary generalized eigenproblem to find out the eigenvector ${\bf z}_{1}$ corresponding to the largest eigenvalue. Then the matrices are updated to ${\bf R}({\bf z}_{1})$ and ${\bf S}({\bf z}_{1})$, treated as fixed, and another eigenproblem is solved, and so on until the procedure converges. 
	We stress that the entire procedure is applied to matrices ${\bf R}({\bf z})$ and ${\bf S}({\bf z})$ that are symmetric, while, in our case,
	the matrix ${\bf M}({\bf z^*})$ is a real block-diagonal matrix but it is not symmetric.
	 
	An alternative iterative method has been proposed in \citet{Jarlebring2014} for scale invariant matrices, that is matrices ${\bf M}({\bf z})$ such that ${\bf M}(\alpha{\bf z})={\bf M}({\bf z})$, $\forall \alpha \in {\mathbb R}$. This inverse algorithm is based on the Jacobian matrix ${\bf J}({\bf z})$ of the problem and the iteration takes the form
	\begin{equation}
		{\bf z}_{k+1}=\frac{\left( {\bf J}-\sigma {\bf I} \right)^{-1}{\bf z}_{k}}{||\left( {\bf J}-\sigma {\bf I} \right)^{-1}{\bf z}_{k}||}
		\label{iterative1}
	\end{equation}
	where $\sigma\in {\mathbb R}$ is called \textit{shift} and controls to which pair of eigenvalue and eigenvector the iteration converges. An ${\bf M}$-version is also discussed, in which the Jacobian matrix is replaced by the matrix ${\bf M}({\bf z})$, at the cost of losing some convergence properties.
	Even neglecting that we have no explicit expression of the Jacobian matrix, however, again, this approach relies dramatically on the invariance property of the matrix ${\bf M}({\bf z})$, which is not the case of the matrix involved in our model.
	
	We will retain the basic idea of the algorithm proposed by \citet{Meyer1997} in the discretization of the ASIS problem that will be proposed shortly, and we will use the ${\bf M}$-version of Eq. \eqref{iterative1} to numerically compute the dominant eigenvector. In fact, we aim at providing an appropriate discretization of the model which can be interpreted as an algorithm for constructing a non-normalized version of the dominant eigenvectors.
	
	Let us first introduce the rescaled time variable $t'=\gamma t$ so that Eq. (\ref{ODE1}) becomes
	\begin{equation}
		\dot{\bf z}={\mathcal R}\left[ {\bf I}-{\rm diag}\, {\bf z} \right] {\bf G}({\bf z}){\bf z}-{\bf z}
	\end{equation}
	equivalent to
	\begin{equation}
		\left\{ 
		\begin{array}{l}
			\dot{\bf x} ={\mathcal R}\left[ {\bf I}_{n} - {\rm diag}\, {\bf x} \right] {\bf A}_{P}({\bf y})\, {\bf x} -{\bf x}\\
			\hfill \\
			\dot{\bf y} ={\mathcal R}\left[ {\bf I}_{m} - {\rm diag}\, {\bf y} \right] {\bf A}_{D}({\bf x})\, {\bf y} -{\bf y}\\
		\end{array}
		\right.
		.
		\label{continuous_eqs_rescaled}
	\end{equation}
	Let us now consider the following forward discretization of the two processes in Eq. (\ref{continuous_eqs_rescaled}). Let $\{t_{k}\}, \ k\in \mathbb{N}$, such that the step size is assumed, for the sake of simplicity, to be constant and equal to $1$: $t_{k+1}-t_{k}=1$. Let us set ${\bf z}_{k}={\bf z}(t_{k})$ and similar expressions for ${\bf x}$ and ${\bf y}$.
	Eq. (\ref{continuous_eqs_rescaled}) transforms into
	the following set of discrete-time Markovian
	equations:
	\begin{equation}
		\left\{ 
		\begin{array}{l}
			{\bf x}_{k+1}={\mathcal R}\left[ {\bf I}_{n} - {\rm diag}\, {\bf x}_{k} \right]{\bf A}_{P}({\bf y}_{k})\, {\bf x}_{k}\\
			\hfill \\
			{\bf y}_{k+1}={\mathcal R} \left[ {\bf I}_{m} - {\rm diag}\, {\bf y}_{k} \right] {\bf A}_{D}({\bf x}_{k})\, {\bf y}_{k}\\
		\end{array}
		\right.
		\label{discrete_solution}
	\end{equation}
	which iteratively update matrices ${\bf A}_{P}({\bf y}_{k})$ and ${\bf A}_{D}({\bf x}_{k})$ and compute the new vectors ${\bf x}_{k+1}$ and ${\bf y}_{k+1}$. The process ends when a stopping tolerance $\varepsilon$ is reached. The discretized ASIS model is illustrated in the Algorithm \ref{algorithm}.
	
	\begin{algorithm}[!ht]
		\DontPrintSemicolon
		\KwIn{Incidence matrix $\bf E$; initial probabilities ${\bf x}_{0}$ and ${\bf y}_{0}$; stopping tolerance $\varepsilon$}
		\KwOut{Steady state probabilities ${\bf x}^{\star}$ and ${\bf y}^{\star}$}
		${\bf x}_{0}=p{\bf u}_{n}$ and ${\bf y}_{0}=p{\bf u}_{m}$\;
		\Repeat{
			$||{\bf x}_{k+1}-{\bf x}_{k}||/||{\bf x}_{k}||+||{\bf y}_{k+1}-{\bf y}_{k}||/||{\bf y}_{k}||<\varepsilon$;
		}{
			${\bf A}_{P}({\bf y}_{k}) \gets {\bf E}\, {\rm diag}({\bf y}_{k}) {\bf E}^{T}-{\rm diag}({\bf E}{\bf y}_{k})$ \;
			${\bf A}_{D}({\bf x}_{k}) \gets {\bf E}^{T}\, {\rm diag}({\bf x}_{k}) {\bf E}-{\rm diag}({\bf E}^{T}{\bf x}_{k})$ \;
			${\bf x}_{k+1} \gets {\mathcal R}  \left[ {\bf I}_{n} - {\rm diag}\, {\bf x}_{k} \right] {\bf A}_{P}({\bf y}_{k})\, {\bf x}_{k}$ \;
			${\bf y}_{k+1} \gets {\mathcal R}  \left[ {\bf I}_{m} - {\rm diag}\, {\bf y}_{k} \right] {\bf A}_{D}({\bf x}_{k})\, {\bf y}_{k}$ \;
		}
		\Return{${\bf x}^{\star}$, ${\bf y}^{\star}$}\;
		\caption{{\sc Self-Adaptive SIS Model}}
		\label{algorithm}
	\end{algorithm}
	
	This algorithm, although modified, traces the idea of the powers method and particularly that in the nonlinear case discussed above.
	Let us observe that, being $\left( {\bf M}({\bf z})-\sigma {\bf I}\right)^{-1}{\bf z}=(1-\sigma)^{-1}{\bf z}$, for $\sigma\in {\mathbb R}$, then matrix $\left( {\bf M}({\bf z})-\sigma {\bf I}\right)^{-1}$ has the same eigenvectors as ${\bf M}({\bf z})$. Therefore, we can devise an inverse iteration method which is similar to the one in Eq. (\ref{iterative1}), provided that we keep as the argument of the nonlinear matrix ${\bf M}$ the non-normalized version of the vector ${\bf z}$.
	It is worth noting that the largest eigenvalue of the block matrix ${\bf M}({\bf z}^{\star})$ is $\lambda_{\bf M}^{(1)}=1$ with multiplicity $2$, as it represents the adjacency matrix of a network with two disconnected components, the network $G_P$ and its line graph $G_D$. The corresponding dominant eigenvectors are ${\bf z}_{1}^*=\left[{\bf x}^{\star},{\bf 0}_{m}\right]^T$ and ${\bf z}_{2}^*=\left[{\bf 0}_{n},{\bf y}^{\star}\right]^T$ and the corresponding normalized eigenvectors are then $\psi_{\bf M}^{(1)}={\bf z}_1^{\star}/||{\bf z}_1^{\star}||$ and $\psi_{\bf M}^{(2)}={\bf z}_2^{\star}/||{\bf z}_2^{\star}||$.
	
	\vspace{3cm}
	
	\subsection{Stability of the general endemic and disease-free steady states}
	\label{Stability of the general endemic and disease-free steady states}
	
	We now turn to the problem of the stability of equilibrium solutions. We present first two preliminary results about the matrix ${\bf G}({\bf z})$ and the Jacobian matrix ${\bf J}({\bf z})$ of the general problem in Eq. \eqref{ODE1}.
	
	\begin{lemma}
		The linear operator ${\bf G}({\bf z}): {\mathbb R}^{n+m}\to {\mathbb R}^{n+m} $ is a homogeneous operator of degree $1$
		\begin{equation}
			{\bf G}(\alpha {\bf z})=\alpha {\bf G}({\bf z}), \quad \forall \alpha \in {\mathbb R}.
		\end{equation}
	\end{lemma}
	\begin{proof}
		By definitions (\ref{ApAd_continuos}), ${\bf A}_{P}(\alpha {\bf y}) =\alpha {\bf A}_{P}( 
		{\bf y})$ and ${\bf A}_{D}(\alpha {\bf x}) =\alpha {\bf A}_{D}({\bf x}), \forall \alpha \in {\mathbb R}$ . 
	\end{proof}
	\begin{lemma}
		The Jacobian matrix ${\bf J}({\bf z})$ of the system in Eq. (\ref{ODE1}) satisfies the following relation
		\begin{equation}
			{\bf J}({\bf z}){\bf z}=\left[ \beta \left( 2\mathbf{I}-3\, {\rm diag}\, {\bf z} \right) {\bf G}\left( \mathbf{z}\right) -\gamma \mathbf{I}\right] {\bf z}.
		\end{equation}
		\label{lemma2}
	\end{lemma}
	\begin{proof}
		By definition,
		\begin{equation*}
				\resizebox{1.00\hsize}{!}{$
			\begin{split}
				{\bf J}({\bf z}){\bf z}&=
				\lim_{\varepsilon \to 0}
				\frac{{\bf H}({\bf z+\varepsilon \bf z})({\bf z+\varepsilon \bf z})-{\bf H}({\bf z})({\bf z})}{\varepsilon}\\
				&= \lim_{\varepsilon \to 0}
				\frac{\beta \left[ \left( {\bf I}-\, {\rm diag} \left( {\bf z+\varepsilon z}\right) \right)  {\bf G}\left( {\bf z+\varepsilon z}
					\right) -\gamma {\bf I} \right] \left( {\bf z+\varepsilon z}\right) -\beta \left[ \left( {\bf I}-\, {\rm diag} \left( {\bf z} \right) \right) {\bf G}\left( {\bf z}\right) 
					-\gamma {\bf I}\right] {\bf z}}{\varepsilon } \\
				&= \lim_{\varepsilon \to 0}
				\frac{\beta (1+\varepsilon)^{2} \left[ \left( {\bf I}-(1+\varepsilon)\, {\rm diag} \left( {\bf z}\right) \right)  {\bf G}({\bf z}) {\bf z} \right] -\beta \left[ \left( {\bf I}-{\rm diag} ( {\bf z}) \right)  {\bf G}({\bf z}){\bf z}\right]-\gamma\varepsilon {\bf z}}{\varepsilon } \\
				&= \lim_{\varepsilon \to 0}
				\frac{\beta \left[  \left( (1+\varepsilon)^{2} -1\right){\bf I}-\left( (1+\varepsilon)^{3}-1\right) \, {\rm diag} ({\bf z}) \right]  {\bf G}({\bf z}) {\bf z}  -\gamma\varepsilon {\bf z}}{\varepsilon } \\
				&= \lim_{\varepsilon \to 0}
				\frac{\beta \left[  \left( 2\varepsilon+\varepsilon)^{2}\right){\bf I}-\left( 3\varepsilon+3\varepsilon^{2}+\varepsilon^{3}\right) \, {\rm diag} ({\bf z}) \right]  {\bf G}({\bf z}) {\bf z}  -\gamma\varepsilon {\bf z}}{\varepsilon } \\
				&=\left[ \beta \left( 2{\bf I}-3\, {\rm diag}\, {\bf z} \right) {\bf G}( {\bf z}) -\gamma {\bf I}\right] {\bf z}    
			\end{split}$}
		\end{equation*}
		\end{proof}
	It is important to note that we cannot provide an explicit expression of the Jacobian matrix ${\bf J}({\bf z})$. However, through Lemma (\ref{lemma2}), we are able to describe the action of this matrix, evaluated in a general vector ${\bf z}$, on the same vector ${\bf z}$.
	We now turn to the main Theorem.

	\begin{theorem}
		Given an undirected, weighted and connected network, a non-null equilibrium solution $\bf z^{\star}$ of Eq. (\ref{ODE1}) represents a stable endemic steady state for the ASIS model if $z^{\star}_i\geq 1-\frac{\sqrt 2}{2}$, $\forall i=1,...,n+m$.
		\label{theorem3}
	\end{theorem}
	
	\begin{proof}
		Let us show that an endemic stable steady state exists, by a constructive proof. The steady state of the general problem in Eq. (\ref{ODE1}) is defined by the nonlinear eigenvalue problem ${\bf H}({\bf z}){\bf z}={\bf 0}$, equivalent to $\beta\,( {\bf{I}-\rm diag}\,{\bf z}){\bf G}({\bf z}){\bf z}=\gamma {\bf z}$. Then, the steady state has to satisfy the equality
		\begin{equation}
			\beta{\bf G}({\bf z}){\bf z}=\gamma ( {\bf{I}-\rm diag}\,{\bf z})^{-1}{\bf z} 
		\end{equation} 
		for $z_{i}\neq 1, \forall i=1,\dots, n+m$. 
		Conversely, in the steady state, by Lemma (\ref{lemma2}), each component $i$ of the vector ${\bf J}({\bf z}){\bf z}$ satisfies: 
		\begin{equation}
			\begin{split}
				({\bf J}({\bf z}){\bf z})_i&=
				(\left[ \beta \left( 2{\bf I}-3\, {\rm diag}\, {\bf z} \right) {\bf G}( {\bf z})-\gamma {\bf I}\right] {\bf z})_i \\
				&= (\left( 2{\bf I}-3\, {\rm diag}\, {\bf z} \right) \beta {\bf G}( {\bf z}){\bf z})_i-\gamma z_i \\
				& = \gamma (\left( 2{\bf I}-3\, {\rm diag}\, {\bf z} \right)({\bf{I}-\rm diag}({\bf z}))^{-1}{\bf z})_i-\gamma z_i \\
				& = \gamma \left( \frac{2-3z_i}{1-z_i}\right){z}_i-\gamma z_i = \gamma \left( \frac{1-2z_i}{1-z_i}\right){z}_i. \\
			\end{split}
		\end{equation}
		We analyze the behavior around a stationary solution ${\bf z}^{\star}$. Let us define the error $\Delta{\bf z}(t)={\bf z}(t)-{\bf z}^{\star}$. By linearizing around ${\bf z}^{\star}$ (see \citet{Medio2001}), we get
		\begin{equation}
			\dot{\Delta}{\bf z}(t)={\bf J}({\bf z}^{\star}){\Delta}{\bf z}(t)={\bf J}({\bf z}^{\star}){\bf z}(t)+\gamma {\bf c}
		\end{equation}
		where ${c}_i=-\frac{1}{\gamma}\left({\bf J}({\bf z}^{\star}){\bf z}^{\star}\right)_i=\left( \frac{2z^{\star}_i-1}{1-z^{\star}_i}\right){z}^{\star}_i$. Now, since $\dot{\Delta}{\bf z}(t)=\dot{\bf z}(t)$, we have
		\begin{equation}
			\dot{\bf z}(t)={\bf J}({\bf z}^{\star}){\bf z}(t)+\gamma{\bf c}.
		\end{equation}
		In general, we do not have the explicit expression of the Jacobian matrix ${\bf J}({\bf z}^{\star})$ but, for ${\bf z}(t) \to {\bf z}^{\star}$, in the neighborhood of ${\bf z}^{\star}$:
		\begin{equation}
			\dot{\bf z}(t)\sim {\bf J}({\bf z}){\bf z}(t)+\gamma{\bf c}
		\end{equation}
		which is approximated by the $n+m$ nonlinear differential equations $\dot{z}_i= \gamma \left( \frac{1-2z_i}{1-z_i}\right){z}_i+\gamma {c}_i$, that is
		\begin{equation}
			\dot{z}_i= \gamma \left[ \frac{{c}_i+(1-{c}_i){ z}_i-2{z}_i^{2}}{1-{z}_i} \right] .
			\label{stability_equation1}
		\end{equation}
		For the sake of simplicity, we set $z_i=z$ and $c_i=c$.
		Eq. (\ref{stability_equation1}) is equivalent to
		\begin{equation}
			\int \frac{1-z}{c+(1-c)z-2z^2}dz=\gamma t+K, \quad {K}\in{\mathbb R}.
			\label{stability_equation2}
		\end{equation}
		For any $0<z^{\star}<1$, the denominator $c+(1-c)z-2z^2$ has two real distinct roots
		\begin{equation}
			\label{roots}
			\left\{ 
			\begin{array}{l}
				\tilde{z}_{1}=-\frac{1}{4}\left[ (c-1)+\sqrt{c^2+6c+1}\right ] =\frac{2z^{\star}-1}{2z^{\star}-2}\\
				\tilde{z}_{2}=-\frac{1}{4}\left[ (c-1)-\sqrt{c^2+6c+1}\right ] =z^{\star}\\
			\end{array}
			\right.
		\end{equation}
		then, by computing the integral: 
		
		\begin{equation}
			\begin{split}
			{\cal I}
			&=\int \frac{1-z}{c+(1-c)z-2z^2}dz
			=\frac{1}{2}\int \frac{z-1}{(z-\tilde{z}_{1})(z-\tilde{z}_{2})}dz \\
			&=\frac{1-z^{\star}}{|2z^{\star 2}-4z^{\star}+1|} \log \frac{|z-\tilde{z}_{1}|^{1+\tilde{z}_{1}}}{|z-\tilde{z}_{2}|^{1+\tilde{z}_{2}}}.
			\label{integral1}
			\end{split}
		\end{equation}
		Therefore, Eq. (\ref{stability_equation2}) becomes
		\begin{equation}
			\frac{\left| z-\frac{2z^{\star}-1}{2z^{\star}-2}\right|^{\frac{4z^{\star}-3}{2z^{\star}-2}} }{\left| z-z^{\star} \right|^{1+z^{\star}}}=\kappa\cdot \exp\left( \frac{|2z^{\star 2}-4z^{\star}+1|}{1-z^{\star}} \gamma t\right), \quad \kappa \in {\mathbb R}^{+}.
			\label{stability_solution1}
		\end{equation}
		Moreover, since $\Delta z=z-z^{\star}$, Eq. (\ref{stability_solution1}) can be rewritten in terms of $\Delta z$ as:
		\begin{equation}
			\frac{\left| \Delta z+\frac{2z^{\star 2}-4z^{\star}+1}{2z^{\star}-2}\right|^{\frac{4z^{\star}-3}{2z^{\star}-2}} }{\left| \Delta z \right|^{1+z^{\star}}}=\kappa\cdot \exp\left( \frac{2z^{\star 2}-4z^{\star}+1}{z^{\star}-1} \gamma t\right).
			\label{stability_solution2}
		\end{equation}
		Let us call $\alpha=\frac{2z^{\star 2}-4z^{\star}+1}{z^{\star}-1} $.  Thus, we have:
		\begin{equation}
			\frac{\left| \Delta z+\frac{\alpha}{2}\right|^{\frac{4z^{\star}-3}{2z^{\star}-2}} }{\left| \Delta z \right|^{1+z^{\star}}}=\kappa\cdot e^{\alpha \gamma t}.
			\label{stability_solution3}
		\end{equation}
		Let us study the two cases, $\alpha>0$ and $\alpha<0$, separately:
		\begin{itemize}
			\item $\alpha >0$, that is $1-\frac{\sqrt 2}{2} < z^{\star}< 1$.
			It is useful to further distinguish, in Eq. \eqref{stability_solution3}, two cases according to the sign of the exponent $\frac{4z^{\star}-3}{2z^{\star}-2}$:
			\begin{equation}
				\label{stability_solution4}
				\left\{ 
				\begin{array}{l}
					\frac{\left| \Delta z+\frac{\alpha}{2}\right|^{\left| \frac{4z^{\star}-3}{2z^{\star}-2}\right|} }{\left| \Delta z \right|^{1+z^{\star}}}=\kappa\cdot e^{\alpha \gamma t} \ \ \qquad {\rm for} \quad 1-\frac{\sqrt 2}{2}<z^{\star}\leq\frac{3}{4}\\
					\frac{\left| \Delta z+\frac{\alpha}{2}\right|^{-\left| \frac{4z^{\star}-3}{2z^{\star}-2}\right|} }{\left| \Delta z \right|^{1+z^{\star}}}=\kappa\cdot e^{\alpha \gamma t} \qquad {\rm for} \quad \frac{3}{4}<z^{\star}<1
				\end{array}
				\right.
			\end{equation}
			
			In both cases, if $t \to +\infty$, then $e^{\alpha \gamma t} \to +\infty$.
			 
			In the first case, if it were $\Delta z\to +\infty$, then the left-hand side would be of order 
			$|\Delta z|^{\frac{4z^{\star}-3}{2z^{\star}-2}-(1+z^{\star})}=|\Delta z|^{\frac{2z^{\star 2}-4z^{\star}+1}{2(1-z^{\star})}}=|\Delta z|^{-\frac{\alpha}{2}}$. The exponent would be negative and the left-hand side would go to $0$, in contrast to the right-hand side going to $+\infty$. In the second case, the exponent of the term in the numerator is already negative and still the error can go neither to a finite nonzero value nor to infinity. Then the only possibility is that $\Delta z \to 0$. \\
			In particular, in this second case, namely for $\frac{3}{4}<z^{\star}<1$, we are able to compute explicitly the Lyapunov exponent. Indeed, when $|\Delta z|$ vanishes, we have
			\begin{equation}
				\label{stability_solution5}
				\frac{\left|\frac{\alpha}{2}\right|^{-\left| \frac{4z^{\star}-3}{2z^{\star}-2}\right|} }{\left| \Delta z \right|^{1+z^{\star}}}\sim \kappa\cdot e^{\alpha \gamma t}.
			\end{equation}
			For $t=0$, we have $\Delta z(0)= \Delta z_{0}$ and, by \eqref{stability_solution5},
			$
			\kappa\approx\frac{\left|\frac{\alpha}{2}\right|^{-\left| \frac{4z^{\star}-3}{2z^{\star}-2}\right|} }{\left| \Delta z_{0}\right|^{1+z^{\star}}}
			$
			so that
			\begin{equation}
				\label{stability_solution6}
				|\Delta z|^{1+z^{\star}} \sim
				\left| \Delta z_{0}\right|^{1+z^{\star}}
				\cdot e^{-\alpha \gamma t}
			\end{equation}
			which, solved for $\Delta z$, gives, for $t \to +\infty$:
			\begin{equation}
				\label{stability_solution7}
				|\Delta z| \sim
				\left| \Delta z_{0}\right|
				\cdot
				e^{-\frac{(-2z^{\star 2}+4z^{\star}+1)}{1-z^{\star 2}}\gamma t}\ 
				\to 0.
			\end{equation}
			We can identify $\zeta=-\frac{(-2z^{\star 2}+4z^{\star}+1)}{1-z^{\star 2}}\gamma$ as the Lyapunov exponent of the dynamical system. In particular, $\zeta$ is always negative, it is equal to $-2$ for $z^{\star}=\frac{3}{4}$ and it goes to $-\infty$ as $z^{\star}\to 1$.

			\item $\alpha< 0$, that is $0\leq z^{\star}< 1-\frac{\sqrt 2}{2}$. In the ratio
			\begin{equation*}
				\frac{\left| \Delta z+\frac{\alpha}{2}\right|^{\frac{4z^{\star}-3}{2z^{\star}-2}} }{\left| \Delta z \right|^{1+z^{\star}}}=k\cdot e^{\alpha \gamma t}
			\end{equation*}
			the right-hand side goes to $0$ for $t \to +\infty$. This implies that either $\Delta z\to \left| \frac{\alpha}{2} \right|$ or $\Delta z\to \infty$. If $\Delta z\to \left| \frac{\alpha}{2} \right|$, the numerator goes to $0$ and the exponent is positive, so that this a consistent solution. If $\Delta z\to \infty$, then ${\frac{\alpha}{2}}$ is negligible, and the ratio is asymptotic again to $|\Delta z|^{-\frac{\alpha}{2}}$, but, since $\alpha<0$, this quantity goes to infinity. 
			Therefore, the only consistent possibility is the first one, where the error tends to a finite value, equal to $\left|\frac{\alpha}{2}\right|=\left|\frac{2z^{\star 2}-4z^{\star}+1}{2(z^{\star}-1)}\right|$.
			Therefore, in the interval $0\leq z^{\star}< 1-\frac{\sqrt 2}{2}$, the error is positive and finite. Specifically, $|\frac{\alpha}{2}|$ is a decreasing function of $z^{\star}$ and varies from $\frac{1}{2}$ to $0$. This implies that, in this interval, we cannot find any stable solution.
		\end{itemize}
		Finally for $\alpha=0$, by solving the integral, Eq. (\ref{stability_equation2}) becomes: $\frac{1}{2\sqrt{2} \Delta z}+\frac{1}{2}\log (2\Delta z)=\gamma t+k$. Therefore, $\Delta z\to 0$ as $t \to \infty$ and the solution is stable.
	\end{proof}
	
	\begin{remark}
		By Theorem \ref{theorem3}, it follows that $1-\frac{\sqrt{2}}{2}$ represents a critical value for the stability of the asymptotic solution. This role is further confirmed by the following observation that applies to the cycle graph. In Theorem \ref{theorem1_cycle} and its proof, we found that the unstable and stable solution for the cycle are $x_{1}^{\star}=\frac{1}{2}\left( 1- \sqrt{1-\frac{2}{\mathcal{R}}} \right)$ and $x_{2}^{\star}=\frac{1}{2}\left( 1+ \sqrt{1-\frac{2}{\mathcal{R}}} \right)$, respectively.
		If we set the initial probability equal to $p=1-\frac{\sqrt 2}{2}$, then the threshold of the model is $\tau_{c}=\frac{1}{2p(1-p)}=1+\sqrt{2}$. Above this threshold, that is for $\mathcal{R}\geq 1+\sqrt{2}$, the unstable solution lies exactly in the instability interval claimed by Theorem \ref{theorem3}, that is $0<x_{1}^{\star}\leq 1-\frac{\sqrt 2}{2}$. Moreover, the stable one lies in the range $\frac{\sqrt 2}{2}<x_{2}^{\star}< 1$. Furthermore, as will be shown in Appendix \ref{appendixA}, for general regular graphs, the unstable solution 
		$x^{\star}_{1}=\frac{1}{2}\left(1-\frac{d-2}{2d(d-1)\mathcal{R}}- \frac{\sqrt{\xi}}{2d(d-1)\mathcal{R}}\right)$ lies below $1-\frac{\sqrt 2}{2}$ exactly for $0< \mathcal{R} \leq {\tau}_{1} \cup \mathcal{R} \geq {\tau}_{2}$ (see  Eq. (\ref{conditions}) in  Appendix \ref{appendixregular}). 
	\end{remark}
	
	Let us conclude with a theorem that characterizes the existence and stability of the disease-free steady state for general topology.
	\begin{theorem}
		If ${\mathcal R}<\frac{1}{p\lambda_1}$, the disease-free steady state of the ASIS model exists and it is stable 
		\label{theorem4}
	\end{theorem}
	
	\begin{proof}
		We show that if the model in Eq. (\ref{ODE1}) is initially below the threshold of the standard SIS model (\ref{SIS3}), it remains below that threshold throughout the process.
		In other words, the existence of the extinction steady state is determined by the initial conditions alone.
		
		At $t=0$, the values of the matrices ${\bf A}_{P}(t)$ and ${\bf A}_{D}(t)$ are given by Eq. (\ref{initialmatrices}), that is ${\bf A}_{P}(0)=p{\bf B}_{P}$ and ${\bf A}_{D}(0)=p{\bf B}_{D}$, where ${\bf B}_{P}$ and ${\bf B}_{D}$ represent the original numerical adjacency matrices of the network and its line graph. Therefore, the early stages of the process are governed by the equation
		\begin{equation}
			{\bf \dot{z}}=
			\beta p
			\left[ {\bf I} - {\rm diag}\, {\bf z} \right]
			\left[ \begin{array}{cc} 
				{\bf B}_{P} & {\bf 0}_{n\times m} \\
				{\bf 0}_{m\times n} & {\bf B}_{D} \\
			\end{array} \right]{\bf z} -\gamma {\bf z}.
			\label{matrix6}
		\end{equation}
		According to \citet{Kiss2017} (see Theorem (3.8), \textit{ibidem}), the process described by Eq. (\ref{matrix6}) exhibits a transcritical bifurcation at the critical value ${\mathcal R}=\frac{1}{p\lambda_{1}}$, where $\lambda_{1}$ is the largest eigenvalue of the block matrix\footnote{Note that the block matrix in Eq. (\ref{blocmatrix}) has two greatest eigenvalues corresponding to the Perron-Frobenius eigenvalues of the two separate matrices ${\bf B}_{P}$ and ${\bf B}_{D}$. In a regular graph, for instance, the first eigenvalue $2(d-1)$  of the matrix ${\bf B}_{D}$ is greater than the corresponding eigenvalue $d$ of the matrix ${\bf B}_{P}$.}
		\begin{equation}
			\left[ \begin{array}{cc} 
				{\bf B}_{P} & {\bf 0}_{n\times m} \\
				{\bf 0}_{m\times n} & {\bf B}_{D} \\
			\end{array} \right].
			\label{blocmatrix}
		\end{equation}
		In particular, the value $\frac{1}{p\lambda_1}$ represents a lower bound for the epidemic threshold $\tau$ of the standard process, that is $\frac{1}{p\lambda_1}<\tau$.
		The key point is the presence of the initial value $p$ in the denominator of this lower bound for the threshold, while the term $\lambda_1$ is fixed and determined only by the original topological structure of the network.
		If, at $t_{0}=0$, we have ${\mathcal R}<\frac{1}{p\lambda_1}<\tau$, then the process starts reducing the individual probabilities $x_{i}(t)$ and  $y_{i}(t)$, on both the graph and the line graph, so that, at a later time ${t_1}>{t_0}$, we have $z_{i}({t_1})<p$, $\forall i=1,\dots, n+m$. Now, by Algorithm \ref{algorithm}, we replace the original weights in the adjacency matrices by the new values $z_{i}({t_1})$.
		This step can be replicated at each subsequent time $0={t_0}<{t_1}<{t_2}<\dots <{t_k}<\dots$, so that $z_{i}({t_k})<z_{i}({t_{k-1}})$ and
		\begin{equation}
			\footnotesize
			{\mathcal R}<\frac{1}{p\lambda_1}<\frac{1}{\max_{i}{z_{i}({t_1})}\lambda_1}<\frac{1}{\max_{i}{z_{i}({t_2})}\lambda_1}<\dots <\frac{1}{\max_{i}{z_{i}({t_k})}\lambda_1}.
		\end{equation}
		Then, the process, at every step, remains below the corresponding threshold of the standard SIS model.
		Since, under these assumptions, the latter has a stable null asymptotic solution, the solution $z_{i}=0$, $\forall i=1,\dots, n+m$ of the ASIS model exists and is asymptotically stable.
	\end{proof}
	Let us observe that Theorem \ref{theorem4} implies that, as far as the extinction steady state is concerned, controlling the initial stages of the process, means controlling the whole process.
	
	\section{Self-adaptive eigenvector centrality}
	\label{Self-adaptive eigenvector centrality}
	We want to show now that the components of the eigenvectors $\psi_{\bf M}^{(1)}$ and $\psi_{\bf M}^{(2)}$ introduced in Section \ref{Nonlineareigenproblem}, and, therefore, the values of the stationary probabilities appropriately normalized, can be interpreted as nonlinear eigenvector centralities.
	
	The idea stems from the observation that, in the limit ${\bf z}\to {\bf 0}$, the matrix ${\bf M}({\bf z})$ approaches ${\mathcal R} {\bf G}\left({\bf z}\right)$, so that
	Eq. (\ref{eigenvalue_problem}) has the typical implicit form that defines an eigenvector centrality. Given a weighted adjacency matrix, we search for the dominant eigenvector whose components are interpreted as a score in which the importance of a node is proportional to that of its  \textit{neighboring elements}, typically adjacent nodes.
	Two aspects distinguish Eq. (\ref{eigenvalue_problem}) from a usual equation defining eigenvector centrality: the presence of matrices that depend on the eigenvectors themselves, as already discussed, and the trade-off between the centralities of the nodes and those of the edges. In fact, Eq. (\ref{eigenvalue_problem}) implies that the centrality of a node is a function of the centrality of the edges it belongs to and the centrality of an edge is a function of the centrality of its extreme nodes.
	
	A similar idea has already been proposed by \citet{Tudisco2021} within a more general but static setting. A generalization of their approach emerges here within a dynamic setting in a quite natural way. Let us observe that the authors define a node and edge score such that the importance $y_{j}$ of an edge $e_{j}\in E$ is a nonnegative number proportional to the importance of the nodes in $e_{j}$, and the importance $x_{i}$ of a node $v_{i}\in V$ is a nonnegative number proportional to the importance of the edges it participates in. In a notation consistent with our paper, their centralities are given by the following equations
	\begin{equation}
		\left\{ 
		\begin{array}{l}
			\lambda {\bf x}={\bf E}\, {\rm diag}({\bf y}_{0}) {\bf y}\\
			\hfill \\
			\mu {\bf y}={\bf E}^{T}\, {\rm diag}({\bf x}_{0}) {\bf x}\\
		\end{array}
		\right.
	\end{equation}
	which are equivalent to
	\begin{equation}
		\left\{ 
		\begin{array}{l}
			{\bf x}=\rho  \left[ {\bf A}_{P}({\bf y}_{0})+{\bf K}_{P}({\bf y}_{0})  \right]  {\rm diag}({\bf x}_{0}) {\bf x}\\
			\hfill \\
			{\bf y}=\rho  \left[ {\bf A}_{D}({\bf x}_{0})+{\bf K}_{D}({\bf x}_{0}) \right] {\rm diag}({\bf y}_{0}) {\bf y}\\
		\end{array}
		\right.
		\label{HighamTudisco}
	\end{equation}
	where $\rho=1/\mu \lambda$. 
	The authors compute the Perron eigenvectors ${\bf x}^{\star}$ and ${\bf y}^{\star}$ of diagonally perturbed adjacency matrices of the graph and the line graph and interpret their components as eigenvector scores for the nodes and the edges, respectively.
	Eq. (\ref{eigenvalue_problem}), namely
	\begin{equation}
		\left\{ 
		\begin{array}{l}
			{\bf x}^{\star} ={\mathcal R}\left[ {\bf I}_{n} - {\rm diag}\, {\bf x}^{\star} \right] {\bf A}_{P}({\bf y}^{\star})\, {\bf x}^{\star} \\
			\hfill \\
			{\bf y}^{\star} ={\mathcal R}\left[ {\bf I}_{m} - {\rm diag}\, {\bf y}^{\star} \right] {\bf A}_{D}({\bf x}^{\star})\, {\bf y}^{\star}\\
		\end{array}
		\right.
		,
		\label{eigenvalue_problem2}
	\end{equation}
	play the same role of Eq. \eqref{HighamTudisco}.
	In this perspective, our model leads to a new centrality measure that we call
	\textit{self-adaptive eigenvector centrality}.
	Such a measure weights the score of an element, either a node or an edge, as proportional to the score of all the elements, nodes and edges, to which it is connected.
	
	Our centrality measure is similar to the one defined by \citet{Tudisco2021}, but with some remarkable differences.
	First, in Eq. (\ref{HighamTudisco}), the matrices are all evaluated at initial fixed values, which correspond to the topological weights of the edges in the graph and in the line graph and that we identified, in our notation, with the initial values ${\bf x}_{0}$ and ${\bf y}_{0}$. Conversely, in Eq. (\ref{eigenvalue_problem2}), matrices dynamically update with the weights computed on the basis of an evolutionary process. Second, the dependence of the elements of the matrices on the scores to be attributed to nodes and edges has a retroactive effect on the meaning of these scores. Let us consider, for instance, a node $i$ in the network $G_P$. Its score turns out to be proportional to $\sum_{j}{\bf A}_{P}({\bf y})_{ij}{x}_{j}$, that is the sum of the products between the score of its neighboring nodes and the score of the corresponding edges connecting them to node $i$. Hence, in our model, the centrality of a node does not depend on the importance of neighboring nodes alone or adjacent edges alone, but on the joint effect of both these elements.

\section{Illustrative example}
\label{Illustrative example}
Let us examine the implementation of the ASIS model through the example illustrated in Fig. \ref{Cartoon}.
The adjacency matrices of the network $G_P$ and $G_D$ are, respectively,
	\begin{equation*}
		{\bf B}_{P}=\left[ 
		\begin{array}{cccc}
			0 & 1 & 1 & 0 \\ 
			1 & 0 & 1 & 0 \\ 
			1 & 1 & 0 & 1 \\ 
			0 & 0 & 1 & 0
		\end{array}
		\right],
		\qquad
		{\bf B}_{D}=\left[ 
		\begin{array}{cccc}
			0 & 1 & 1 & 0 \\ 
			1 & 0 & 1 & 1 \\ 
			1 & 1 & 0 & 1 \\ 
			0 & 1 & 1 & 0
		\end{array}
		\right].
	\end{equation*}
	Matrices in Eq. (\ref{ApAd_continuos}), at $t=0$, are then ${\bf A}_{P}(0)=p{\bf B}_{P}$ and ${\bf A}_{D}(0)=p{\bf B}_{D}$, for $0<p<1$. By introducing the variable ${\bf z}\in {\mathbb R}^{8}$,
	matrix ${\bf G}({\bf z})$ in Eq. (\ref{matrixG}) and vector ${\bf H}({\bf z}){\bf z}$ in Eq. (\ref{matrixH}) take the form
	\begin{equation*}
		{\bf G}({\bf z})=\left[ 
		\begin{array}{cccccccc}
			0 & y_{1} & y_{2} & 0 & 0 & 0 & 0 & 0 \\ 
			y_{1} & 0 & y_{3} & 0 & 0 & 0 & 0 & 0 \\ 
			y_{2} & y_{3} & 0 & y_{4} & 0 & 0 & 0 & 0 \\ 
			0 & 0 & y_{4} & 0 & 0 & 0 & 0 & 0 \\
			0 & 0 & 0 & 0 & 0 & x_{1} & x_{2} & 0 \\
			0 & 0 & 0 & 0 & x_{1} & 0 & x_{3} & x_{3} \\
			0 & 0 & 0 & 0 & x_{2} & x_{3} & 0 & x_{3} \\
			0 & 0 & 0 & 0 & 0 & x_{3} & x_{3} & 0
		\end{array}
		\right]
	\end{equation*}
	\begin{equation*}
		{\bf H}({\bf z}){\bf z}=
		\left[ 
		\begin{array}{l}
			\beta (1-x_{1})(x_{2}y_{1}+x_{3}y_{2})-\gamma x_{1}\\
			\beta (1-x_{2})(x_{1}y_{1}+x_{3}y_{3})-\gamma x_{2}\\
			\beta (1-x_{3})(x_{1}y_{2}+x_{2}y_{3}+x_{4}y_{4})-\gamma x_{3}\\
			\beta (1-x_{4})\cdot x_{3}y_{4}-\gamma x_{4}\\
			\beta (1-y_{1})(x_{1}y_{2}+x_{2}y_{3})-\gamma y_{1}\\
			\beta (1-y_{2})(x_{1}y_{1}+x_{3}y_{3}+x_{3}y_{4})-\gamma y_{2}\\
			\beta (1-y_{3})(x_{2}y_{1}+x_{3}y_{2}+x_{3}y_{4})-\gamma y_{3}\\
			\beta (1-y_{4})(x_{3}y_{2}+x_{3}y_{3})-\gamma y_{4}\\
		\end{array}
		\right].
	\end{equation*}
	The nonlinear eigenproblem ${\bf z}^{\star}=
	{\mathcal R} \left[ {\bf I} - {\rm diag}\, {\bf z}^{\star} \right] {\bf G}\left({\bf z}^{\star}\right){\bf z}^{\star}$, described in Eq. ({\ref{eigenvalue_problem}}), that leads to the steady states solutions and to the self-adaptive eigenvector centralities is explicitly
	\begin{equation}
		\left[ 
		\begin{array}{l}
			x_{1}^{\star}\\
			x_{2}^{\star}\\
			x_{3}^{\star}\\
			x_{4}^{\star}\\
			y_{1}^{\star}\\
			y_{2}^{\star}\\
			y_{3}^{\star}\\
			y_{4}^{\star}\\
		\end{array}
		\right]
		=
		\left[ 
		\begin{array}{l}
			\mathcal{R} (1-x_{1}^{\star})(x_{2}^{\star}y_{1}^{\star}+x_{3}^{\star}y_{2}^{\star})\\
			\mathcal{R} (1-x_{2}^{\star})(x_{1}^{\star}y_{1}^{\star}+x_{3}^{\star}y_{3}^{\star})\\
			\mathcal{R} (1-x_{3}^{\star})x_{1}^{\star}y_{2}^{\star}+x_{2}^{\star}y_{3}^{\star}+x_{4}^{\star}y_{4}^{\star})\\
			\mathcal{R} (1-x_{4}^{\star})\cdot x_{3}^{\star}y_{4}^{\star}\\
			\mathcal{R} (1-y_{1}^{\star})(x_{1}^{\star}y_{2}^{\star}+x_{2}^{\star}y_{3}^{\star})\\
			\mathcal{R} (1-y_{2}^{\star})x_{1}^{\star}y_{1}^{\star}+x_{3}^{\star}y_{3}^{\star}+x_{3}^{\star}y_{4}^{\star})\\
			\mathcal{R} (1-y_{3}^{\star})(x_{2}^{\star}y_{1}^{\star}+x_{3}^{\star}y_{2}^{\star}+x_{3}^{\star}y_{4}^{\star})\\
			\mathcal{R} (1-y_{4}^{\star})(x_{3}^{\star}y_{2}^{\star}+x_{3}^{\star}y_{3}^{\star})\\
		\end{array}
		\right].
		\label{eigenexample}
	\end{equation}
	
	By Eq. (\ref{eigenexample}), it is clear that the centrality of a node is proportional to the sum of the products of the respective scores of nodes and edges connected to it.
	For example, the centrality $x_{1}^{\star}$ of the node $1$ is proportional to $(x_{2}^{\star}y_{1}^{\star}+x_{3}^{\star}y_{2}^{\star})$: the first term is the product of the score of node $2$ and the score of the edge connecting nodes $1$ and $2$; the second term is the product of the score of node $3$ and the score of the edge connecting nodes $1$ and $3$. 
	
	We now present some numerical experiments. In Fig. \ref{fig3}, panels (a-c), we show the prevalence of infected/adopted individuals in the network $G_P$, that is the cumulative probabilities $x_{i}(t)$ as functions of $t$, under different conditions. Nodes $1$ and $2$ are equivalent and the curves have the same color code as in Fig. \ref{Cartoon}. Fig. \ref{fig3}, panels (d-f), shows the incidence, that is the instantaneous increments $dx_{i}(t)$, under the same corresponding conditions. Node $3$, as expected, is the node with the highest asymptotic probability, being the most central. The opposite for node $4$.
	
	Fig. \ref{fig9} represents the contour plots of the mean prevalence for the network $G_P$, under different conditions and at different times. The mean prevalence in the plots is the arithmetic mean of the probabilities $x_{i}(t)$ in the network $G_P$. In Fig. \ref{fig9}, panel (a), we plot a snapshot at a fixed time of the mean prevalence as a function of the infection rate $\beta$ and recovery rate $\gamma$. In Fig. \ref{fig9}, panel (b), we plot the phase diagram at a fixed value of the infection rate $\beta$ as a function of $\gamma$ and $t$, and, in Fig. \ref{fig9}, panel (c), the phase diagram at a fixed value of the recovery rate $\gamma$ as a function of $\beta$ and $t$. The last two panels make it clear the presence of a transcritical bifurcation at a specific value of the reproductive number $\mathcal{R}$. The values of the parameters used to build the plots are specified in the caption of the figure.

{\onecolumngrid
	\vspace{\columnsep}
		\begin{center}
			\begin{figure}[H]
				\centering
				\subfloat[]{\includegraphics[width=0.40\textwidth]{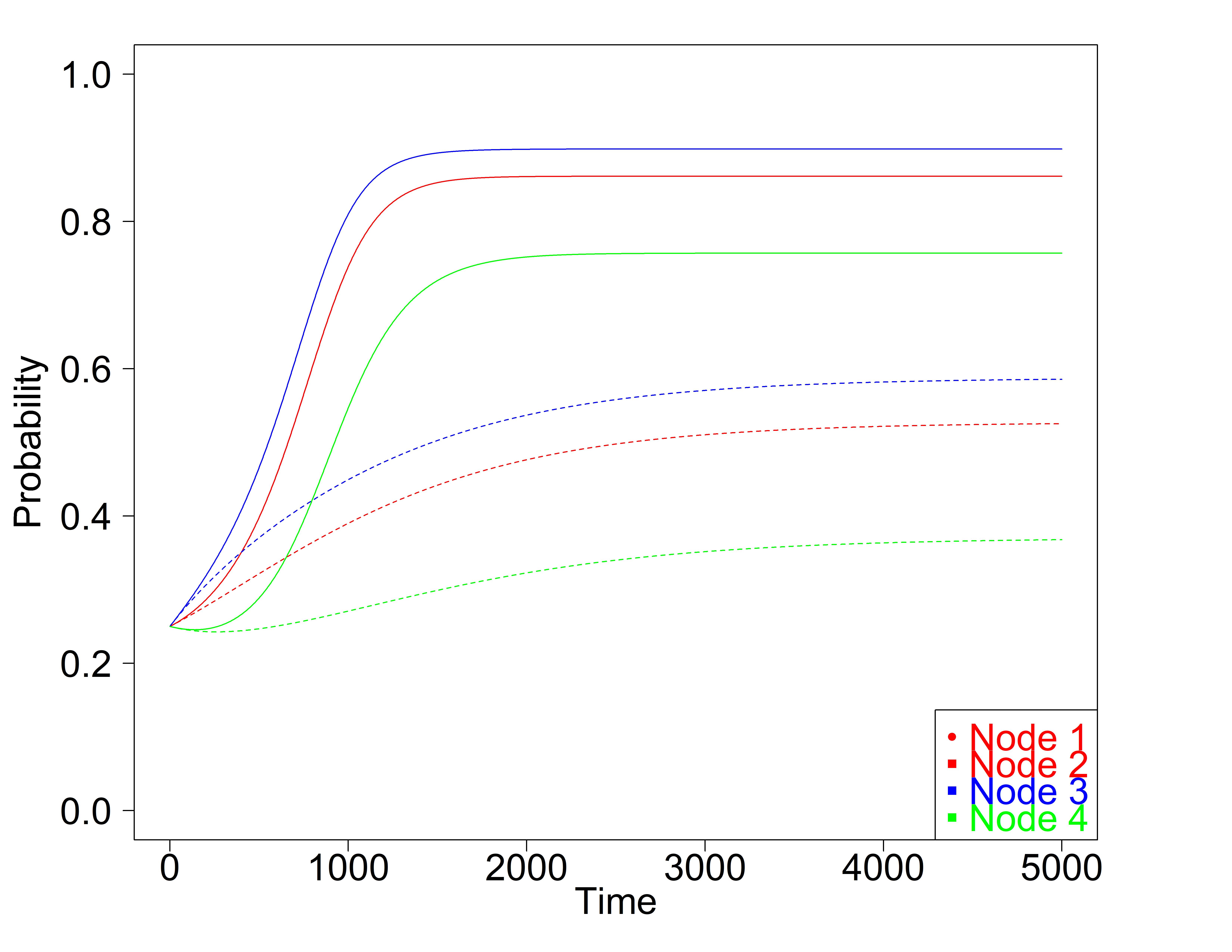}}
				\subfloat[]{\includegraphics[width=0.40\textwidth]{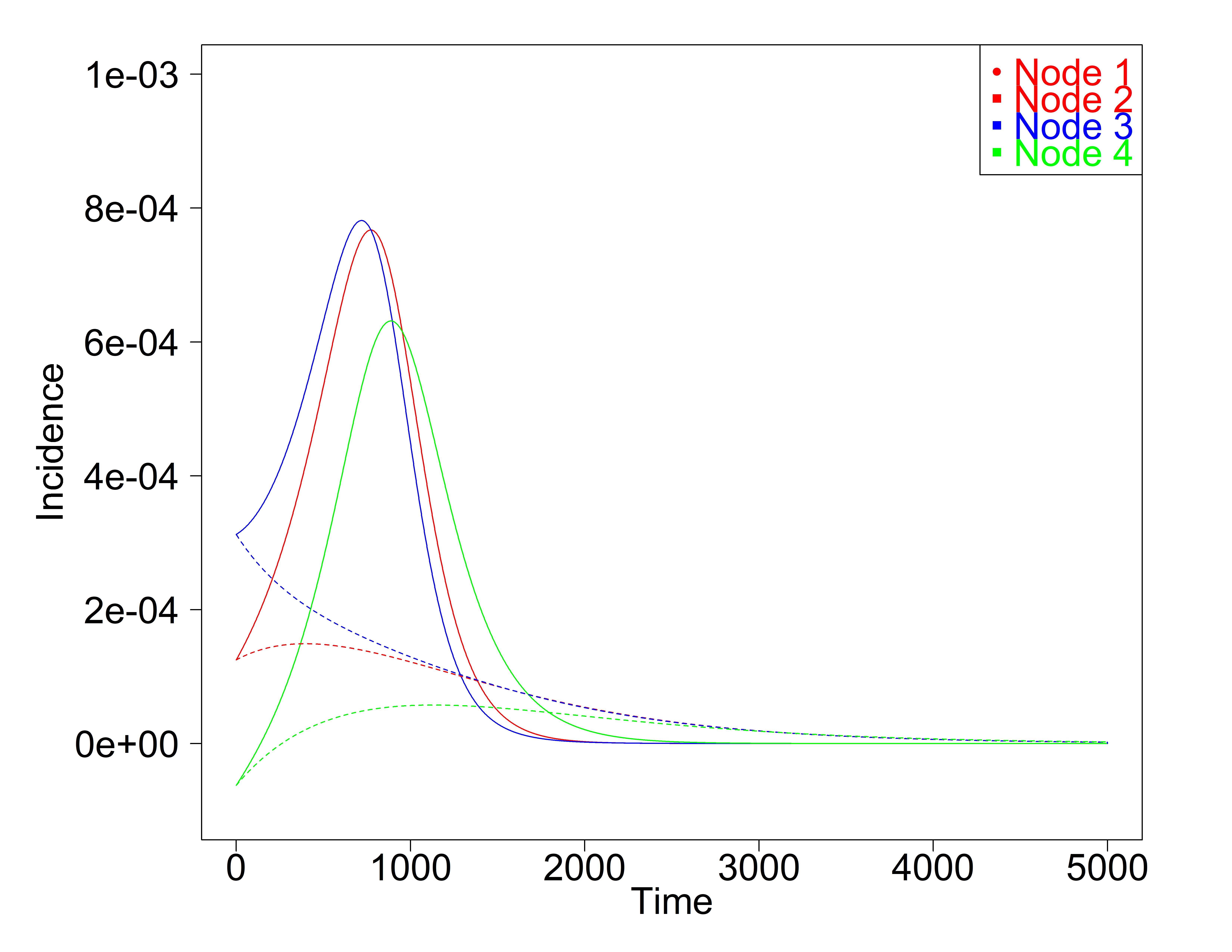}}\\
				\subfloat[]{\includegraphics[width=0.40\textwidth]{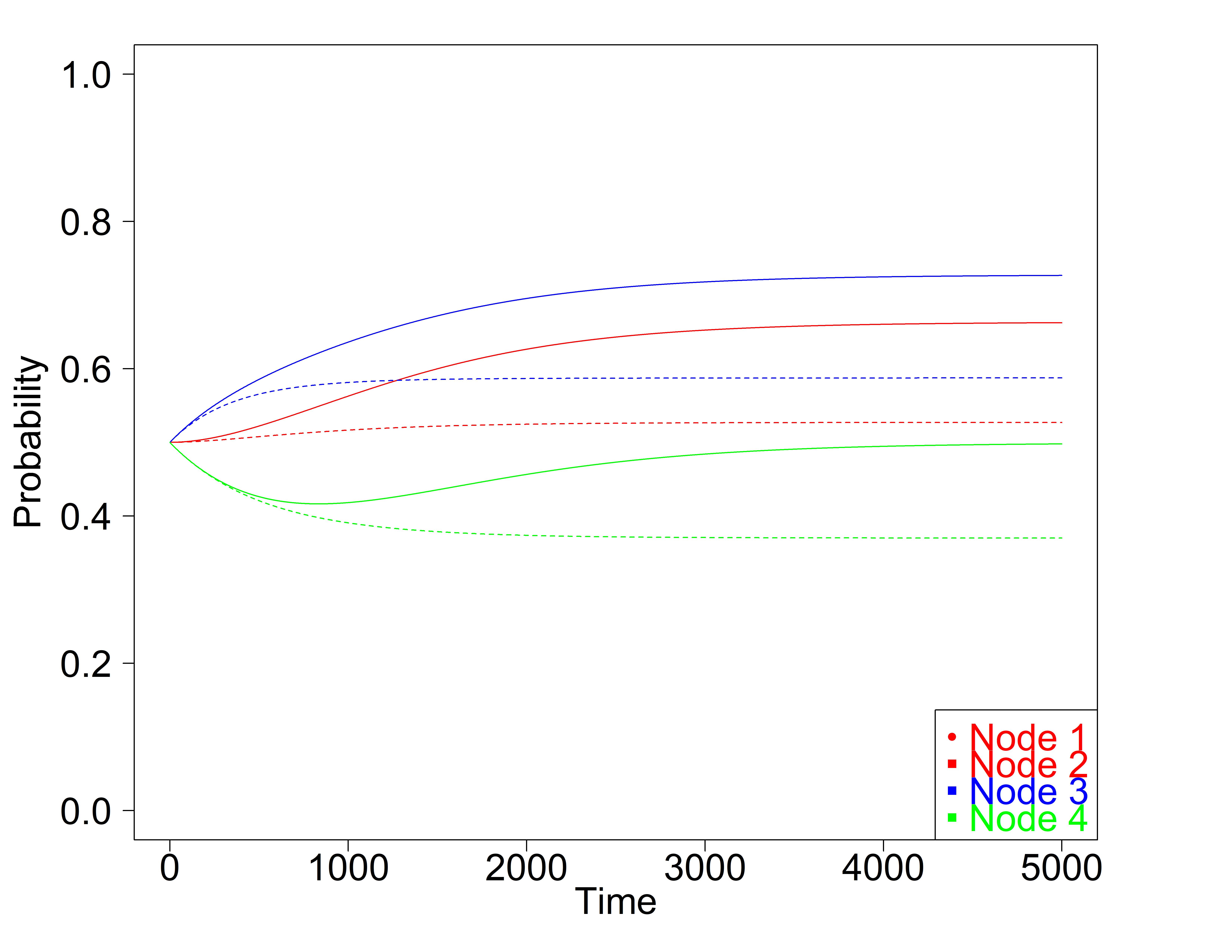}}
				\subfloat[]{\includegraphics[width=0.40\textwidth]{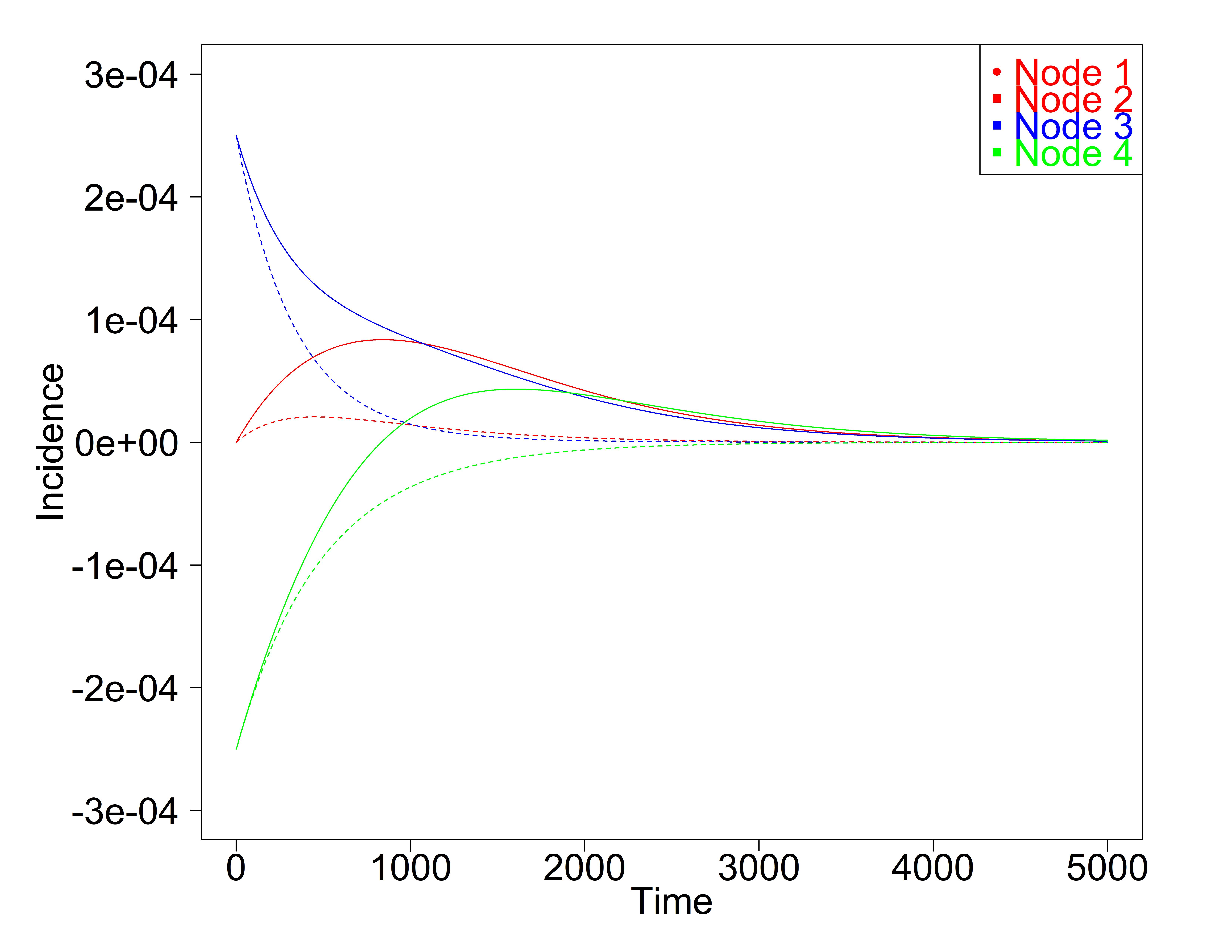}}\\
				\subfloat[]{\includegraphics[width=0.40\textwidth]{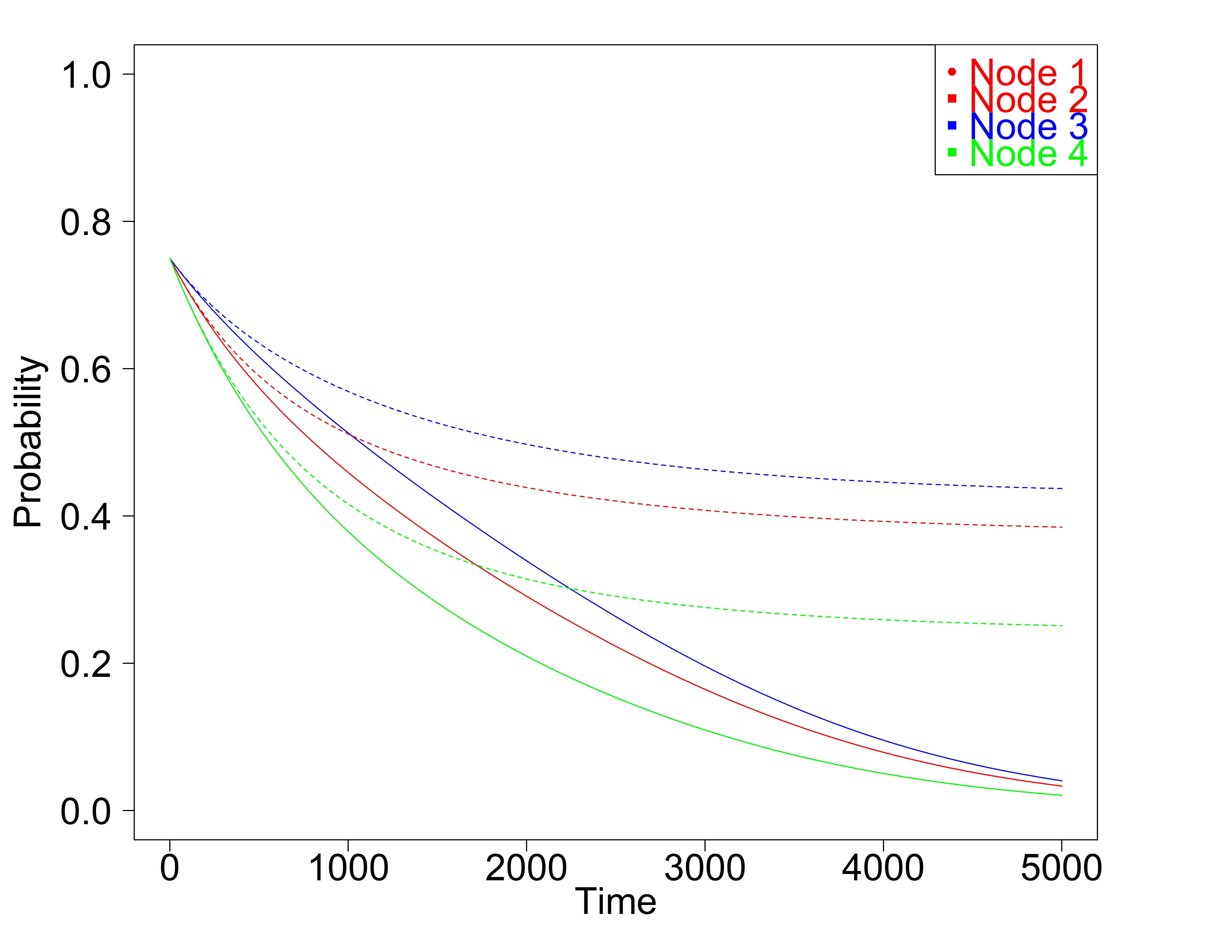}}
				\subfloat[]{\includegraphics[width=0.40\textwidth]{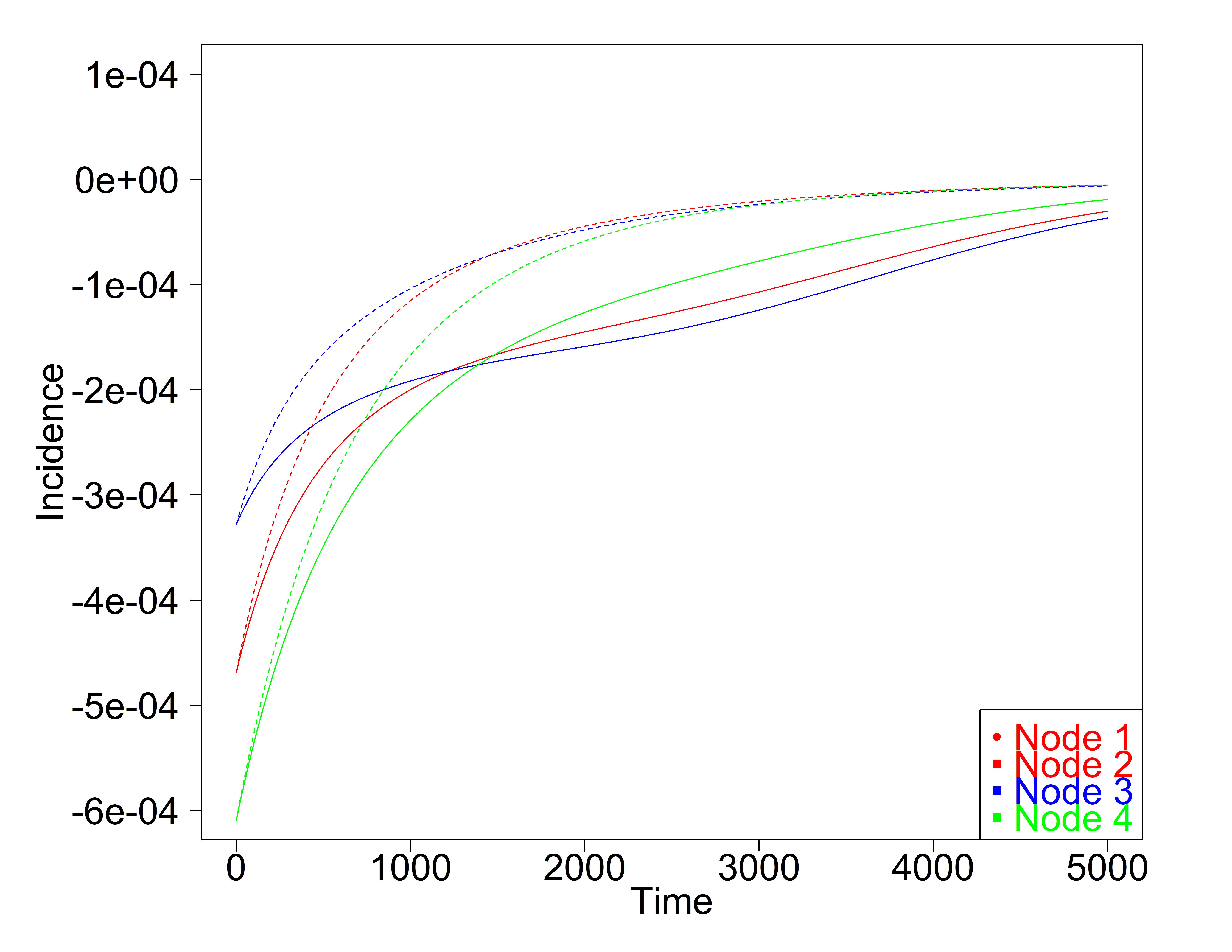}}
				\caption{Prevalence (panels (a), (c) and (e)) and incidence (panels (b),  (d) and (f)) for individual nodes in the examined network, under different conditions; (a) and (b): $p=0.25$, $\beta=0.004$ and $\gamma=0.001$; (c) and (d): $p=0.50$, $\beta=0.002$ and $\gamma=0.001$; (e) and (f) $p=0.75$, $\beta=0.001$ and $\gamma=0.001$. Solid lines represent the self-adaptive SIS model, dashed lines the standard SIS model.}
				\label{fig3} 
			\end{figure}
		\end{center}
	\vspace{\columnsep}
	\twocolumngrid
}

\begin{center}
	\begin{figure}[H]
		\centering
		\subfloat[]{\includegraphics[width=0.40\textwidth]{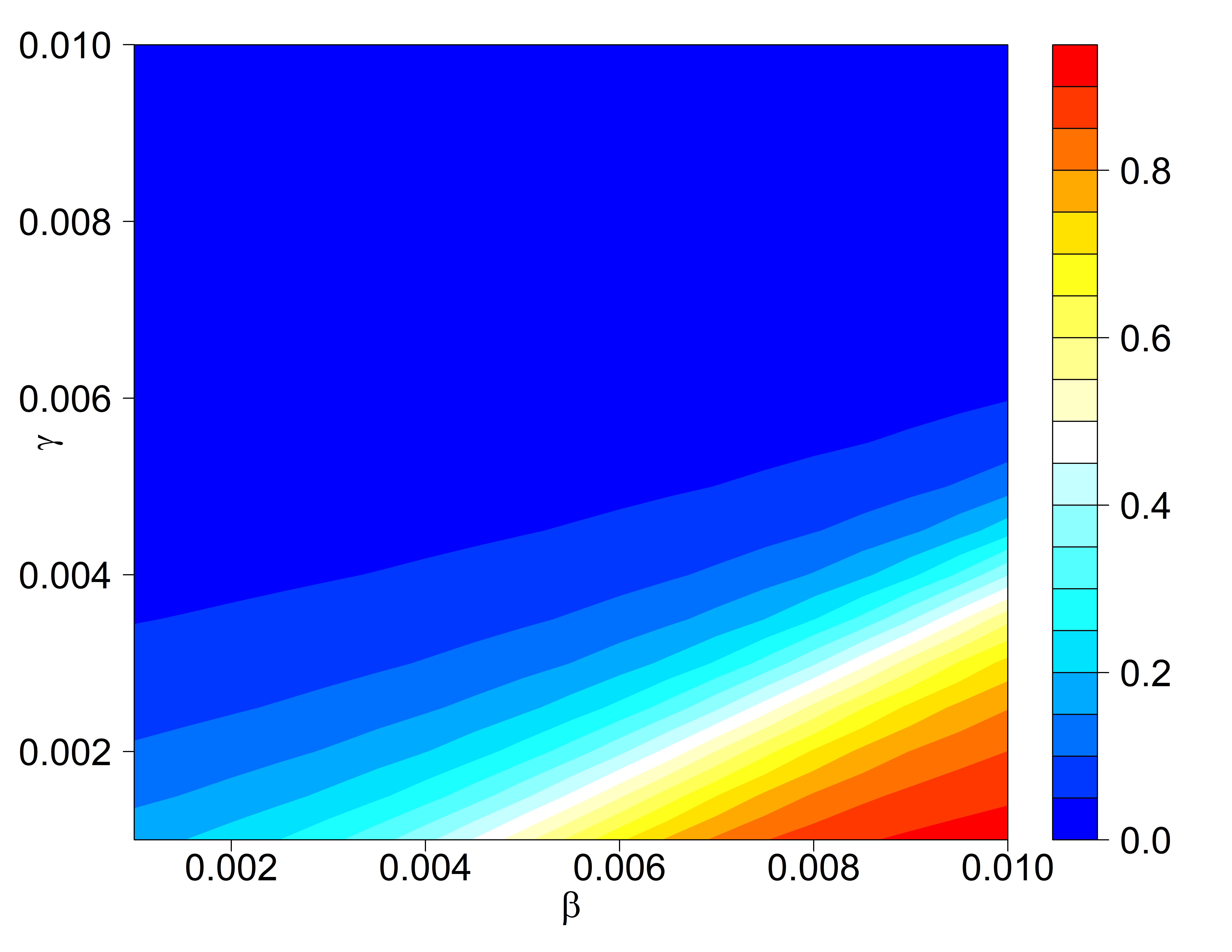}}\\
		\subfloat[]{\includegraphics[width=0.40\textwidth]{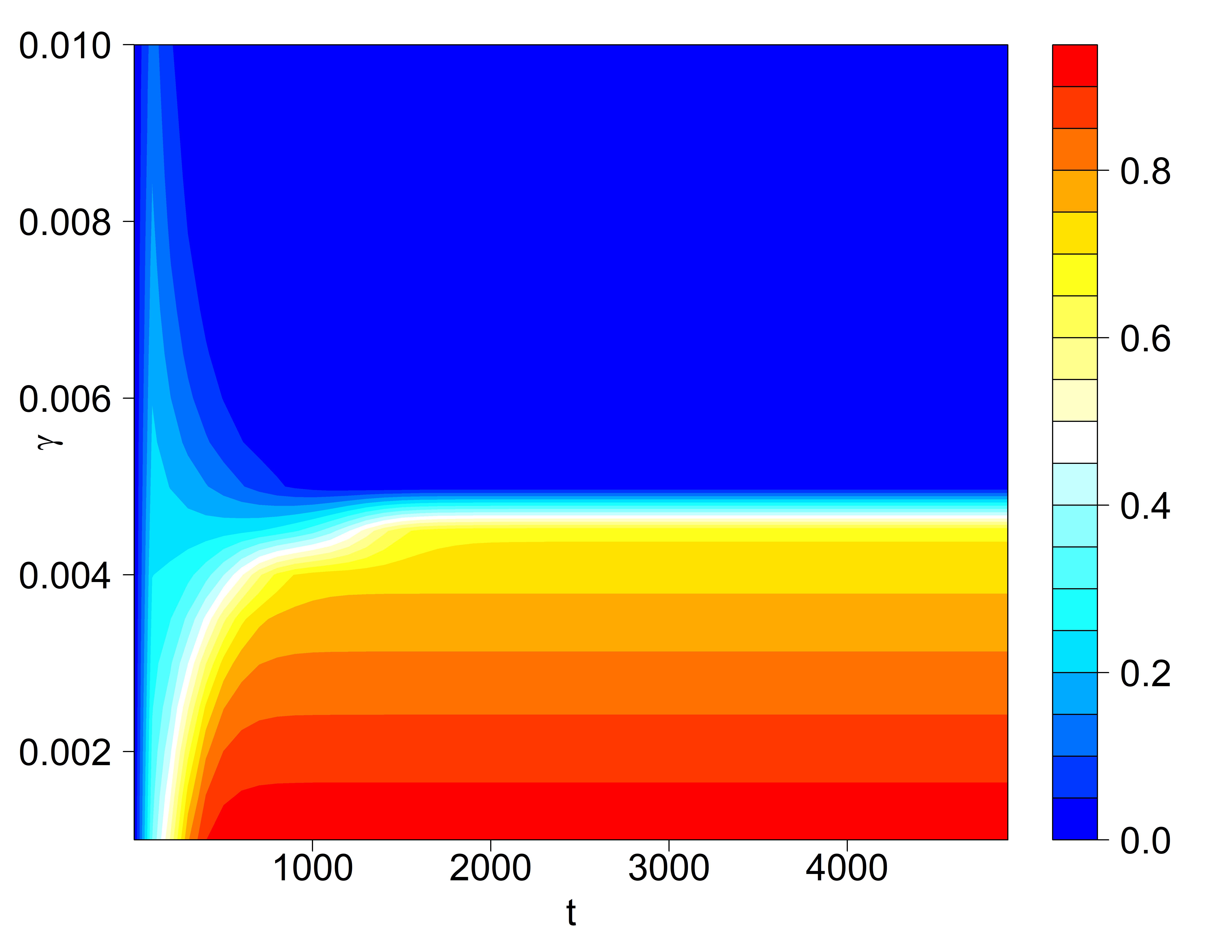}}\\
		\subfloat[]{\includegraphics[width=0.40\textwidth]{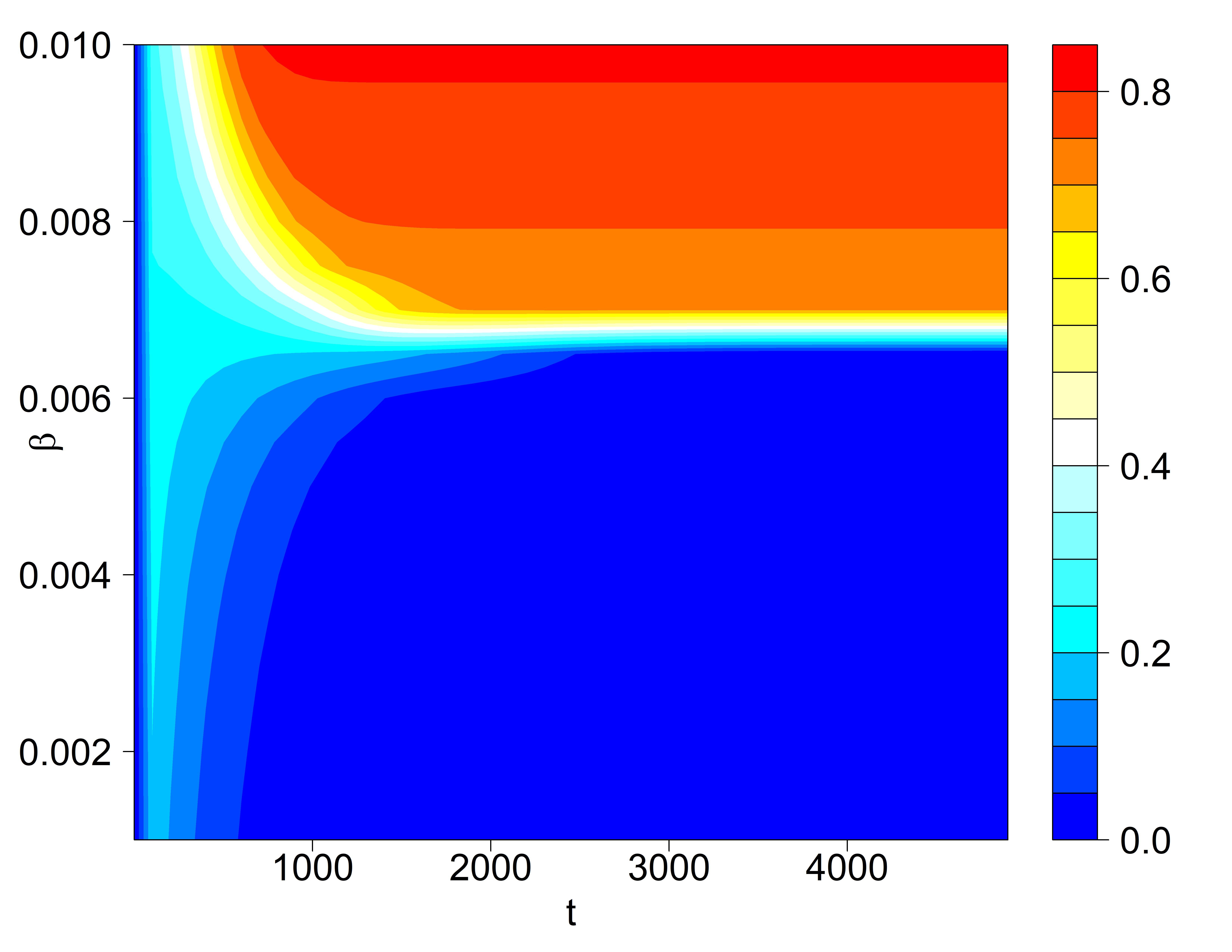}}
		\caption{Mean prevalence in the toy network, under different conditions on $\beta$ and $\gamma$ and at different times $t$. In panel (a) $t=500$; in panel (b) $\beta=0.010$; in panel (c) $\gamma=0.003$.}
		\label{fig9} 
	\end{figure}
\end{center}
 \vspace{-4mm}
Let us now focus on the particular case where $p=0.25$, $\beta=0.004$ and $\gamma=0.001$, which is represented in Fig. \ref{fig3}, panels (a) and (b). The steady states of the four nodes are: $x^{\star}_{1}=x^{\star}_{2}=0.8613893$, $x^{\star}_{3}=0.8984515$ and $x^{\star}_{4}=0.7569336$.\footnote{Similarly, we find for the edges $y^{\star}_{a}=0.8615711$ and $y^{\star}_{b}=y^{\star}_{c}=0.9031810$, and $y^{\star}_{d}=0.8665192$.}
The corresponding normalized eigenvectors are given by
\begin{align*}
\psi_{\bf M}^{(1)}=[0.508981,0.508981,0.530881,0.447260,0,0,0,0 ]^{T}, \\
\psi_{\bf M}^{(2)}=[0,0,0,0,0.487407,0.510946,0.510946,0.490206 ]^{T}. 
\end{align*}
The endemic steady state at the end of the ASIS process is then, as expected, the non-normalized dominant eigenvector of the matrix ${\bf M}({\bf z}^{\star})$, and the values of the final probabilities of each node are proportional to the components of the dominant normalized eigenvectors of the matrix ${\bf M}({\bf z}^{\star})$.
The components of the two eigenvectors $\psi_{\bf M}^{(1)}$ and $\psi_{\bf M}^{(2)}$ are therefore interpreted as the self-adaptive eigenvector centralities for nodes and edges defined in Section \ref{Self-adaptive eigenvector centrality}.
	
The stability of these solutions has been analyzed in subsection \ref{Stability of the general endemic and disease-free steady states}. The error $|\Delta z(t)|=|z(t)-z^{\star}|$ can be computed by Eq. (\ref{stability_solution7}), which predicts an exponential decay as a function of time. In Fig. \ref{fig5}, we illustrate in log-scale the exponential decay of the numerical error for the four nodes in the network under examination. Specifically, in the numerical simulation, we choose $|\Delta z_{0}|=|z(1000)-z^{\star}|$. 
\vspace{-2mm}
\begin{figure}[H]
	\centering
	{\includegraphics[width=0.40\textwidth]{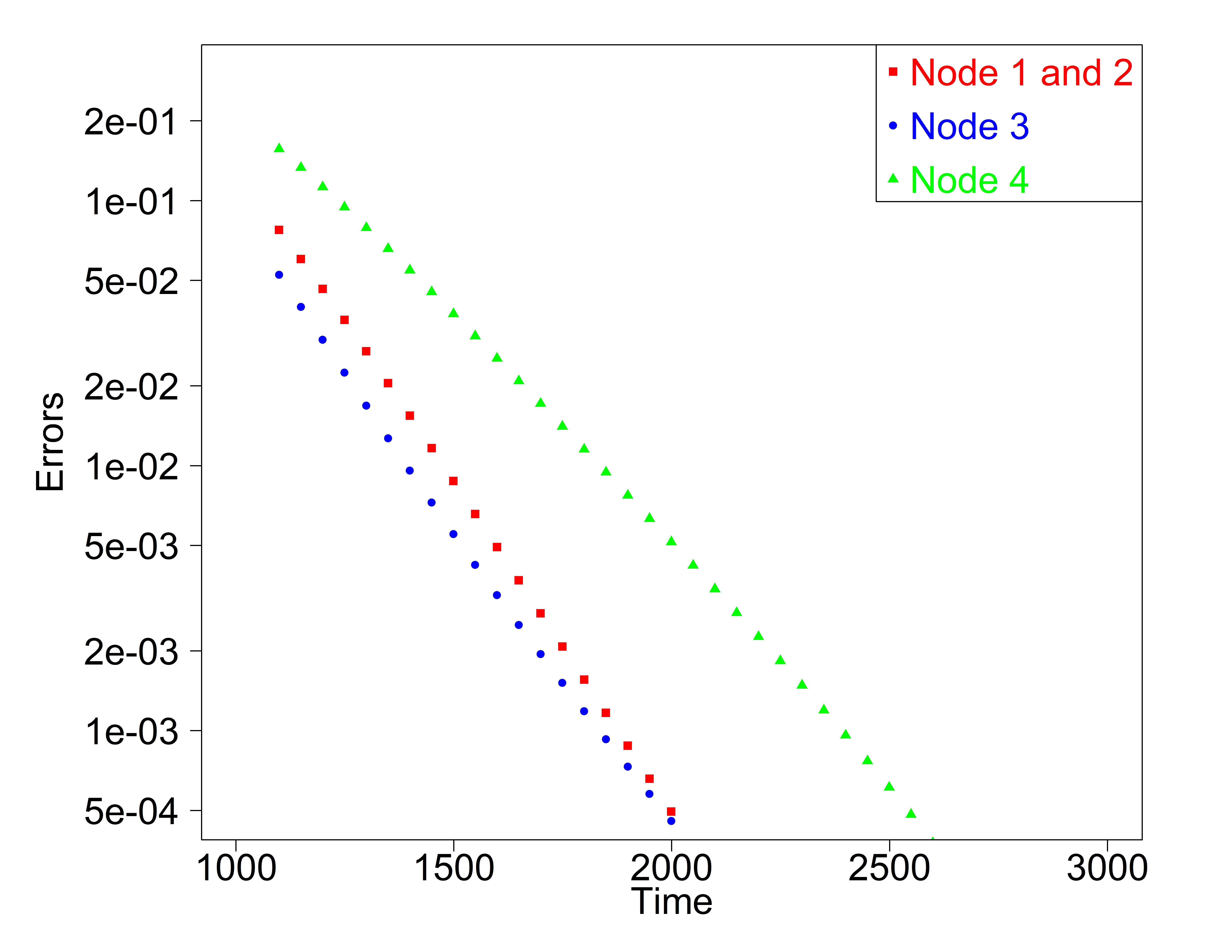}}
	\caption{Errors $\Delta z$ for the four nodes in the network example, for $p=0.25$, $\beta=0.004$ and $\gamma=0.001$.}
	\label{fig5} 
\end{figure}
\vspace{-2mm}
The stability of the solutions is in this case also guaranteed by the negative sign of the eigenvalues of the explicit Jacobian matrix ${\bf J}({\bf z})$, evaluated in the asymptotic solution.

\section{Numerical Experiments}
\label{Numerical Experiments}
In this section, we provide a numerical analysis to test the behavior of the proposed model. To this end, we consider three alternative classes of graphs: a random graph, based on \citet{ER1959}(ER) model (see also \citet{ER1960}), a small-world (SW) network, based on \citet{Watts1998}model and a \citet{Bar1999}(BA) model. Some sensitivity analyses have been explored by evaluating the effect of both the topological aspects of the network (as the number of nodes, density, etc.) and the parameters of the model.
	
We start focusing on ER graphs and testing the effect of the reinforcement factor $e\in [0,1]$, defined in Section \ref{Reinforcement factor}, on the diffusion process. To this end, we consider an ER model with $30$ nodes and an edge attachment probability of $0.5$.
In Fig. \ref{ER}, we show the average prevalence rates over time for both values of $\mathcal{R}$ above and below the threshold.
In Fig. \ref{ER}, panel (a), we observe that the more the factor $e$ tends to $1$, the faster the asymptotic level is reached. The ASIS model relies on the mutual reinforcement effect between the original network and the dual graph and this aspect can be noticed by the fact that, when $e=0$, the prevalence rates $x$ and $y$ tend to be farther apart than in the case of higher values of $e$. When $e=0$, we actually have two separate and independent SIS processes on the two networks.
In Fig. \ref{ER}, panel (b), we notice instead that, when $\mathcal{R}$ is below the epidemic threshold, for all values of $e$ the diffusion will die out and go to zero asymptotically. Differences between models seem smoothed in this case, although it is confirmed a slower convergence for the classical SIS model.
	
We now focus on the patterns of the prevalence rates obtained by applying an ASIS model with fixed parameters, and varying either the density (see Fig. \ref{ERd}) or the number of nodes (see Fig. \ref{ERn}) of the ER graphs. We notice that a higher density leads to a reduction of the heterogeneity between the prevalence of nodes. On the one hand, the structure of the network is more similar to the complete graph and hence the variability of prevalence rates between nodes is lower. On the other hand, a very fast convergence toward the steady state is observed. Vice versa, very sparse graphs lead to a higher heterogeneity between nodes as well as a lower convergence. In terms of the self-adaptive eigenvector centrality discussed in Section \ref{Self-adaptive eigenvector centrality}, this implies larger differences in the centralities of nodes and edges.
\begin{figure}[H]
		\centering
		\subfloat[]{\includegraphics[width=0.45\textwidth]{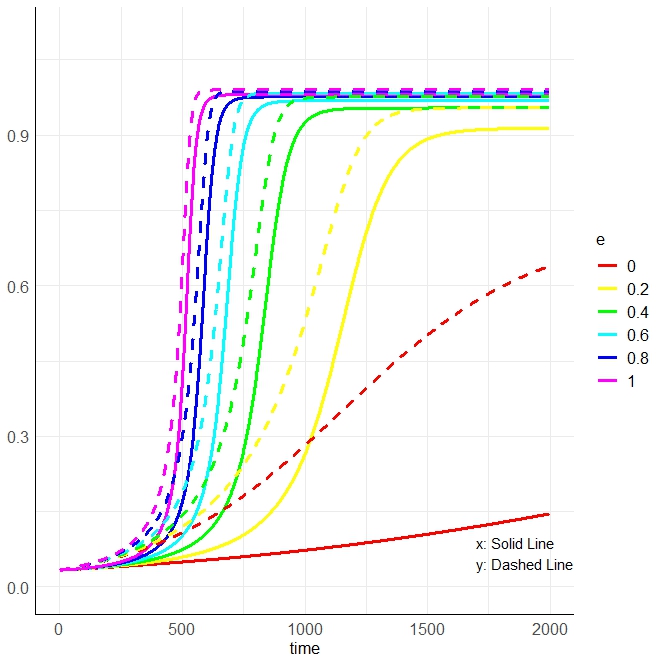}}\\
		\subfloat[]{\includegraphics[width=0.45\textwidth]{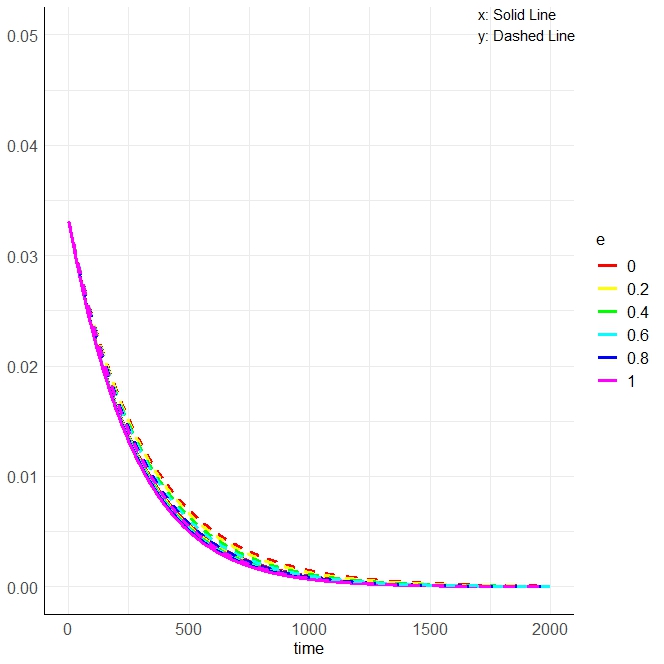}}
		\caption{Average prevalence rates in the ASIS model for different values of $e$. We consider a ER graph with $30$ nodes and density $0.5$ and we set $p=\frac{1}{30}$. Panel (a): $\beta=0.004$ and $\gamma=0.001$; panel (b): $\beta=0.001$ and $\gamma=0.004$.}
		\label{ER} 
	\end{figure}

\begin{figure}[H]
\centering
\subfloat[]{\includegraphics[width=0.45\textwidth]{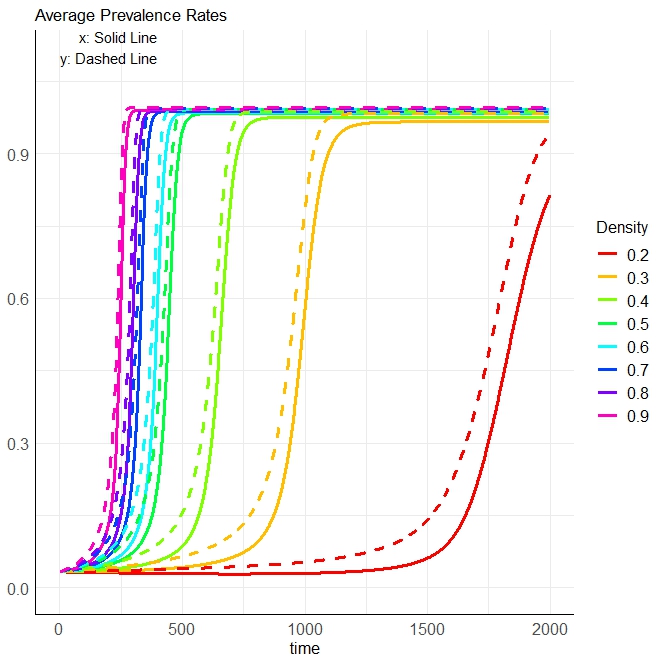}}\\
\subfloat[]{\includegraphics[width=0.45\textwidth]{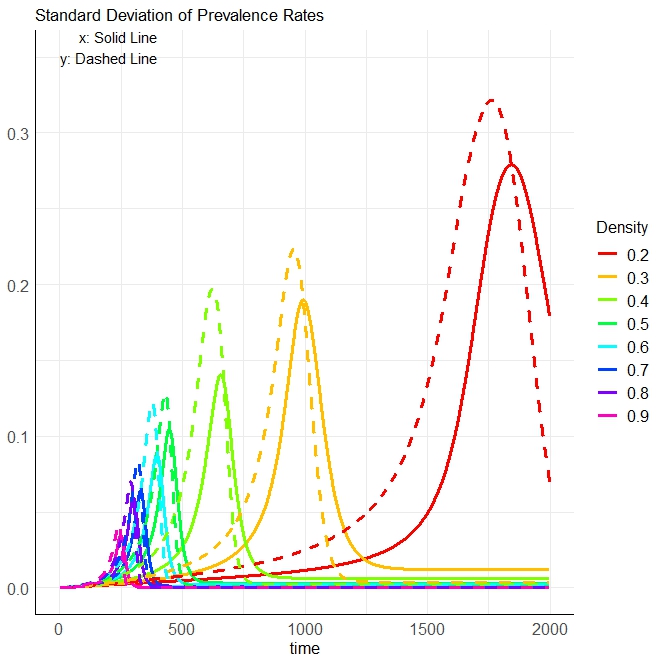}}
\caption{Average (panel (a)) and standard deviation (panel (b)) of prevalence rate distributions of the ASIS model obtained by considering a ER graph with 30 nodes and varying the density from 0.2 to 0.9 with steps of 0.1. We set $p=\frac{1}{30}$, $\beta=0.004$ and $\gamma=0.001$.}
\label{ERd} 
\end{figure}
	
\begin{figure}[H]
\centering
\subfloat[]{\includegraphics[width=0.45\textwidth]{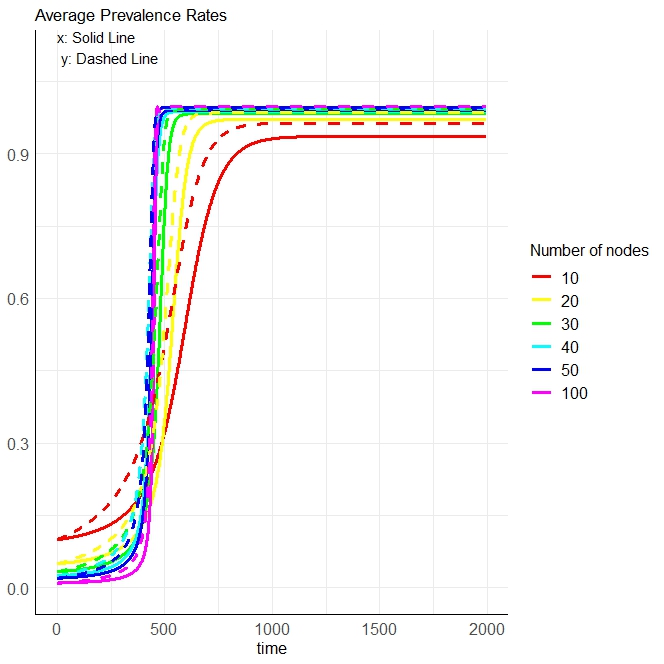}}\\
\subfloat[]{\includegraphics[width=0.45\textwidth]{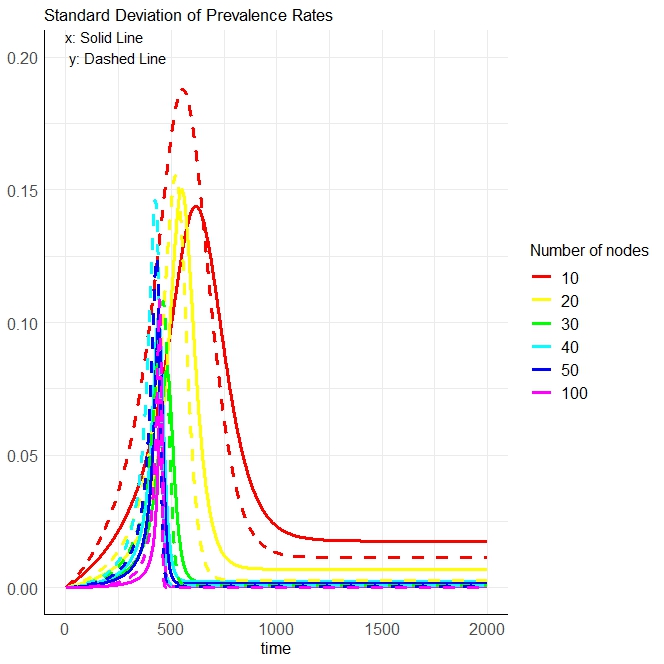}}
\caption{Average (panel (a)) and standard deviation (panel (b)) of prevalence rate distributions of the ASIS model obtained by considering a ER graph with a variable number of nodes and a density equal to 0.5. We set $p=\frac{1}{n}$, $\beta=0.004$ and $\gamma=0.001$.}
\label{ERn} 
\end{figure}

Moving on to consider the effect of the number of nodes (see Fig. \ref{ERn}) and assuming one node infected at the beginning of the process (i.e. $p=\frac{1}{n}$), we observe a slower propagation for smaller networks. In this case, the networks have a similar density, but when the order of the graph is higher, although a lower probability is observed at the beginning, the spreading dynamics increases and reaches the endemic steady state faster. It is also noteworthy that the size of the variability of the prevalence rates is not affected by the number of nodes. Indeed, in Fig. \ref{ERn}, panel (b), we observe that the heterogeneity between nodes and edges is similar, but the curve is shifted forward in time due to a slower process for smaller graphs.
	
We now focus on the ASIS diffusion on different graph models. In Fig. \ref{ERt}, we provide a comparison of the prevalence rates for diffusion processes above the threshold in the ER, SW and BA models.
According to the mean prevalence, we do not observe great differences between the models. On average, when the network has the same number of nodes and edges, the patterns are similar with a slightly lower endemic steady state for the BA model. However, the topological characteristics of the BA graph are caught in terms of a greater heterogeneity between nodes and edges (see Fig. \ref{ERt}, panel (b)). Indeed, a higher volatility among prevalence rates is noticeable for this model. This can be explained by the fact that the BA graph follows a power-law degree distribution, having few nodes with a significantly higher number of connections, while the majority of nodes have only a few connections. As a consequence, the BA graph is favorable for information cascades due to its scale-free nature. Influential nodes have indeed a higher chance of triggering large-scale information cascades. This means that information can  propagate quickly through the network, leading to widespread adoption or dissemination, and large differences between nodes are observed at the steady state. The differences in terms of variability between ER and SW graphs are less relevant, although we notice that the heterogeneity is a bit larger for SW. In SW graphs, the presence of strong local clustering and short path lengths allows for rapid containment of outbreaks within specific clusters, limiting the overall spread. However, when the infection bridges different clusters through long-range connections, it leads to a larger-scale epidemic. The ER graphs are more susceptible to disease spread due to the lack of strong clustering and more random connectivity. The absence of localized clusters hinders the containment of outbreaks and information spread more uniformly across the network providing a greater homogeneity between nodes.

\begin{figure}[H]
\centering
\subfloat[]{\includegraphics[width=0.45\textwidth]{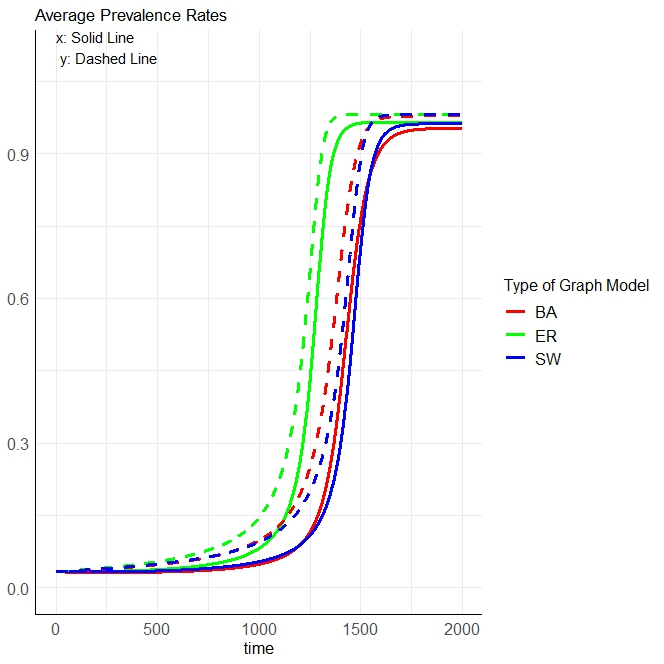}}\\
\subfloat[]{\includegraphics[width=0.45\textwidth]{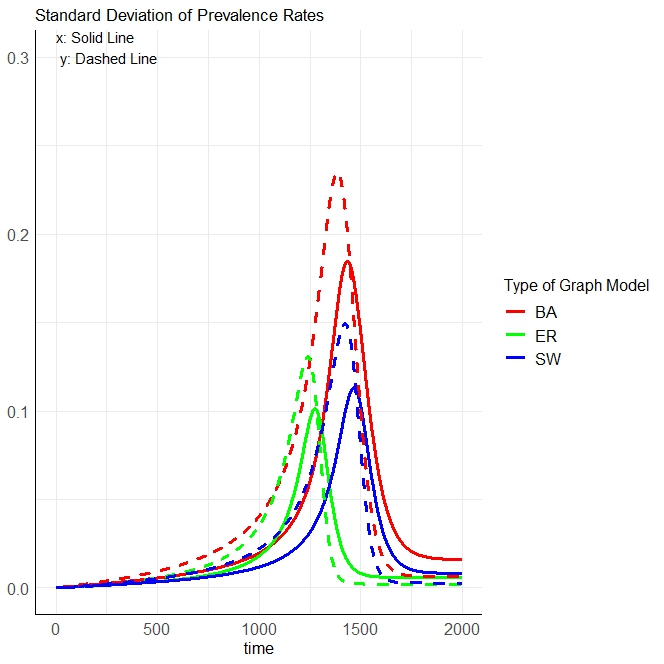}}
\caption{Average (panel (a)) and standard deviation (panel (b)) of prevalence rate distributions of the ASIS model on different graph models (ER, SW and BA, respectively). In all cases, networks have been generated considering $30$ nodes and a density equal to $0.5$. The same parameters have been used for all models: $p=\frac{1}{n}$, $\beta=0.002$ and $\gamma=0.001$.}
\label{ERt} 
\end{figure}
	
	
\clearpage
	
\section{Social reinforcement in lattice and random networks}
\label{Social reinforcement in lattice and random networks}

While epidemic models are employed to analyze the spread of opinions and behaviors, it is important to note that the dynamics of epidemic diffusion and of information dissemination differ in several key aspects.

The propagation of epidemics is primarily driven by biological factors such as transmission rates and incubation periods, while the dissemination of information or behavior is influenced by social and psychological factors such as individual beliefs, social status, and personal interests.

Moreover, the spread of disease requires physical contact, and in the absence of any policy, individuals usually exhibit passive behavior. Conversely, the dissemination of information, particularly in contemporary times, predominantly happens through online connections in addition to traditional face-to-face communication. In this context, individuals tend to take a more active role in making decisions, such as approving or disapproving behaviors.

By recognizing these differences, we can accurately capture the dynamics of each type of spread and devise effective interventions, including strategies to counteract disease propagation and misinformation.

Hence, it has been frequently emphasized in the literature that SIS and SIR models fail to explain the spread of information or behavior, for example in online social networks. \cite{Fortunato2009,Moro2009,Moro2011}

For instance, in 2010, Centola conducted an influential experiment on the spread of behavior in online social networks.\cite{Centola2010} The experiment showed the critical role that social reinforcement plays in the online spread of behavior. Social reinforcement refers to the typical condition in which an individual requires multiple prompts from peers before adopting an opinion or behavior. \cite{Peyton2009,Onnela2010} Indeed, the experiment showed that a single signal has a very weak effect on individuals' decision making, while redundant signals can increase the probability of approval and behavior adoption.

Specifically, among the six networks analyzed by Centola, three were regular networks and three were random networks of the same size and average degree. The primary outcome of the experiment challenges the prior assumption that random networks are more conductive to the propagation of behavior when compared to regular networks. In fact, behavior spreads faster and to a greater extent in highly clustered regular networks than in random ones, because in the former individuals receive more redundant signals.

This evidence prompts the search for diffusion models that involve reinforcing communication action between nodes.

In particular, Zheng \textit{et al.}\cite{Zheng2013} developed an interesting model where the primary diffusion rate $\beta$, representing the probability that an individual will adopt the behavior after receiving the information for the first time, incorporates the strength of social reinforcement. This reinforcement factor considers how many times an individual receives a specific piece of information. This aspect is particularly relevant in online social networks, where connections are often weaker compared to face-to-face communication. Their findings align with the online behavioral diffusion experiment. In fact, their model confirmed that when $\beta$ takes small/medium values, social reinforcement has an effect on the spreading process, and spreading is faster and further in regular networks than in random ones. For a large primary spreading rate, an individual who receives information about her neighbor's behavior for the first time has a higher probability to adopt it and to take the same action, so that the factor of social reinforcement becomes less influential.

The ASIS model proposed in this paper automatically incorporates a reinforcing action in the communication between nodes. It achieves this by continuously updating the edge weight based on the actual probability that a given node has adopted a behavior or opinion. In particular, it avoids the need to define extrinsic rules to update the infection rate values. Due to these characteristics, it is well suited to provide an accurate description of diffusion phenomena of the nature described above.  

In fact, it incorporates, in addition to the $\beta$ infection and $\gamma$ recovery parameters, the reinforcement factor $e$ defined in Section \ref{Reinforcement factor}. This parameter represents the intensity of the reinforcement action in the communication between nodes. When $e=0$, there is no reinforcement and the model is suitable for describing disease propagation (SIS model). When $e\neq 0$, it includes such a reinforcement and is suitable for describing the information dissemination, for example, in online social networks (ASIS model). As $e$ grows from $0$ to $1$ the intensity of the reinforcement grows accordingly.

Hence, we tested the hypothesis that the interaction reinforcement introduced in our model and measured by $e$ may favor information dissemination within a regular social network compared to a random network. To ensure computational efficiency, we performed a variety of numerical simulations on moderately sized binary networks. We defer to a subsequent dedicated paper the detailed analysis of an extended real-world network. In general, the numerical evidence supports the results of Centola's experiment and aligns with the model by Zheng \textit{et al.}.

We perform here a comparative analysis between random networks, specifically Erd\H{o}s-R\'{e}nyi networks, and various types of regular networks: square lattice with von-Neumann neighborhood, square lattice with Moore neighborhood, cycle and regular network with degree 3. In each pair of graphs under comparison, we maintained an identical number of nodes $n$ and the same density $\delta$ (and consequently the same size $m$). Results on random networks were averaged over $100$ different instances of Erd\H{o}s-R\'{e}nyi networks of type $G(n,m)$ with the same parameters.

In Fig. \ref{fig11}, we illustrate the behavior of the square lattice with $n=25$ and the corresponding random graph in the parameter space $(\beta,e)$, where $\beta$ is the infection rate and $e$ is the adaptive parameter of our model. In this simulation we specifically assumed: $\delta=0.1333$, $p=5/25$, $\gamma=0.02$, $0.02\leq \beta \leq 0.12$ and $1\leq t\leq 400$. Panels (a) and (b) show the contour plot of steady state probabilities in the plane $(\beta,e)$, averaged over the nodes in the network. We call these averaged values $X^{\star}_{\rm lattice}$ and $X^{\star}_{\rm random}$ for the two kinds of network. In all the contour plots, blue represents low probability values while red represents high values. The horizontal slice at $e=0$ represents the standard SIS model, while the horizontal slice at $e=1$ represents the fully adaptive SIS model. As $e$ increases from $0$ to $1$ (with step $0.1$), the reinforcement effect in the social interactions increases. In general, in both panels (a) for the lattice and (b) for the random network, in order to have the same asymptotic probability as $\beta$ increases, a lower reinforcement effect $e$ is sufficient. For a fixed value of $\beta$, the asymptotic probability increases with $e$. In Fig. \ref{fig11}, panel (c), we plot the difference $X^{\star}_{\rm lattice}-X^{\star}_{\rm random}$ between the asymptotic probability values on the lattice network and the random network. As observed, this difference can take on positive and negative values in the plane, contingent upon whether diffusion predominates in the lattice model (positive values, red in the plot) or in the random model (negative values, blue in the plot). Let us consider, for instance, the value $\beta=0.04$: as $e$ increases from $0$ to $1$, the difference increases by approximately $0.1$, thus showing a $10\%$ higher probability of diffusion in the lattice model compared to the random model with the introduction of the reinforcement effect. The difference proves to be significant within a specific range of small to medium values of $\beta$, up to about $0.06$. The standard SIS model ($e=0$) shows the widest interval in which the difference is negative and, therefore, the spreading range is much larger in random networks than in regular networks. When $\beta$ is beyond a certain value, the intensity of the infection process becomes such that it levels out any difference. When the reinforcement grows, the spreading range in the regular network tends to be greater than that in random network over a wider range of $\beta$. This evidence can be further confirmed by observing panel (d) in Fig. \ref{fig11}. In this panel we depicted the time evolution of the prevalence curves for a fixed value $\beta=0.04$ and for different values of the reinforcement parameter. The solid lines represent the lattice network, the dashed lines the random one. All solid and dashed curves are coupled with the same color. The color refers to the value of $e$, from the red one in the bottom ($e=0$) to the violet one in the top ($e=1$). As can be seen, for lower values of $e$ the dashed lines end up above the solid lines and the spreading is higher in the random network than in the regular one. Conversely, for higher values of $e$, the solid lines end up above the dashed one, indicating the dominance of the spreading process in the lattice network over the random one.

For $\beta$ values approximately above $0.06$, the process enters the overactive region where the infection rate is large enough that the social reinforcement strength does not affect the spreading range. In this case, the regular network appears to foster diffusion better than the random one.


Let us now explore the dependence of this behavior on network size and density.  Fig. \ref{fig12}, panels (a)-(d), replicates the aforementioned observations on a larger square grid ($n=64$) with a lower density $\delta=0.0556$. In this case, we assumed: $p=8/64$, $\gamma=0.02$,  $0.02\leq \beta \leq 0.10$ and $1\leq t\leq 400$. In particular, panel (c) confirms the earlier findings with some distinctions. There is a well-identified region in which diffusion on a random network dominates, for $\beta$ values below approximately 0.06. Again, by increasing the parameter $e$ with a fixed value of $\beta$, we transition toward regions where diffusion on regular networks dominates that on random networks, but it is notable that higher values of the parameter are needed for this shift. Essentially, as the density of the network decreases, a greater reinforcement parameter is necessary to transition from one regime to another. The dominance of diffusion on a random network at low values of the parameter is confirmed by cases where diffusion on a lattice leads to extinction while diffusion on a random network reaches a stationary state, as depicted in panel (d) for $\beta=0.04$. In general, our conclusion is consistent with Centola's expectation that, in sparse networks, reducing the network density can narrow the difference in $X^{\star}$ between regular and random networks.

\begin{figure}[H]
			\centering
			\subfloat[]{\includegraphics[width=0.30\textwidth]{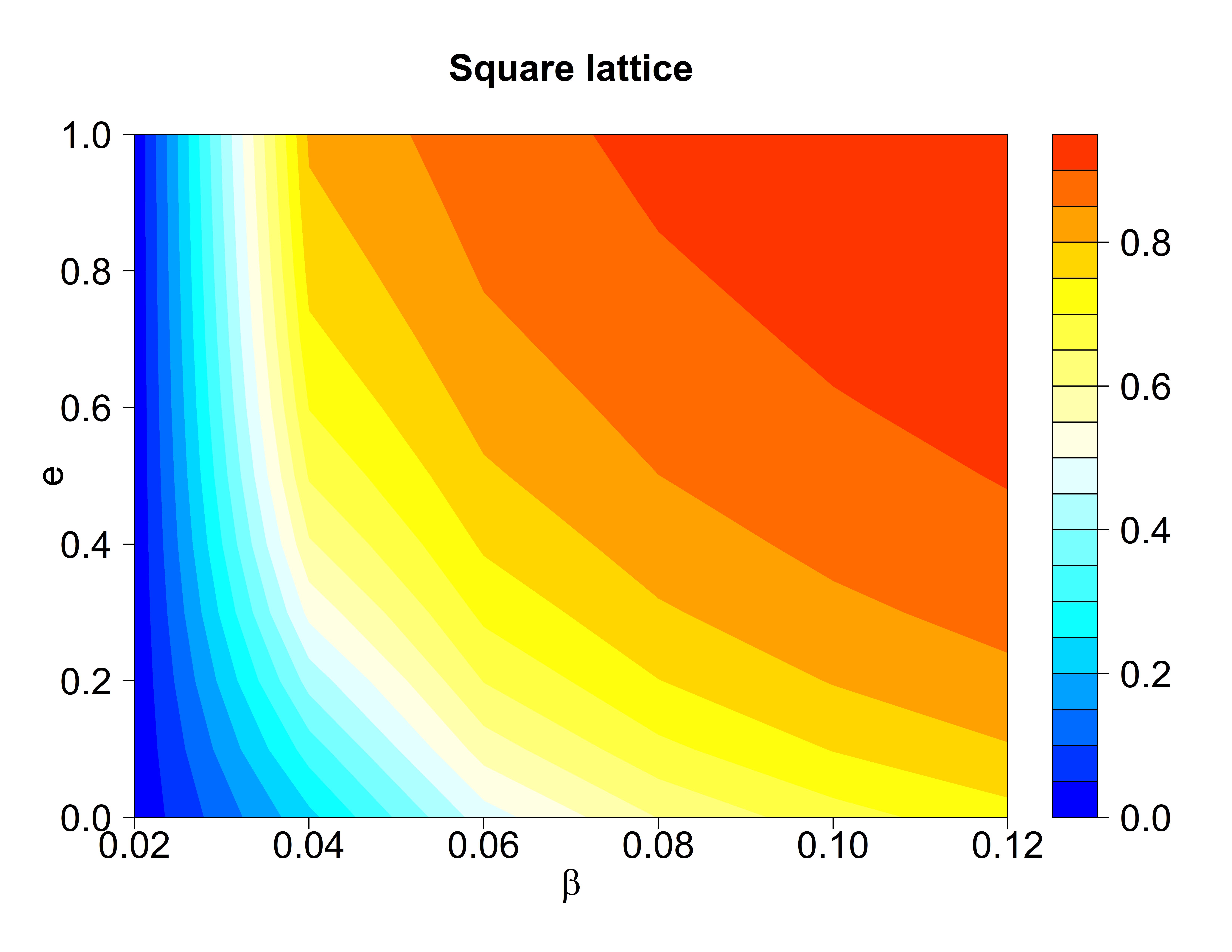}}\\ \vspace{-4mm}
			\subfloat[]{\includegraphics[width=0.30\textwidth]{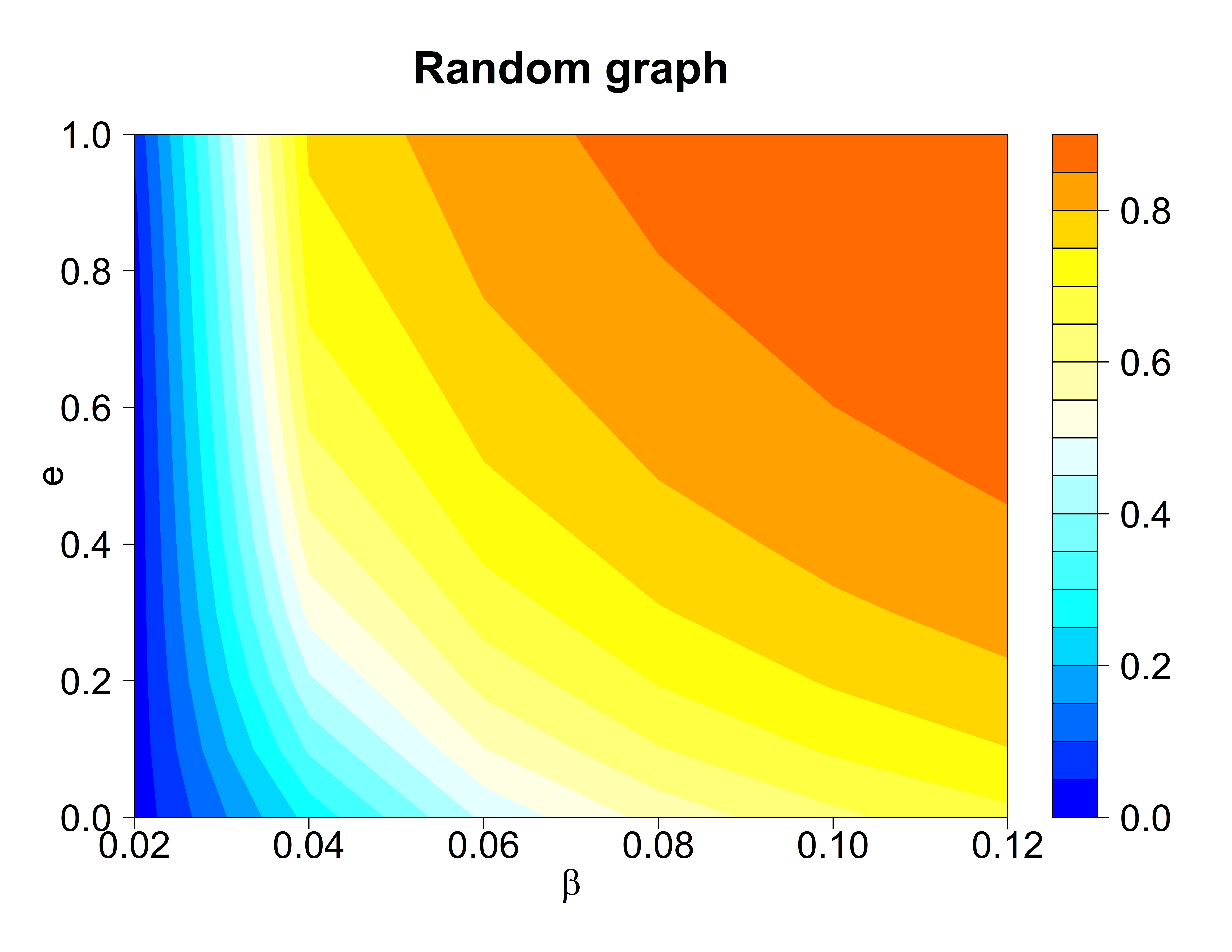}}\\ \vspace{-4mm}
			\subfloat[]{\includegraphics[width=0.30\textwidth]{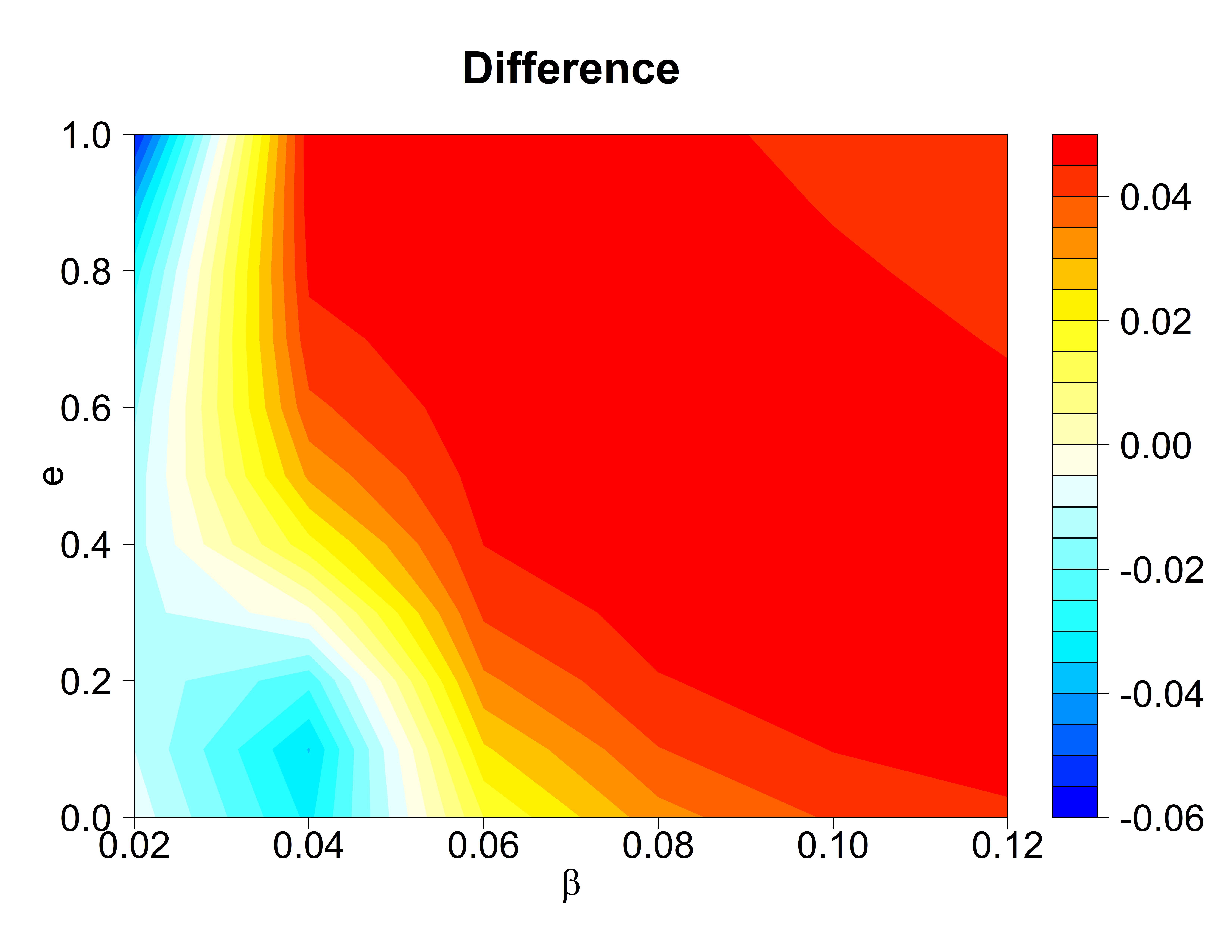}}\\ \vspace{-4mm}
			\subfloat[]{\includegraphics[width=0.30\textwidth]{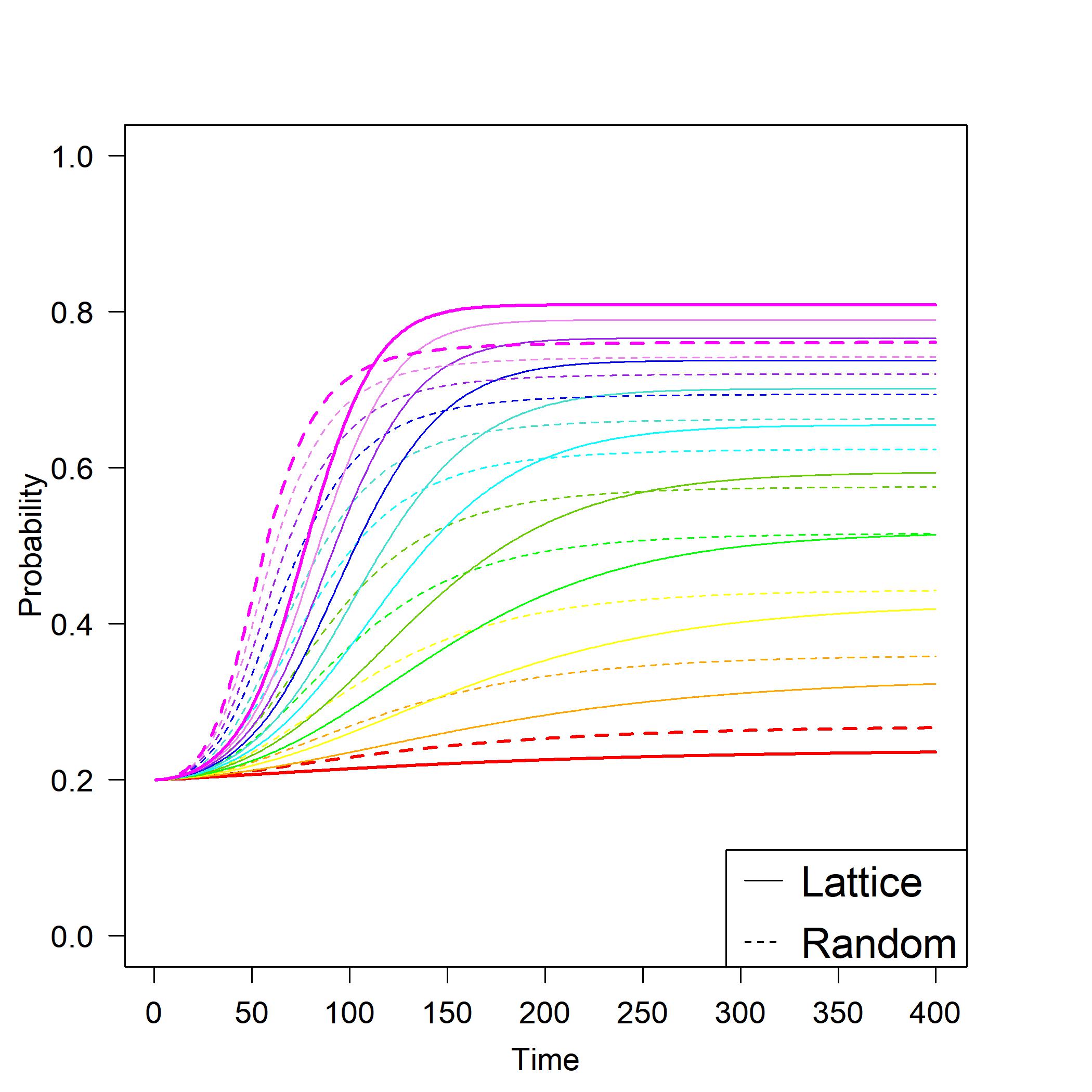}}
			\caption{Contour plot of the steady state probabilities in the parameter space $(\beta, e)$: in panel (a) $X^{\star}_{\rm lattice}$ for the square lattice, and in panel (b) $X^{\star}_{\rm random}$ for the random network with $n=25$, $\delta=0.1333$, $p=5/25$, $\gamma=0.02$. In panel (c) contour plot of the difference $X^{\star}_{\rm lattice}-X^{\star}_{\rm random}$. In panel (d) time evolution of the mean prevalence for $e$ in $[0,1]$. See the text for detailed explanation.}
			\label{fig11} 
\end{figure}

\begin{figure}[H]
	\centering
	\subfloat[]{\includegraphics[width=0.30\textwidth]{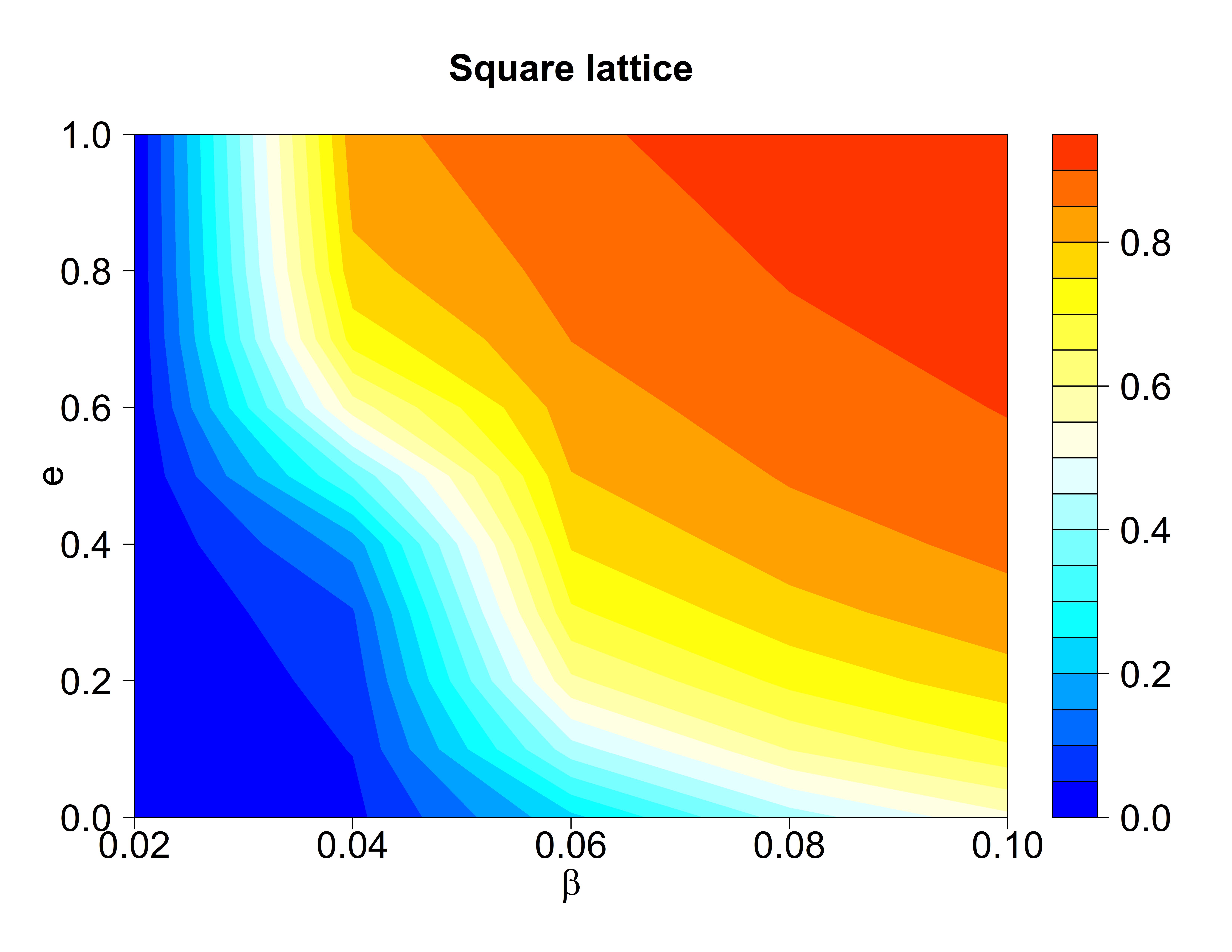}}\\ \vspace{-4mm}
	\subfloat[]{\includegraphics[width=0.30\textwidth]{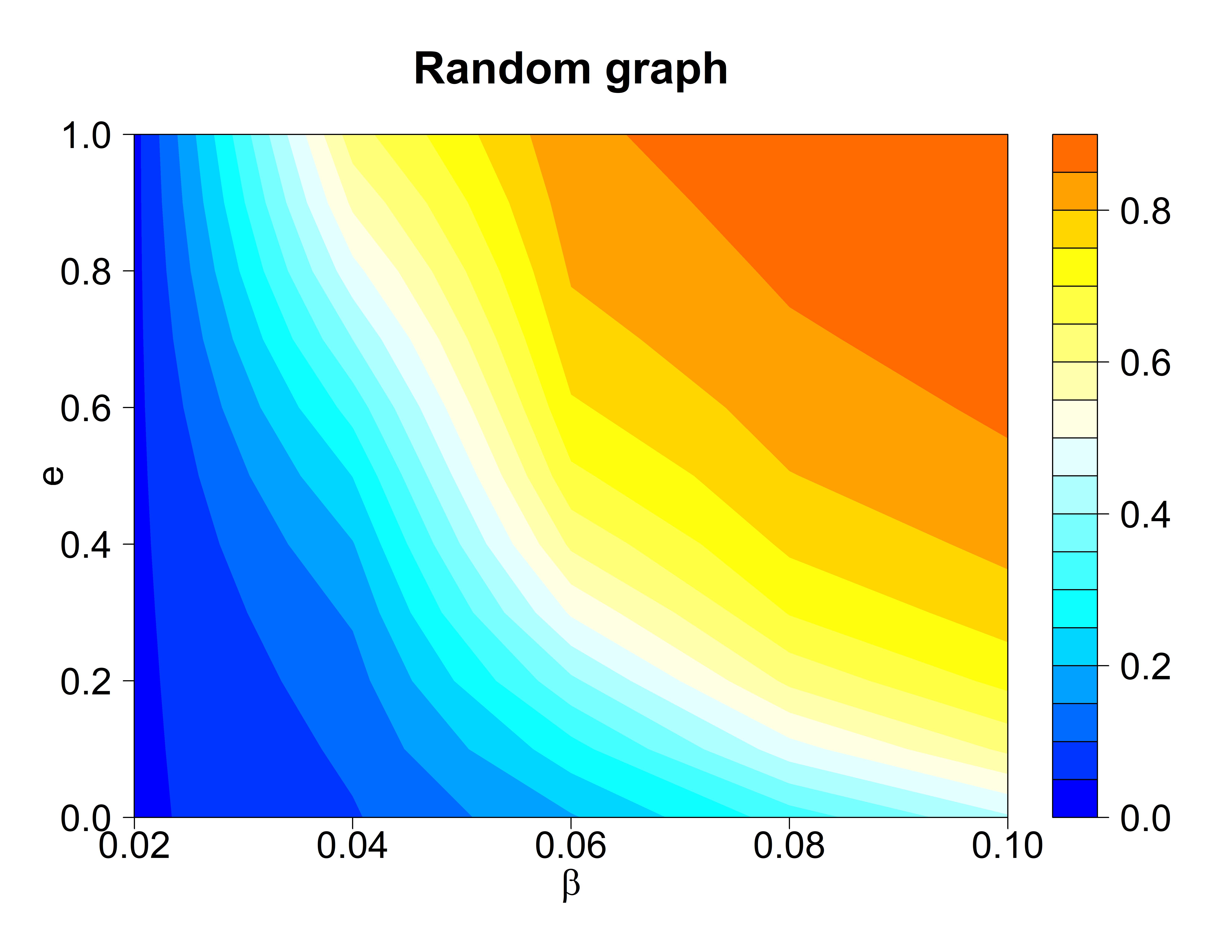}}\\ \vspace{-4mm}
	\subfloat[]{\includegraphics[width=0.30\textwidth]{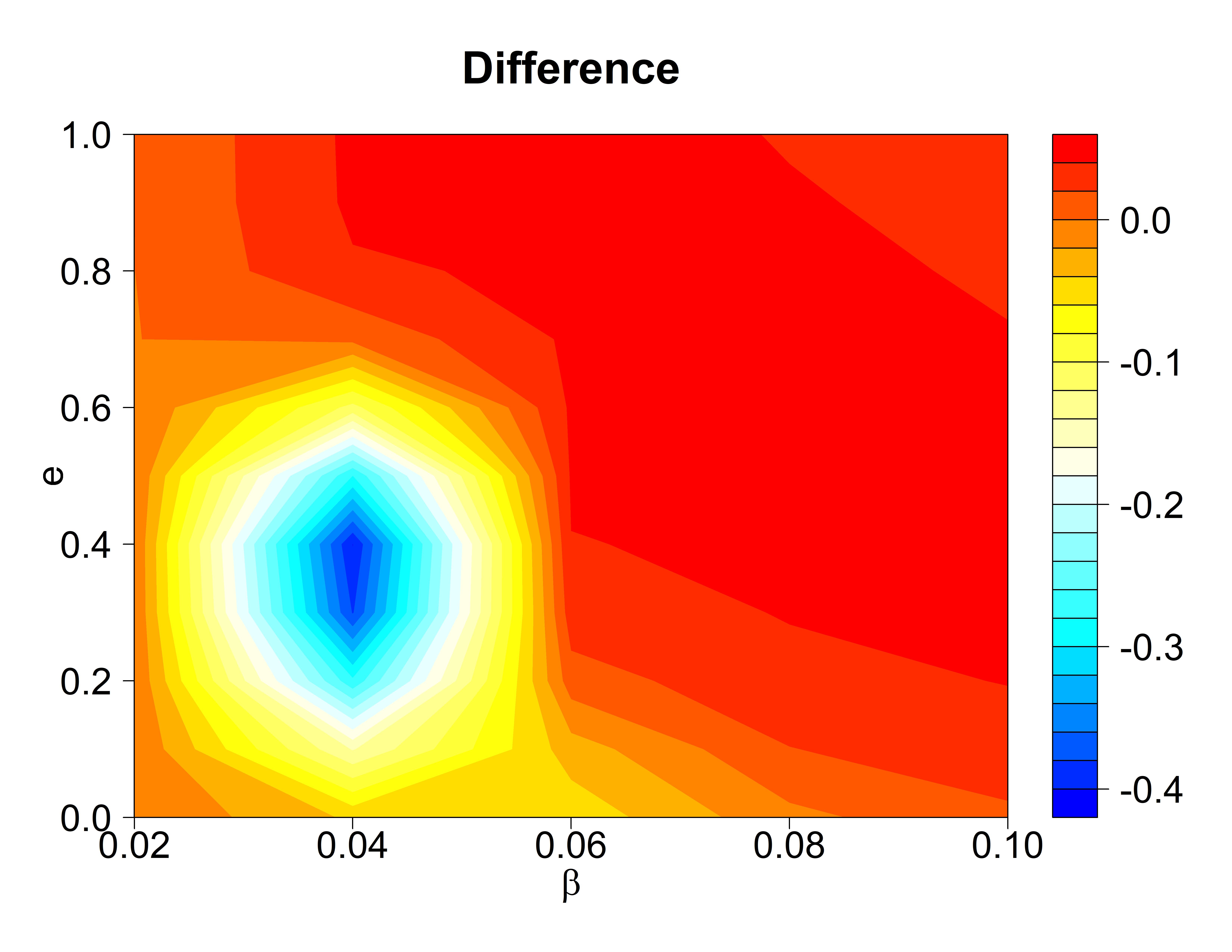}}\\ \vspace{-4mm}
	\subfloat[]{\includegraphics[width=0.30\textwidth]{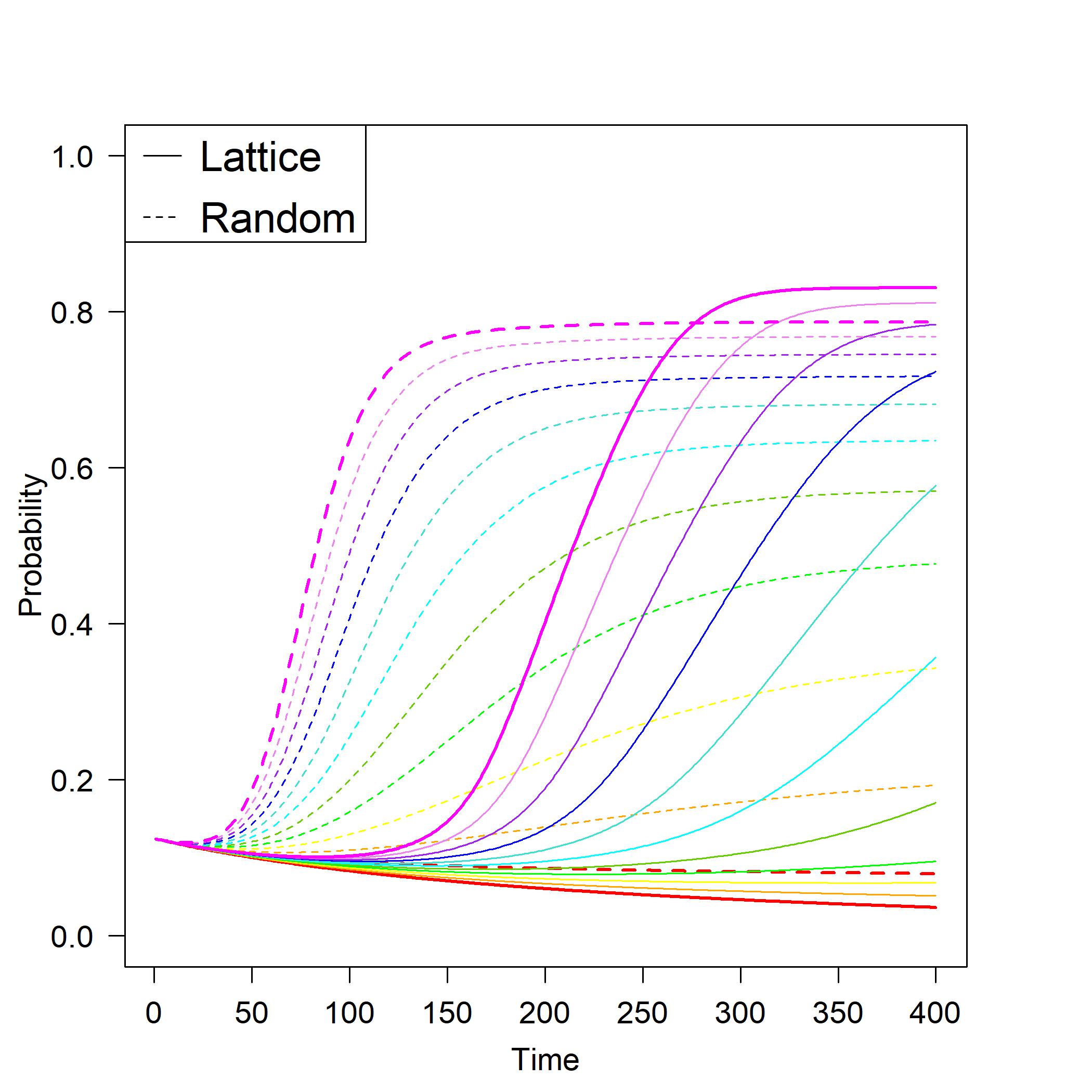}}
	\caption{Contour plot of the steady state probabilities in the parameter space $(\beta, e)$: in panel (a) $X^{\star}_{\rm lattice}$ for the square lattice, and in panel (b) $X^{\star}_{\rm random}$ for the random network with $n=64$, $\delta=0.0556$, $p=8/64$, $\gamma=0.02$. In panel (c) contour plot of the difference $X^{\star}_{\rm lattice}-X^{\star}_{\rm random}$. In panel (d) time evolution of the mean prevalence for $e$ in $[0,1]$. See the text for detailed explanation.}
	\label{fig12} 
\end{figure}

The analysis has been then repeated on a Moore lattice. In this case, we obtain an interesting temporal dynamics of the dominance of one regime over the other. Fig. \ref{fig13} shows four different time snapshots of the difference $X^{\star}_{\rm lattice}-X^{\star}_{\rm random}$ at times $t=100$, $t=200$, $t=300$ and $t=400$, the latter representing the time at which the steady state is reached for all values of $e$.

Fig. \ref{fig13}, panel (a), shows that in the early stages of the epidemic there is a well-defined island for low values of $\beta$, where the reinforcement exerts a strong pulling effect, leading to an increase of the spread on regular networks. Achieving the same effect over time only requires a gradually decreasing value of $e$. On this type of network and with these sizes, the adaptivity of the model appears to prevail in the initial phases of the process, while in the final phases the absolute values of the difference seem to decrease, while maintaining the reinforcement effect in favor of the regular network model. These results are in line with those obtained in the model by Zheng \textit{et al.}\cite{Zheng2013}.

Finally, we tested the difference between the behavior of regular graphs, i.e. cycles and regular graphs with constant degree equal to 3, and comparable random networks with the same number of nodes and density. The results are represented in Fig. \ref{fig14}. Panels (a) and (b) refer to a cycle with $n=20$ nodes, $\delta=0.1053$, $p=4/20$ and $\gamma=0.01$.  Panels (c) and (d) refer to a 3-regular graph with $n=20$ nodes, $\delta=0.1579$, $p=4/20$ and $\gamma=0.01$.
In this figure, the contour plots refer to the asymptotic values only, at $t=400$ and $t=800$ respectively. The effect of reinforcement on the setting of the regime in the two types of networks emerges strongly. It is observed very clearly that at low values of infectivity, as the parameter $e$ increases, the regular graph exhibits a higher spread of infection than the comparable random graph.

\begin{figure}[H]
	\centering
	\subfloat[]{\includegraphics[width=0.30\textwidth]{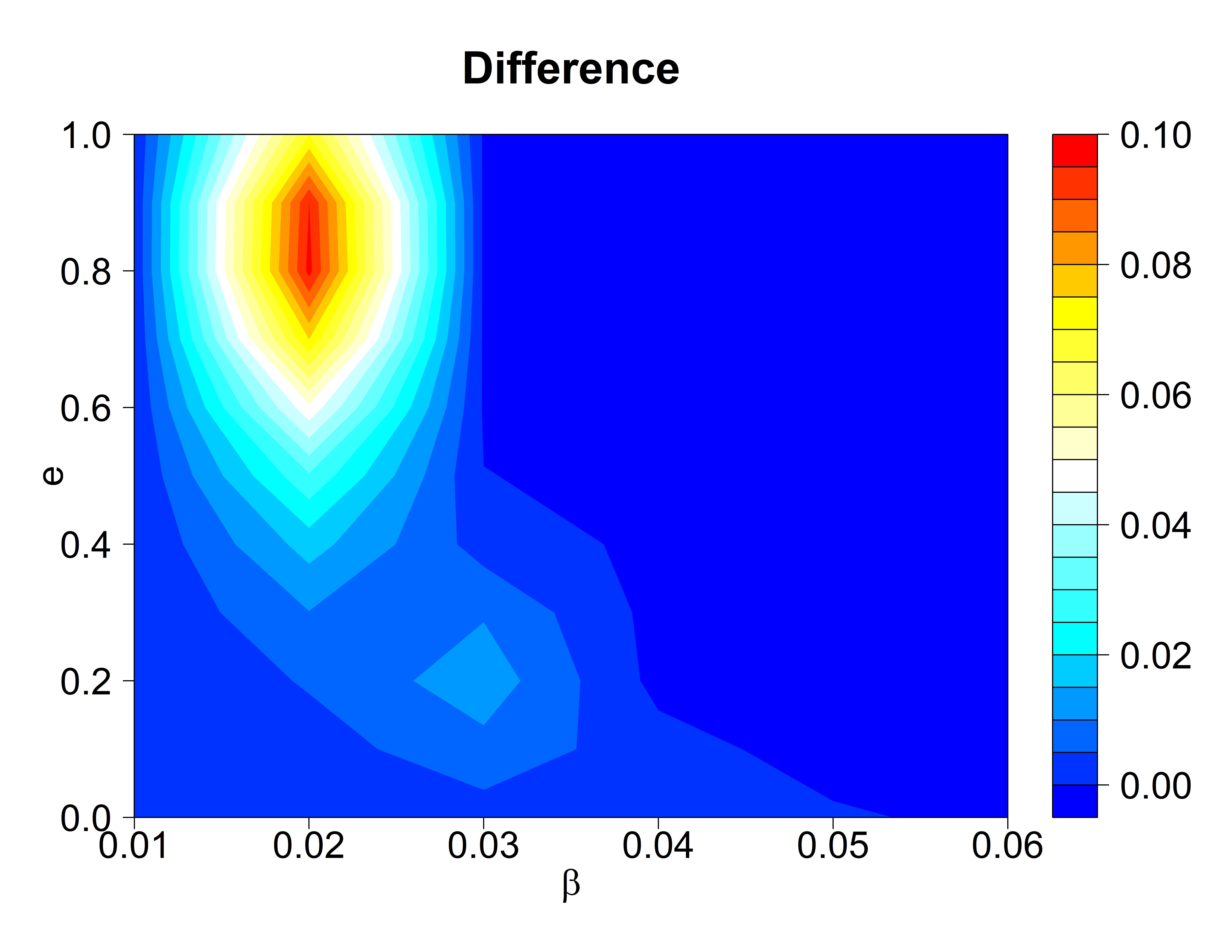}}\\ \vspace{-4mm}
	\subfloat[]{\includegraphics[width=0.30\textwidth]{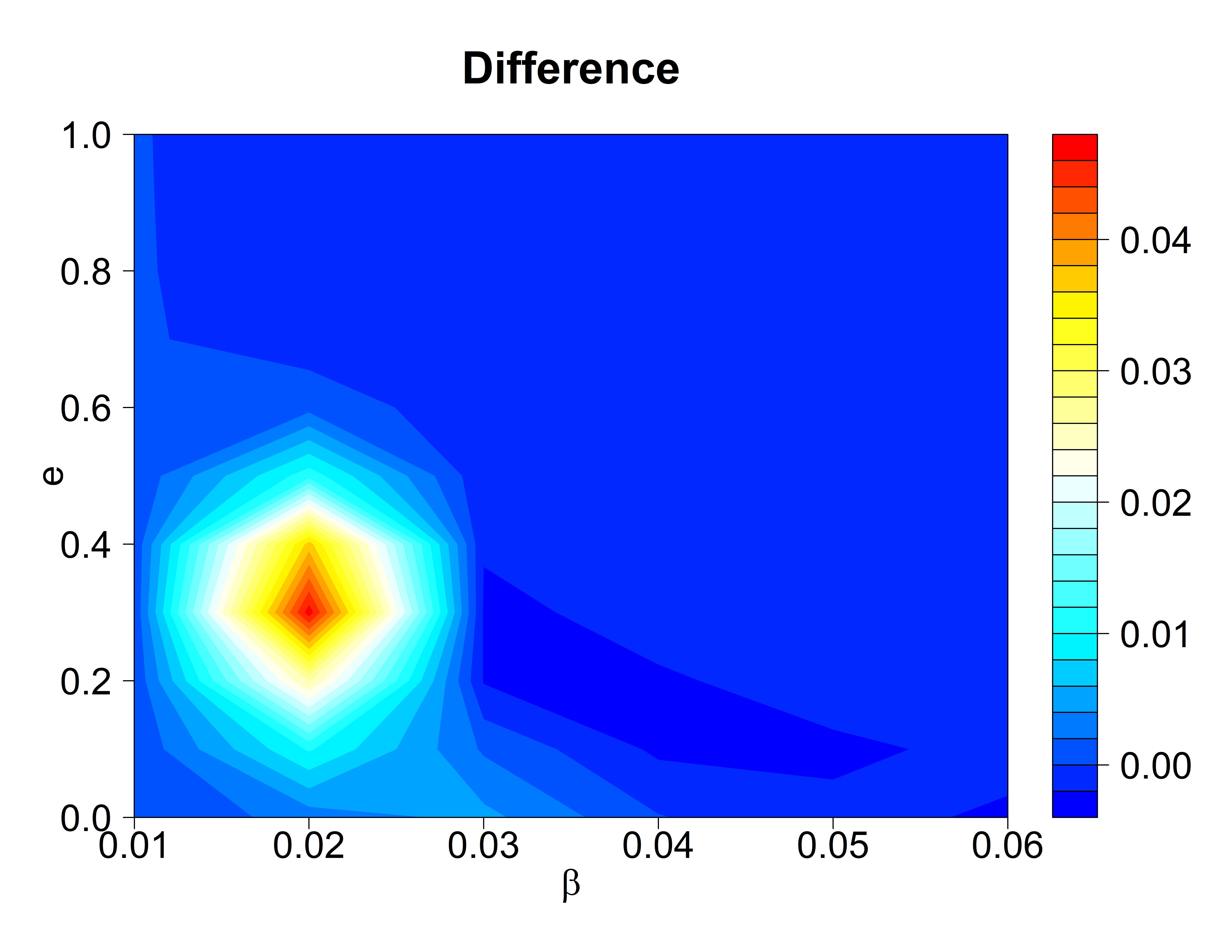}}\\ \vspace{-4mm}
	\subfloat[]{\includegraphics[width=0.30\textwidth]{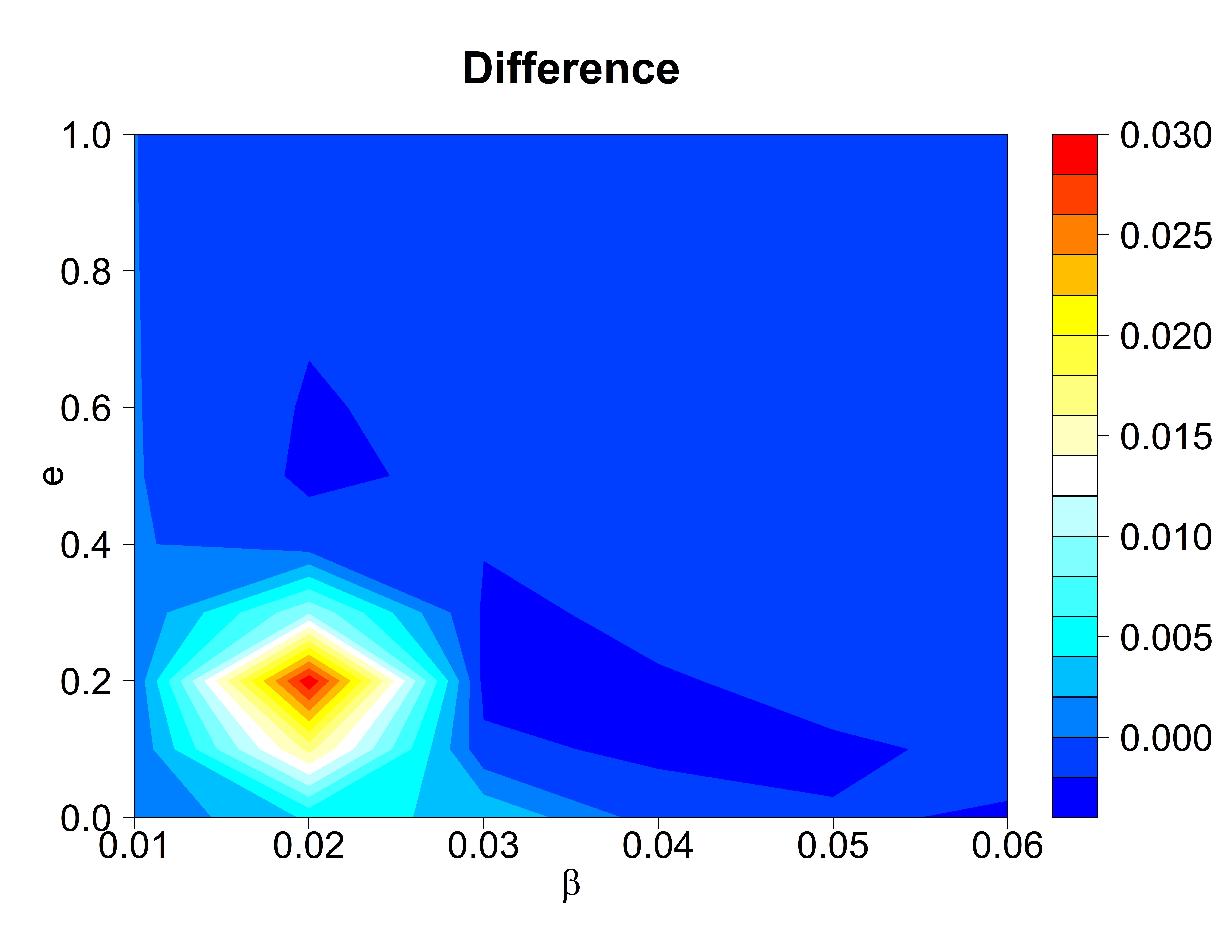}}\\ \vspace{-4mm}
	\subfloat[]{\includegraphics[width=0.30\textwidth]{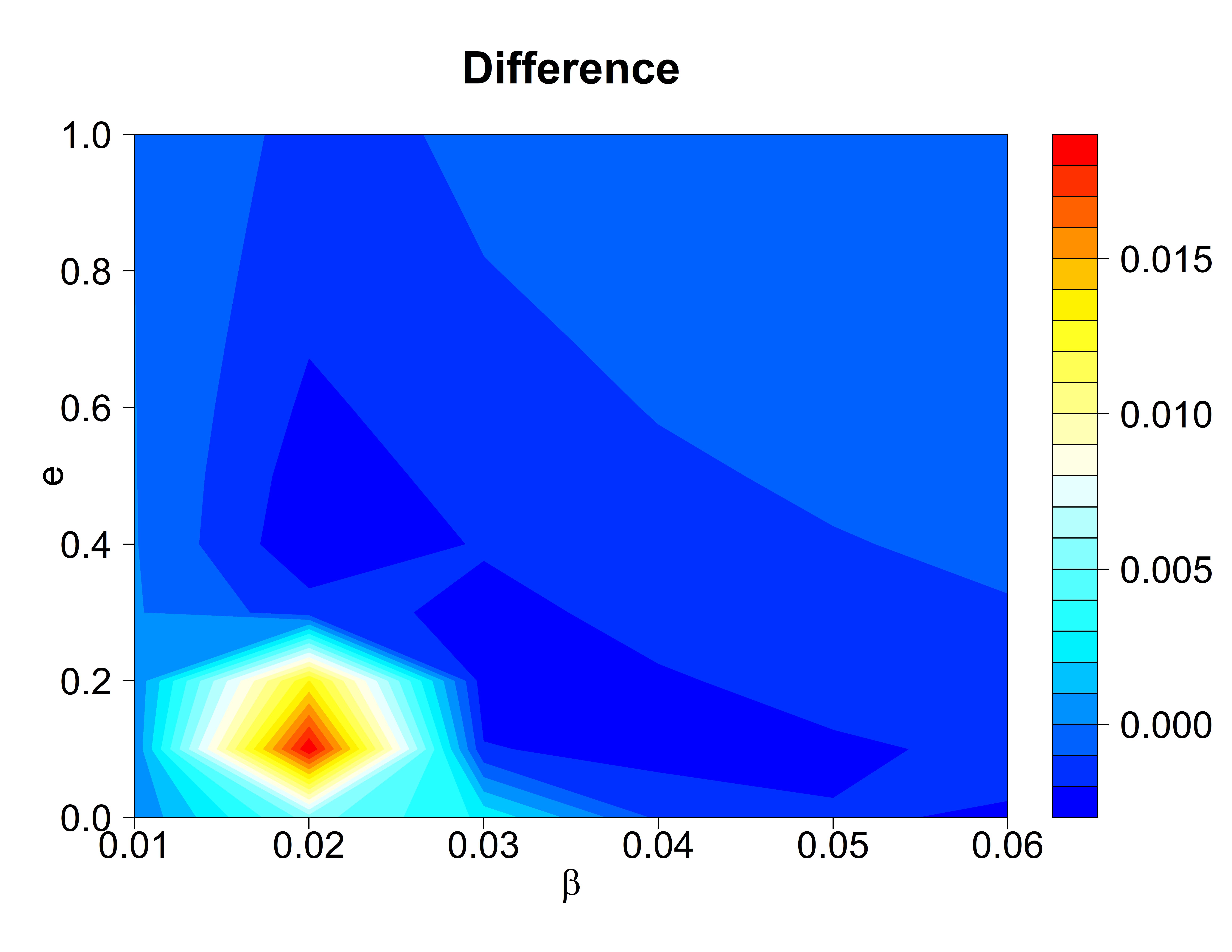}}
	\caption{Four snapshots of the evolution the contour plot of the difference $X^{\star}_{\rm lattice}-X^{\star}_{\rm random}$ on a Moore lattice with $n=25$, $\delta=0.34$, $p=3/25$, $\gamma=0.02$ and $t$ equal to (a) $100$, (b) $200$, (c) $300$, (d) $400$.}
	\label{fig13} 
\end{figure}

\begin{figure}[H]
	\centering
	\subfloat[]{\includegraphics[width=0.30\textwidth]{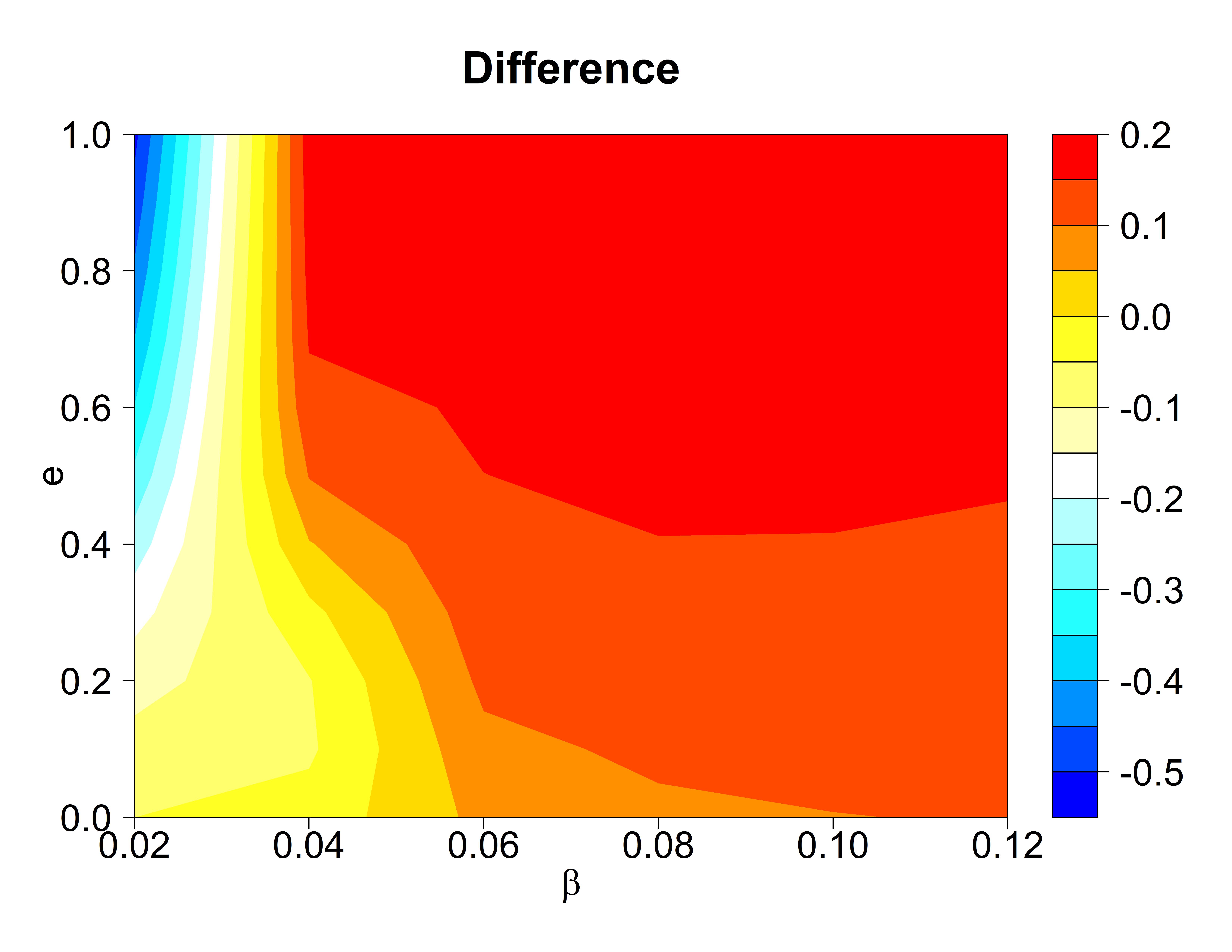}}\\ \vspace{-4.5mm}
	\subfloat[]{\includegraphics[width=0.30\textwidth]{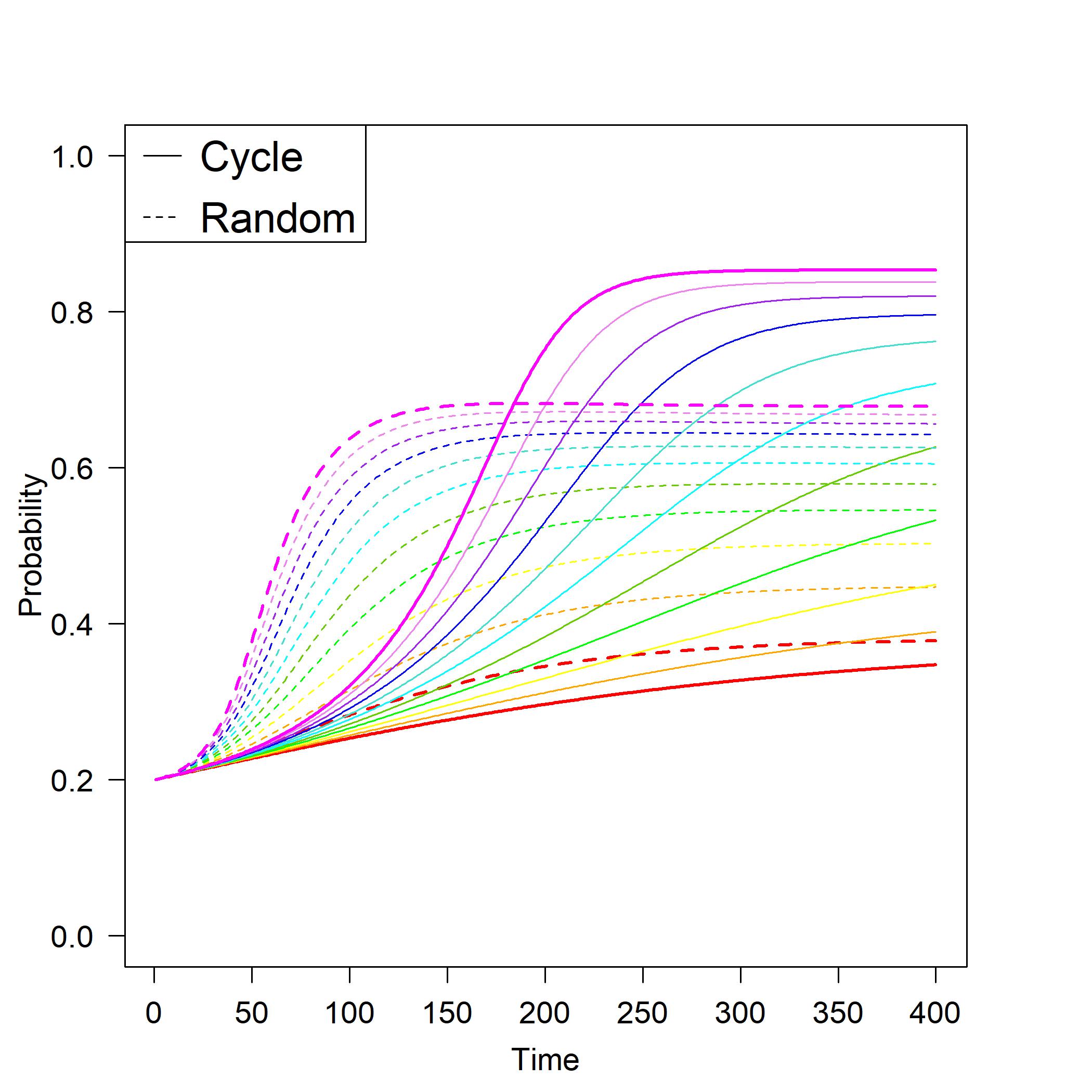}}
	\\ \vspace{-4mm}
	\subfloat[]{\includegraphics[width=0.30\textwidth]{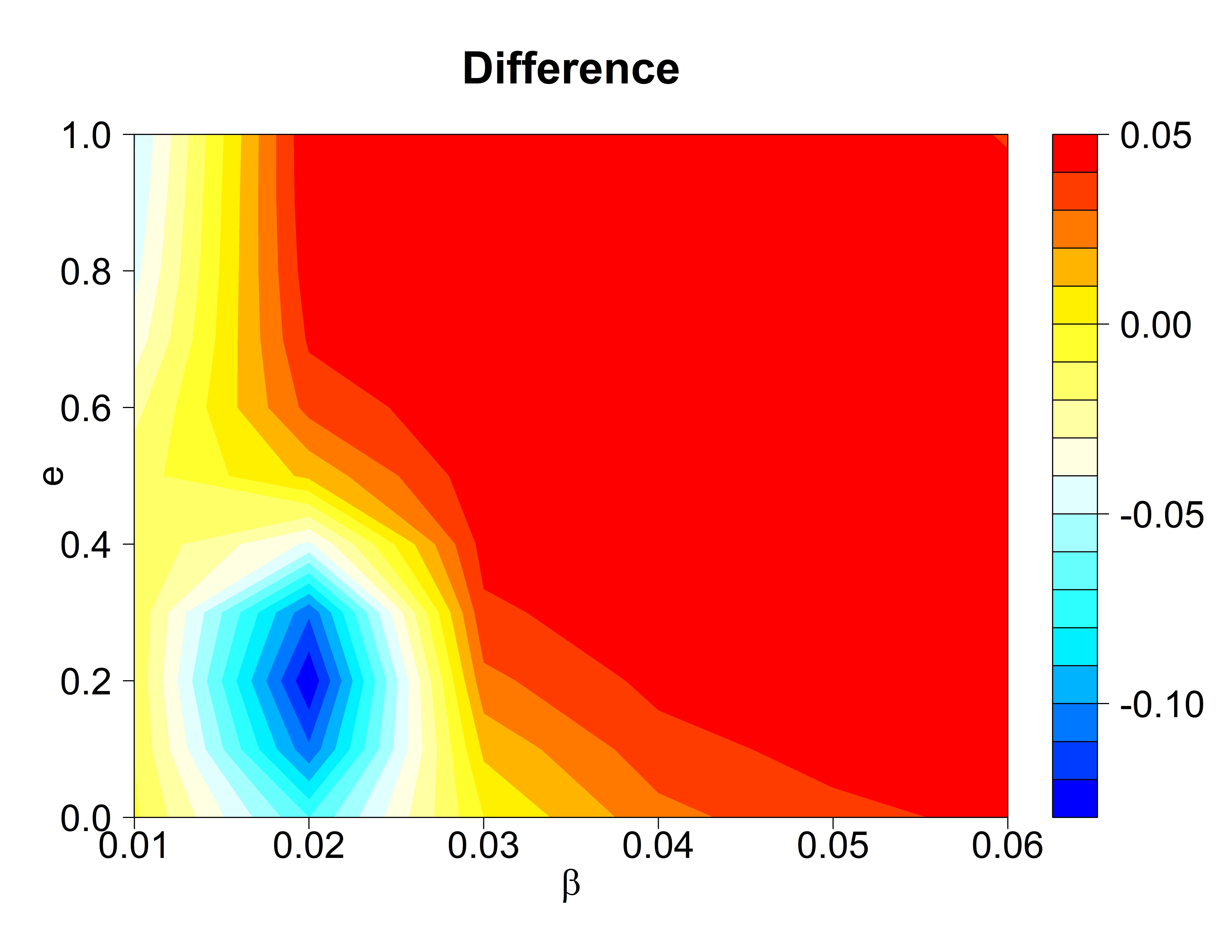}}\\ \vspace{-4.5mm}
	\subfloat[]{\includegraphics[width=0.30\textwidth]{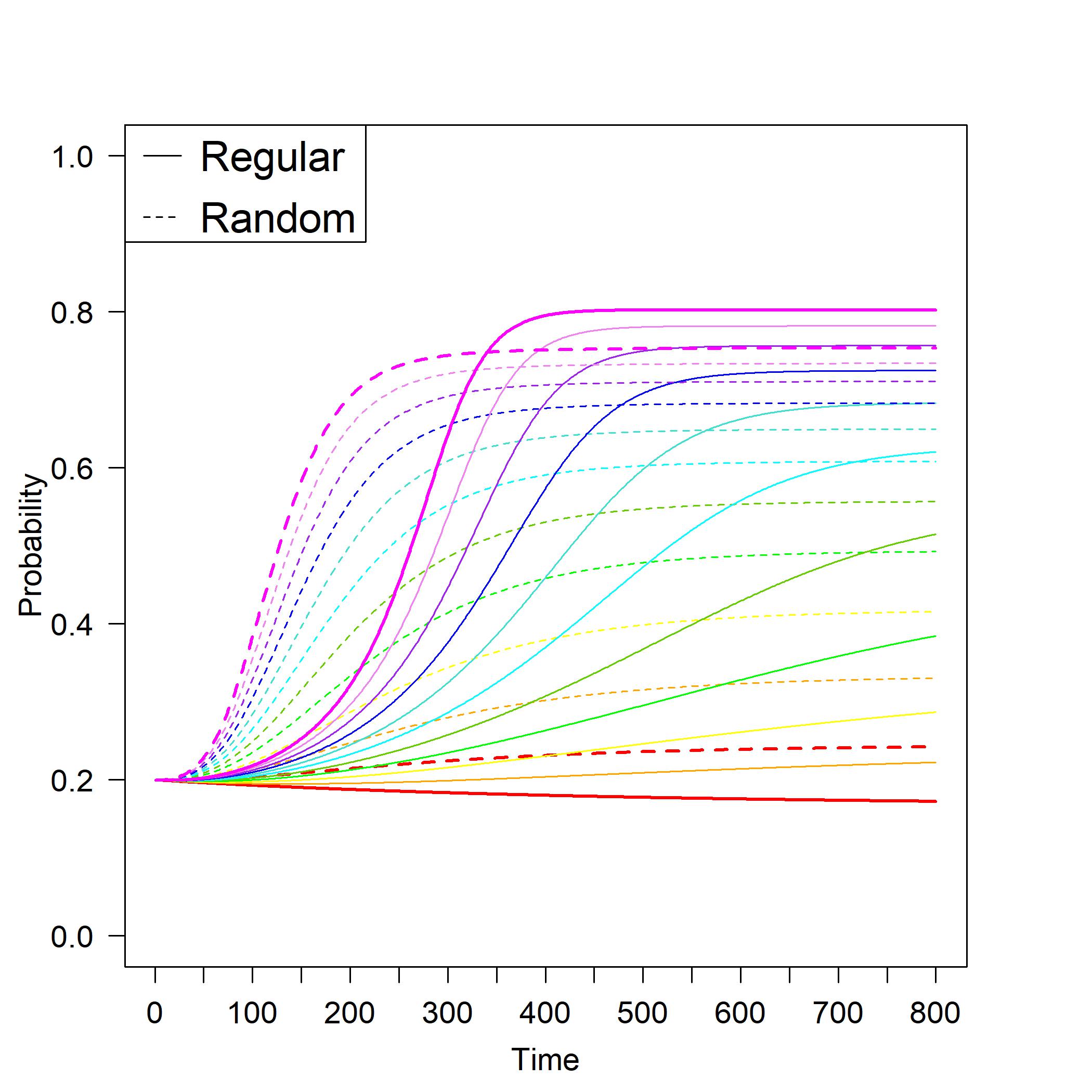}}
	\caption{Difference $X^{\star}_{\rm lattice}-X^{\star}_{\rm random}$ and time evolution of the prevalence for different values of $e$: panels (a) and (b), cycle with $n=20$, $\delta=0.1053$, $p=4/20$ and $\gamma=0.01$; panels (c) and (d), $3-$regular graph with $n=20$, $\delta=0.1579$, $p=4/20$ and $\gamma=0.01$.}
	\label{fig14} 
\end{figure}


\section{Conclusion}
We propose a new paradigm of interaction between a complex network and its line graph, which is used to implement a self-adaptive epidemic model based on the SIS model equations on networks. We discuss the existence and stability properties of the asymptotic solutions of the model for general network topologies. We also provide the solution in a closed form for some specific synthetic graphs. These asymptotic endemic values are then interpreted as a new centrality measure for both nodes and edges.

In its current form, the model allows for a reinforcement action, where the probability of an edge being a transmission channel increases as the infection probability of the nodes connected to it increases, and vice versa. We used this key factor to interpret the effects of reinforcement that typically operates in online social networks during processes of opinion or behavior adoption.
	
A slight variation of the model will allow the introduction of a penalty effect, where the greater the probability of a node being infected, the less weight is assigned to the edges connected to it. In this way, we are confident that we can extend the model's capabilities to different real-world scenarios. For instance, in the domains of viability and navigation, the weights of links directed to a node may decrease when its susceptibility to a particular form of disruption or shock is identified.

Finally, we point out that the proposed idea, i.e. the real-time interaction of a graph and its line graph, can be extended beyond the SIS model studied here and we argue that other dynamic processes can be effectively included in the proposed paradigm.
	
	\hfill
	
%

{\bf DATA AVAILABILITY}

Data generated and analyzed during the current study are available from the corresponding author on request.

{\bf AUTHOR DECLARATIONS}

\textbf{Conflict of Interest}

The authors have no conflicts to disclose.

%

\clearpage

\appendix
	
\section{Application to synthetic graphs}
\label{appendixA}
	In this Appendix, we provide analytical results for some specific classes of graphs. We report the proof of Theorem \ref{theorem1_cycle} on the existence and stability of the equilibrium point for cycle graphs. We then generalize the result to regular graphs, that include the complete graphs as special cases. Finally, we investigate the case of star graphs, which does not lead to a closed solution, but to a bound on its stationary states.
	
	\subsection{Cycle graphs $C_n$}
	\label{subsection_Cycle}
	We prove Theorem \ref{theorem1_cycle}.
	\begin{proof}
		The equilibrium points are solutions of the equation $x(2\beta x^{2}-2\beta x+\gamma)=0$, that is $x^{\star}=0$ and $x^{\star}=\frac{\beta \pm \sqrt{\beta (\beta -2\gamma)}}{2\beta}$. It is immediate to observe that, studying the sign of the derivative, the stable equilibrium points are only $x^{\star}=0$ for $\beta < 2 \gamma$, and $x^{\star}=\frac{\beta + \sqrt{\beta (\beta -2\gamma)}}{2\beta}$ for $\beta > 2 \gamma$. If the initial probability $p$ is below the unstable equilibrium point, that is if $p< \frac{\beta - \sqrt{\beta (\beta -2\gamma)}}{2\beta}$, the stable asymptotic solution is again $x^{\star}=0$, since $\dot{x}<0$. Therefore,
		\begin{equation}
			\resizebox{.88\hsize}{!}{$
			\label{equilibrium_cycle2}
			\left\{ 
			\begin{array}{lll}
				x^{\star}=0 & {\rm if} & \beta < 2 \gamma \ {\rm or}\ \beta > 2 \gamma \land  p< \frac{\beta - \sqrt{\beta (\beta -2\gamma)}}{2\beta} \\
				\hfill \\
				x^{\star}=\frac{1}{2}\left( 1+ \sqrt{1-\frac{2\gamma}{\beta}} \right) & {\rm if} & \beta > 2 \gamma \land  p> \frac{\beta - \sqrt{\beta (\beta -2\gamma)}}{2\beta}
			\end{array}
			\right.$}
		\end{equation}
		Recall that we set $q=1-p$. Since $p> \frac{\beta - \sqrt{\beta (\beta -2\gamma)}}{2\beta}$ implies $\beta > \frac{2\gamma}{1-(q-p)^2}=\frac{\gamma}{2pq}$, for $0<p<\frac{1}{2}$, and $ \beta > 2 \gamma$ for $\frac{1}{2}\leq p<1$, 
		we can identify $\tau_{c}(p)=\frac{1}{2pq}$ as the threshold of the epidemic dynamics on cycles for $0<p<\frac{1}{2}$. For $\frac{1}{2}\leq p <1$ the threshold becomes constant and equal to $2$.
	\end{proof}
	
	\subsection{Regular graphs $K_{n}^{d}$}
	\label{appendixregular}
	We now generalize the results obtained for cycle graphs to a regular graph with $n$ nodes, degree $d<n$,  $m=\frac{1}{2}nd$ edges, and adjacency matrix ${\bf B}_{P}$. The corresponding line graph is regular, has $m=\frac{1}{2}nd$ vertices, $\frac{1}{2}nd(d-1)$ edges and degree $2(d-1)$. Let ${\bf B}_{D}$ be its binary adjacency matrix. The  symmetry of matrices ${\bf B}_{P}$ and ${\bf B}_{D}$ ensures that $x_{i}(t)=x(t),\ \forall i=1,\dots, n$ and $y_{j}(t)=y(t),\ \forall j=1,\dots, m$, but, in general, $x(t)\neq y(t)$. Moreover:
	${\rm diag}\, {\bf x}(t)=x(t){\bf I}_{n}$,
	${\rm diag}\, {\bf y}(t)=y(t){\bf I}_{m}$,
	${\rm diag}({\bf E}{\bf u}_{m})=d{\bf I}_n$,
	${\rm diag}({\bf E}^{T}{\bf u}_{n})=2{\bf I}_{m}$,
	${\bf E}\, {\rm diag}\, {\bf u}_{m} {\bf E}^{T}-{\rm diag}({\bf E}{\bf u}_{m})={\bf B}_{P}\in {\mathbb R}^{n\times n}$, and
	${\bf E}^{T}\, {\rm diag}\, {\bf u}_{n} {\bf E}-{\rm diag}({\bf E}^{T}{\bf u}_{n})={\bf B}_{D}\in {\mathbb R}^{m\times m}$,
	so that
	\begin{equation}
		\left\{ 
		\begin{array}{l}
			{\bf A}_{P}(t)=y(t){\bf B}_{P}\\
			\hfill \\
			{\bf A}_{D}(t)=x(t){\bf B}_{D}\\
		\end{array}
		\right.
		.
	\end{equation}
	For infectivity and recovery rates equal for the networks $G_P$ and $G_D$, Eq. (\ref{continuos_eqs}) becomes
	\begin{equation}
			\resizebox{.88\hsize}{!}{$
		\label{solution_regular}
		\left\{ 
		\begin{array}{l}
			\dot{x}_{i}(t)=\beta\left[1-x_{i}(t) \right]y(t)\sum_{h=1}^{n} ({\bf A}_{P})_{ih}\, x_{h}(t)-\gamma x_{i}(t)\qquad i=1,\dots, n \\
			\hfill \\
			\dot{y}_{j}(t)=\beta\left[1-y_{j}(t) \right]x(t)\sum_{h=1}^{m} ({\bf A}_{D})_{jh}\, y_{h}(t)-\gamma y_{j}(t)\quad j=1,\dots, m\\
		\end{array}
		\right.$}
	\end{equation}
	Let us handle the equation in $x_{i}(t)=x(t)$:
	\begin{equation}
		\begin{split}
			\dot{x}(t)=&\beta\left[1-x(t) \right]y(t)\sum_{h=1}^{n} ({\bf A}_{P})_{ih}\, x_{h}(t)-\gamma x(t)\\
			=&\beta\left[1-x(t) \right]y(t)x(t)\sum_{h=1}^{n} ({\bf A}_{P})_{ih}-\gamma x(t)\\
			=&\beta d \left[1-x(t) \right]y(t)x(t)-\gamma x(t).
		\end{split}
	\end{equation}
	Similarly for $y(t)$, so that we get the system
	\begin{equation}
		\left\{ 
		\begin{array}{l}
			\dot{x}(t)=\beta d \left[1-x(t) \right]y(t)x(t)-\gamma x(t)\\
			\hfill \\
			\dot{y}(t)=2\beta (d-1) \left[1-y(t) \right]x(t)y(t)-\gamma y(t)\\
		\end{array}
		\right.
		.
		\label{regular_equations}
	\end{equation}
	The nature of the steady state equilibrium points of the problem (\ref{regular_equations}) is characterized by the following:
	
	\begin{theorem}
		\label{theorem2}
		The stable equilibrium points of the ASIS model on the d-regular graph $K_{n}^{d}$ with infectivity rate $\beta$ and recovery rate $\gamma$ on the network $G_P$ are given by
		\begin{equation}
			\label{solutions_regular_final}
			\left\{ 
			\begin{array}{lll}
				x^{\star}=0 & {\rm if} & \mathcal{R}<\tau_{r}\\
				\hfill \\
				x^{\star}=\frac{1}{2}\left(1-\frac{d-2}{2d(d-1)\mathcal{R}}+ \frac{\sqrt{\xi}}{2d(d-1)\mathcal{R}}\right) & {\rm if} & \mathcal{R}>\tau_{r}
			\end{array}
			\right.
		\end{equation}
		where $\mathcal{R}= \frac{\beta}{\gamma}$, $\xi=\left[ (d-2)-2 d(d-1)\mathcal{R} \right]^{2}-8d^2(d-1)\mathcal{R}$ and
		\begin{equation}
			\label{threshold_regular}
			\tau_{r}=
			\left\{ 
			\begin{array}{lll}
				\frac{d+(d-2)p}{2d(d-1)}\cdot \frac{1}{pq} & {\rm if} & 0<p<\frac{1}{1+\sqrt{\frac{2(d-1)}{d}}}\\
				\left[ \frac{1}{\sqrt{d}}+\frac{1}{\sqrt{2(d-1)}}\right]^{2} & {\rm if} & \frac{1}{1+\sqrt{\frac{2(d-1)}{d}}}\leq p<1
			\end{array}
			\right.
		\end{equation}
		is the threshold of the epidemic dynamics on regular graphs.
	\end{theorem}
	
	\begin{proof}
		
		The equilibrium points of the problem (\ref{regular_equations}) are given by the null solutions $x^{\star}=0$ and $y^{\star}=0$ and by the solutions of the nonlinear system
		\begin{equation}
			\label{nonlinearesystem}
			\left\{ 
			\begin{array}{l}
				d\mathcal{R} (1-x)y-1=0\\
				\hfill \\
				2(d-1)\mathcal{R} (1-y)x-1=0
			\end{array}
			\right.
			.
		\end{equation}
		The solving equation in $x$ is $2d(d-1)\mathcal{R} x^2+\left[ (d-2)-2d(d-1)\mathcal{R}\right]x+d=0$. Therefore, we have
		\begin{equation}
			\label{solutions}
			\left\{ 
			\begin{array}{l}
				x^{\star}=\frac{1}{2}\left(1-\frac{d-2}{2d(d-1)\mathcal{R}}\pm \frac{\sqrt{\xi}}{2d(d-1)\mathcal{R}}\right)\\
				\hfill \\
				y^{\star}=\frac{1}{2}\left(1+\frac{d-2}{2d(d-1)\mathcal{R}}\pm \frac{\sqrt{\xi}}{2d(d-1)\mathcal{R}}\right)\\
			\end{array}
			\right.
		\end{equation}
		with $\xi=\left[ (d-2)-2 d(d-1)\mathcal{R} \right]^{2}-8d^2(d-1)\mathcal{R}$. Note that $\xi \geq 0$ for $0< \mathcal{R} \leq {\tau}_{1} \cup \mathcal{R} \geq {\tau}_{2}$ where
		\begin{equation}
				\resizebox{.90\hsize}{!}{$
			\begin{split}
			\footnotesize
			\label{conditions}
			{\tau}_{1,2}&\coloneqq \frac{3d-2}{2d(d-1)}\pm \sqrt{\frac{2}{d(d-1)}}
			=\frac{1}{d}+\frac{1}{2(d-1)}\pm \frac{2}{\sqrt{2d(d-1)}}\\
			&=\left( \frac{1}{\sqrt{d}}\pm\frac{1}{\sqrt{2(d-1)}}\right)^{2}.
			\end{split}$}
		\end{equation}
		Let us focus on the steady states for the primary process. Let us distinguish the following cases:
		\begin{itemize}
			\item ${\tau}_{1}< \mathcal{R} < {\tau}_{2}$: there is a unique equilibrium point, a unique steady state solution and it is $x^{\star}=0$. 
			
			\item $0< \mathcal{R} < {\tau}_{1}$: the two non-trivial solutions $x^{\star}$ in Eq. (\ref{solutions}) exist but they both are negative. Therefore the sign of the right-hand side in Eq. (\ref{regular_equations}), that is the sign of $\dot{x}$, is positive below $x^{\star}=0$ and negative above $x^{\star}=0$. Therefore, the null solution is again the only meaningful stable solution. 
			
			\item $\mathcal{R} \geq {\tau}_{2}$: in addition to the null solution, both the non-trivial solutions $x^{\star}$ in Eq. (\ref{solutions}) exist and they are positive. We represent in Fig. \ref{derivative} the signs of the first derivative $\dot{x}$, where $x^{\star}_{1}$ and $x^{\star}_{2}$ refer to the solutions in Eq. (\ref{solutions}).
		
			To conclude the discussion about stability, let us observe that, if the initial probability $p$ at time $t=0$ lies below the value of $x^{\star}_{1}$ then again the only stable steady state remains $x^{\star}=0$. If, instead, $p>x^{\star}_{1}$, that is
			\begin{equation}
				p>\frac{1}{2}\left(1-\frac{d-2}{2d(d-1)\mathcal{R}}- \frac{\sqrt{\xi}}{2d(d-1)\mathcal{R}}\right)
				\label{disequazione}
			\end{equation}
			the stable steady state becomes $x^{\star}_{2}$. Inequality (\ref{disequazione}) solved for $\mathcal{R}$ gives
			
			\begin{equation*}
				\label{threshold1}
				\mathcal{R} > {\tau}_{r}=
				\left\{ 
				\begin{array}{lll}
					\frac{d+(d-2)p}{2d(d-1)}\cdot \frac{1}{pq} & {\rm if} & 0<p<\frac{1}{1+\sqrt{\frac{2(d-1)}{d}}}\\
					\left[ \frac{1}{\sqrt{d}}+\frac{1}{\sqrt{2(d-1)}}\right]^{2} & {\rm if} & \frac{1}{1+\sqrt{\frac{2(d-1)}{d}}}\leq p<1
				\end{array}
				\right.
				.
			\end{equation*}
		\end{itemize}
\begin{figure}[H]
	\centering
	{\includegraphics[width=0.40\textwidth]{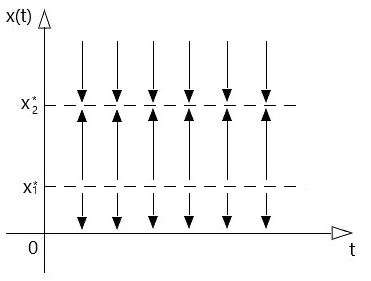}}
	\caption{Sign of the derivative around the equilibrium points $x^{\star}_{1}$ and $x^{\star}_{2}$.}
	\label{derivative} 
\end{figure}
\end{proof}
	
	\begin{remark}
		Theorem \ref{theorem2} extends Theorem \ref{theorem1_cycle} proved for the case of the cycle $C_n$. Indeed, when $d=2$, we have $\xi=16 \mathcal{R}(\mathcal{R}-2)$, $\tau_{1}=0$ and $\tau_{2}=2$. Moreover $\frac{1}{1+\sqrt{\frac{2(d-1)}{d}}}=\frac{1}{2}$.
	\end{remark}
	
	\begin{remark}
		Threshold $\tau_{r}$ in Eq. (\ref{threshold_regular}), in general, depends on both $p$ and $d$. There is a critical value, that is 
		$\frac{1}{1+\sqrt{\frac{2(d-1)}{d}}}$, which discriminates the two values of $\tau_{r}$.
		In both cases $\tau_{r}$ is a decreasing function of $d$, as expected. When $p$ is below the critical value, $\tau_{r}$ depends on $p$ and it increases when $p$ decreases. Above the critical value, $\tau_{r}$ is independent of $p$. Note also that it is equal to $\frac{1}{2}$ for $d=2$, and tends to $\sqrt{2}-1$ when $d$ approaches $+\infty$.
		Interestingly, the threshold of the standard SIS model on a d-regular graph is $\frac{1}{p\lambda_{1}}=\frac{1}{p2(d-1)}$ and it is always lower than $\tau_{r}$ for any $0<p<1$.
	\end{remark}
	
	Consider as an example a regular graph with $n=6$ nodes and $d=3$, so $m=9$ edges.
	Under these conditions, the threshold is
	\begin{equation*}
		\label{threshold_example}
		\tau_{r}=
		\left\{ 
		\begin{array}{lll}
			\frac{3+p}{12pq} & {\rm if} & 0<p<0.464\\
			1.161 & {\rm if} & 0.464 \leq p<1
		\end{array}
		\right.
		.
	\end{equation*}
	\vspace{-6mm}
	\begin{figure}[H]
		\centering
		\subfloat[]{\includegraphics[width=0.34\textwidth]{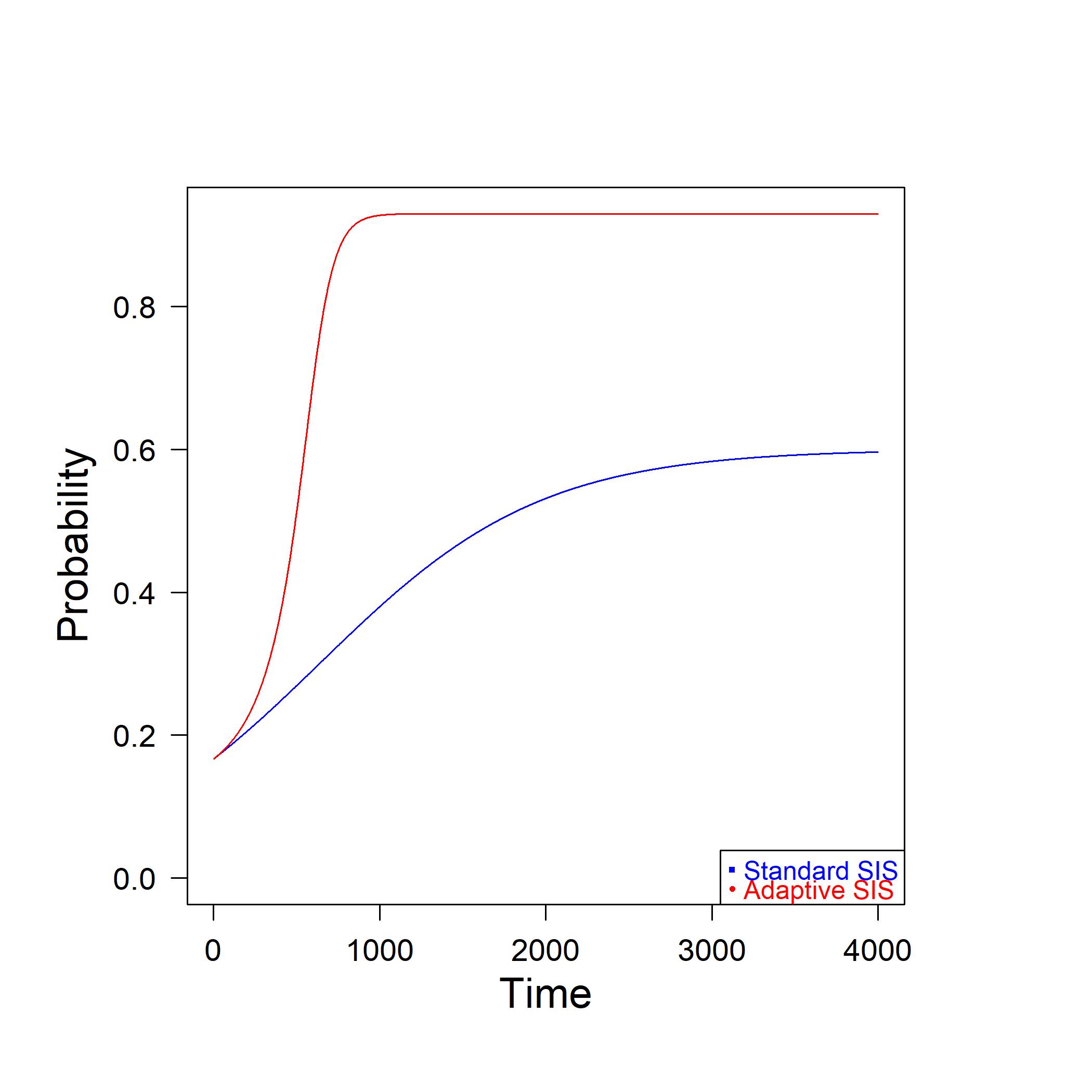}}\\  \vspace{-4mm}
		\subfloat[]{\includegraphics[width=0.34\textwidth]{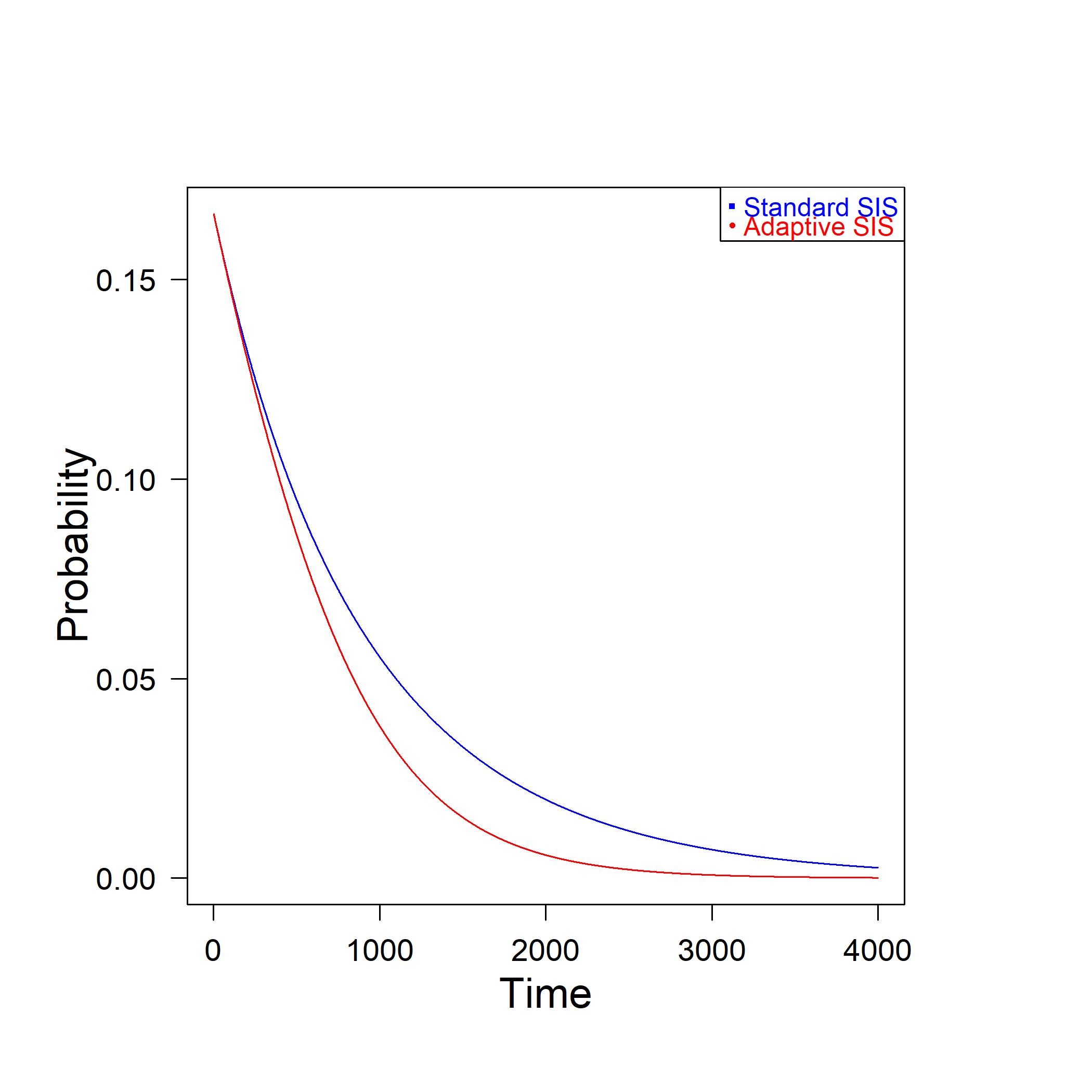}}
		\caption{Prevalence $x(t)$ for the ASIS model (in red) and for the standard SIS model (in blue) for a graph $K_{6}^{3}$ with (a) $\beta=0.005$ and $\gamma=0.001$; (b) $\beta=0.002$ and $\gamma=0.002$.}
		\label{fig16} 
	\end{figure}

	In Fig. \ref{fig16}, panel (a), we plot the evolution of the model for $\beta=0.005$, $\gamma=0.001$, then $\mathcal{R}=5$. We choose $p=1/6=0.167<0.464$, so that $\tau_{r}=1.9$. We have $\mathcal{R}>\tau_{r}$ and the stable steady state is $x^{\star}=0.9295435$. In Fig. \ref{fig16}, panel (b), we plot the evolution of the model for $\beta=0.002$, $\gamma=0.002$, $\mathcal{R}=1$ and $p=1/6$. Now $\mathcal{R}<\tau_{r}$ and the asymptotic steady state is $x^{\star}=0$.
	
\begin{remark}
		We provide now a graphical interpretation of the previous results in the $x-y$ plane. The derivatives in Eq. (\ref{nonlinearesystem}) are both positive in a finite region identified by \vspace{-2mm}
		\begin{equation}
			\label{interpretation}
			\left\{ 
			\begin{array}{l}
				y\geq \frac{1}{d\mathcal{R} (1-x)}\\
				\hfill \\
				y\leq 1-\frac{1}{2(d-1)\mathcal{R}x}\\
			\end{array}
			\right.
		\end{equation}
    	whose boundary curves intersect at points
		\begin{equation}
			\begin{split}
			\label{solutions1}
			&\left\{ 
			\begin{array}{l}
				x^{\star}_{1}=\frac{1}{2}\left(1-\frac{d-2}{2d(d-1)\mathcal{R}}- \frac{\sqrt{\xi}}{2d(d-1)\mathcal{R}}\right)\\
				\hfill \\
				y^{\star}_{1}=\frac{1}{2}\left(1+\frac{d-2}{2d(d-1)\mathcal{R}}- \frac{\sqrt{\xi}}{2d(d-1)\mathcal{R}}\right)\\
			\end{array}
			\right.
			\\
			& \hspace{-10mm}{and}\\
			&\left\{ 
			\begin{array}{l}
				x^{\star}_{2}=\frac{1}{2}\left(1-\frac{d-2}{2d(d-1)\mathcal{R}}+ \frac{\sqrt{\xi}}{2d(d-1)\mathcal{R}}\right)\\
				\hfill \\
				y^{\star}_{2}=\frac{1}{2}\left(1+\frac{d-2}{2d(d-1)\mathcal{R}}+ \frac{\sqrt{\xi}}{2d(d-1)\mathcal{R}}\right)\\
			\end{array}
			\right.
			\end{split}
		\end{equation}
		Fig. \ref{fig17}, panel (a), illustrates the region in Eq. (\ref{interpretation}) and the intersection points in Eq. (\ref{solutions1}) for $\beta=0.005$, $\gamma=0.001$, $n=6$, and $d=4$. Fig. \ref{fig17}, panel (b), illustrates the trajectory (green line) of the time evolution of the epidemic in the $x-y$ plane under the same conditions and $p=1/6$. The plus sign ($+$) indicates the starting point of the phase diagram and the empty circle ($\circ$) the ending (asymptotic) point toward the attractive stable solution.
		
		Fig. \ref{fig18}, panels (a-d), illustrates the trajectories (green line) of the evolution of the epidemic in the $x-y$ plane when the two nontrivial solutions in Eq. (\ref{solutions1}) exist. Fig. \ref{fig18}, panels (e-h), illustrates the analog trajectories (green line) when the only equilibrium point is the null solution. To better illustrate the behavior under different conditions we have relaxed the assumption that the initial probability is identical for nodes in network $G_P$ and nodes in network $G_D$ and we used different values for the initial probabilities $p_x$ for the variable $x$ and $p_y$ for the variable $y$. In the different panels, we used the following parameters: (a) $\beta=0.005$, $\gamma=0.001$, $p_{x}=0.1$, $p_{y}=0.9$; (b) $\beta=0.005$, $\gamma=0.001$, $p_{x}=0.9$, $p_{y}=0.1$; (c) $\beta=0.002$, $\gamma=0.001$, $p_{x}=0.10$, $p_{y}=0.15$; (d) $\beta=0.002$, $\gamma=0.001$, $p_{x}=0.9$, $p_{y}=0.9$; (e) $\beta=0.002$, $\gamma=0.002$, $p_{x}=1/6$, $p_{y}=1/6$; (f) $\beta=0.002$, $\gamma=0.002$, $p_{x}=0.8$, $p_{y}=0.8$; (g) $\beta=0.002$, $\gamma=0.002$, $p_{x}=0.2$, $p_{y}=0.8$; (h) $\beta=0.002$, $\gamma=0.002$, $p_{x}=0.8$, $p_{y}=0.2$.
		
		\begin{figure}[H]
			\centering
			\subfloat[]{\includegraphics[width=0.45\textwidth]{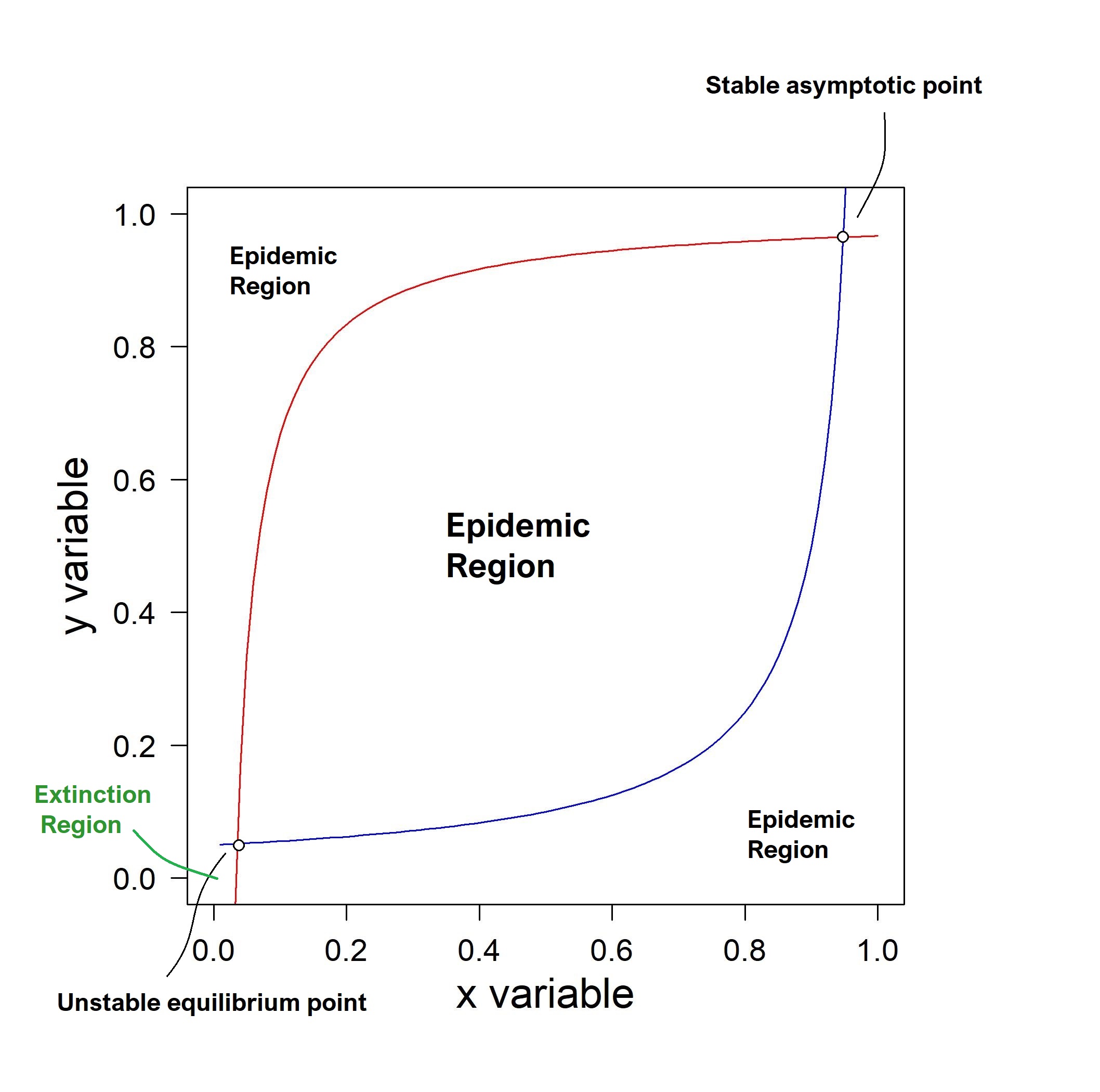}}\\
			\subfloat[]{\includegraphics[width=0.45\textwidth]{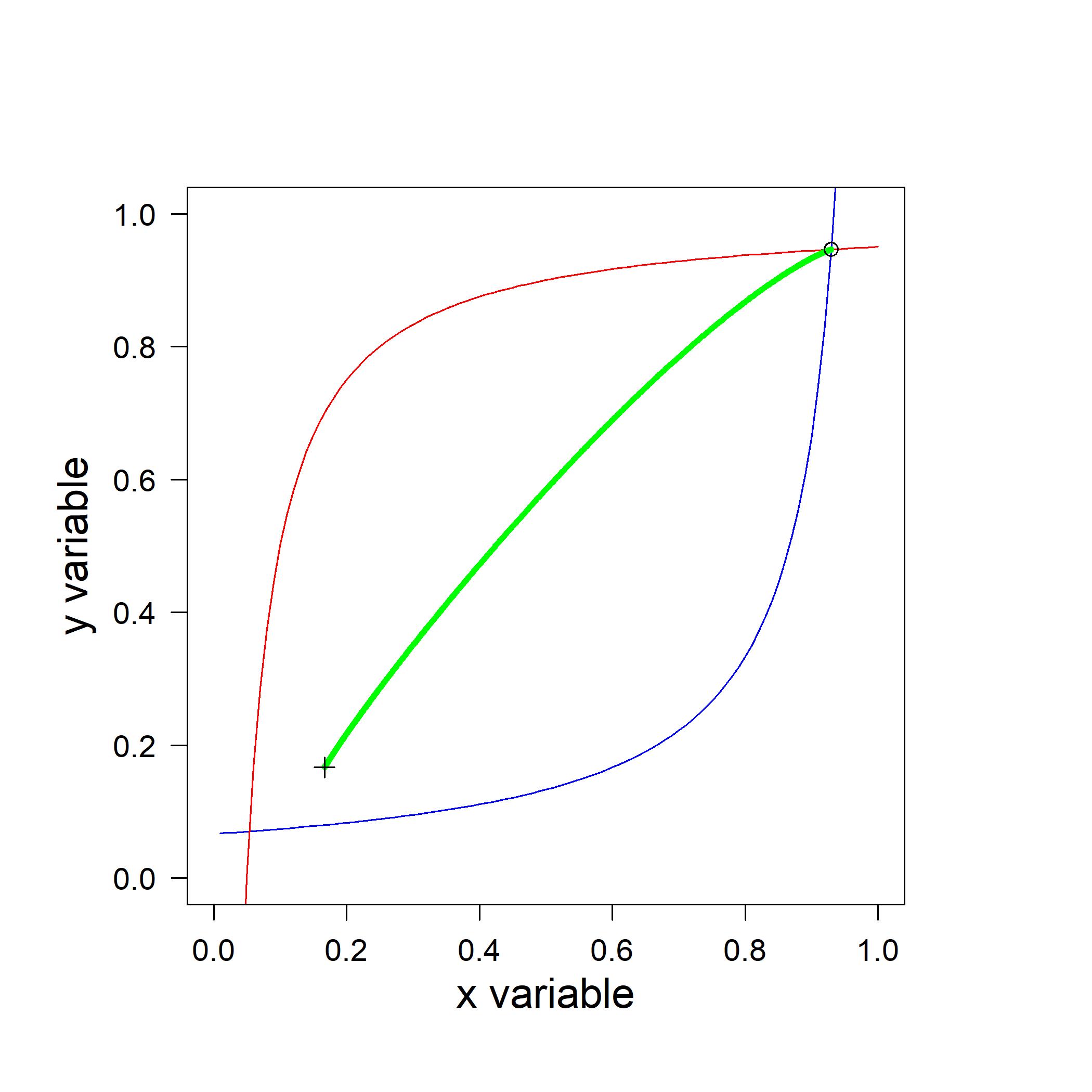}}
			\caption{(a) Different regions of the plane $x-y$ according the the asymptotic behavior of the model; (b) phase diagram (green line) of the evolution of the probabilities $x$ and $y$ for $\beta=0.005$, $\gamma=0.001$ and $p=1/6$ in the regular graph with $n=6$ and $d=4$. The plus sign ($+$) is the starting point, the empty circle ($\circ$) is the ending point.}
			\label{fig17} 
		\end{figure}

			\begin{figure}[H]
			\centering
			\subfloat[]{\includegraphics[width=0.25\textwidth]{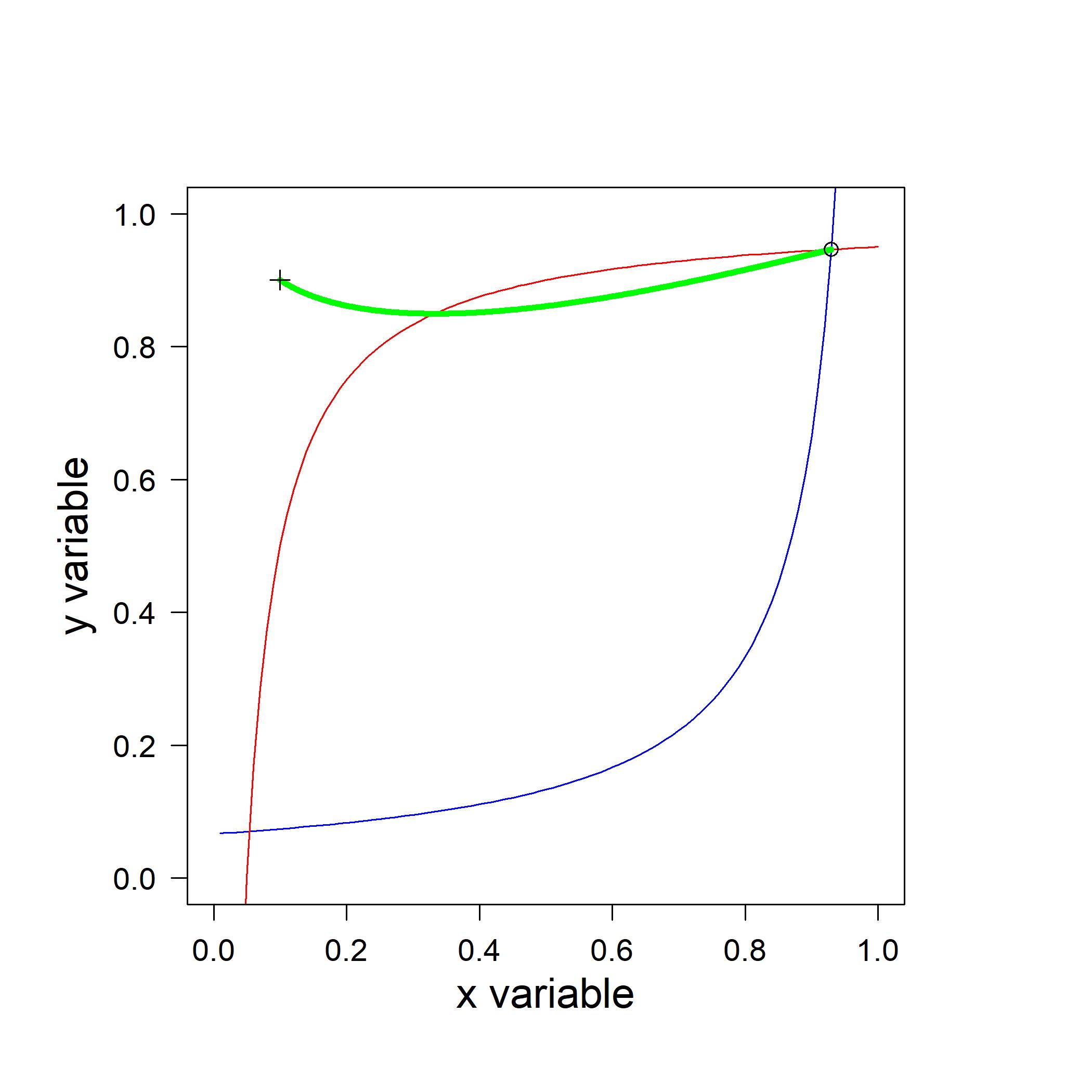}}
			\subfloat[]{\includegraphics[width=0.25\textwidth]{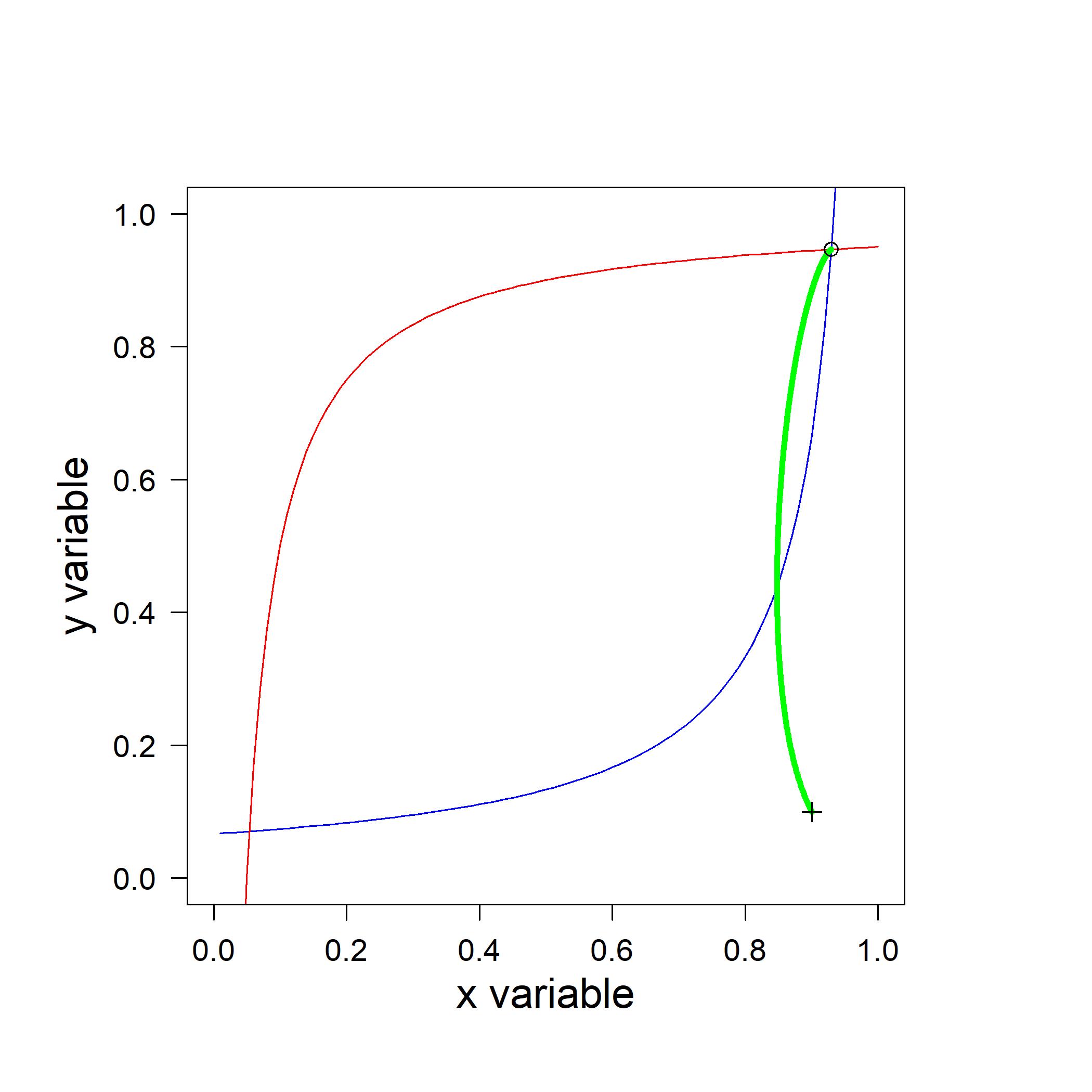}}\\  \vspace{-4mm}
			\subfloat[]{\includegraphics[width=0.25\textwidth]{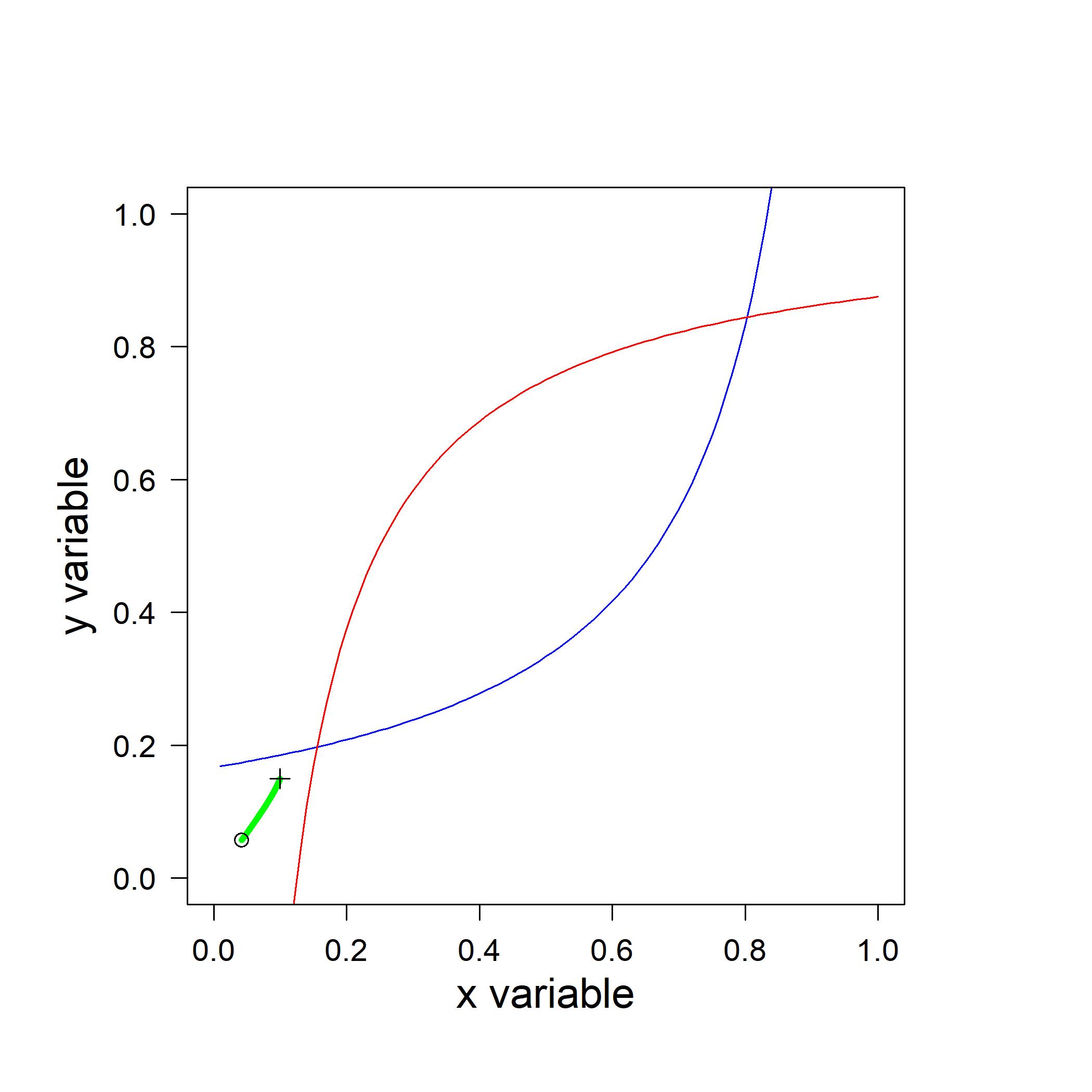}}
			\subfloat[]{\includegraphics[width=0.25\textwidth]{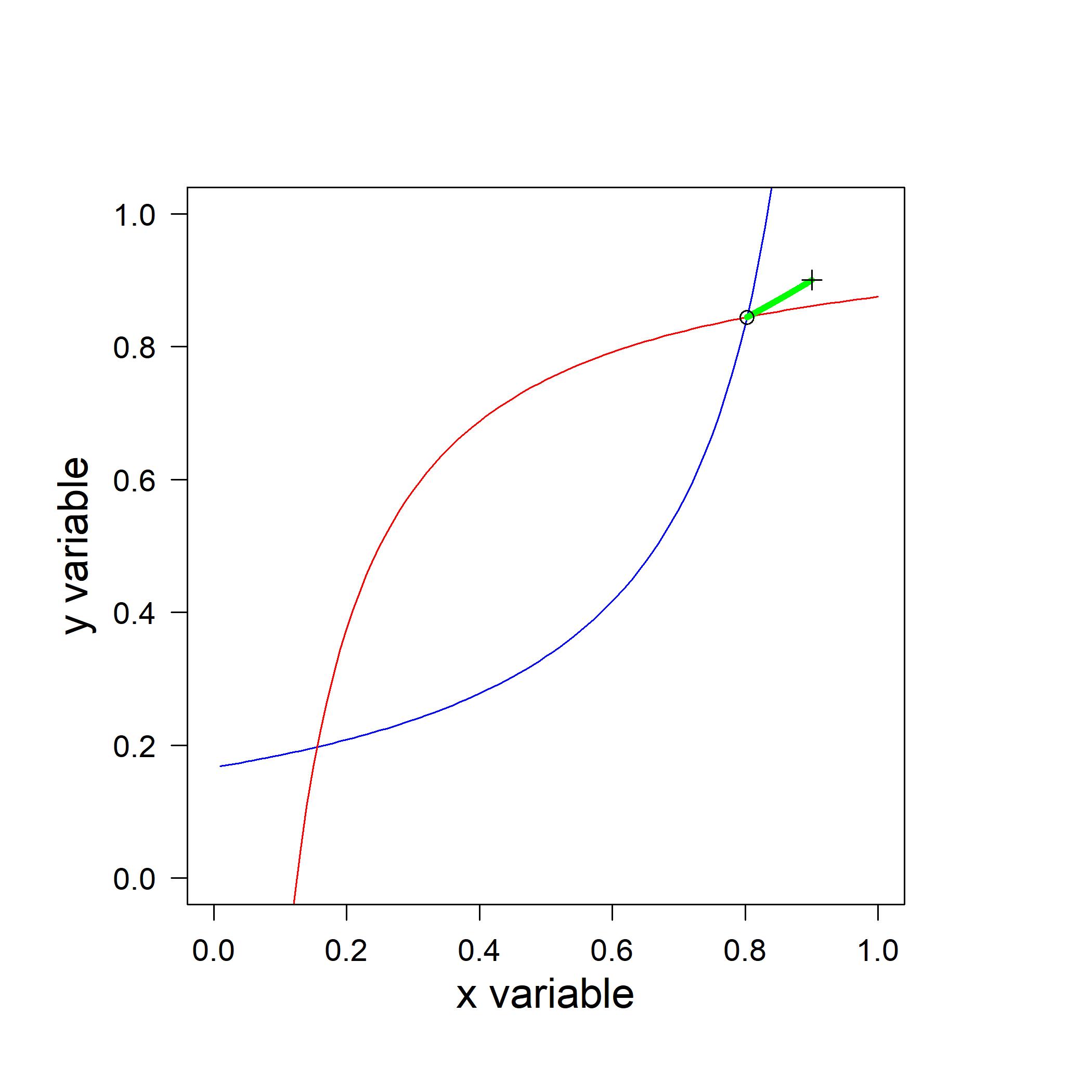}}\\  \vspace{-4mm}
			\subfloat[]{\includegraphics[width=0.25\textwidth]{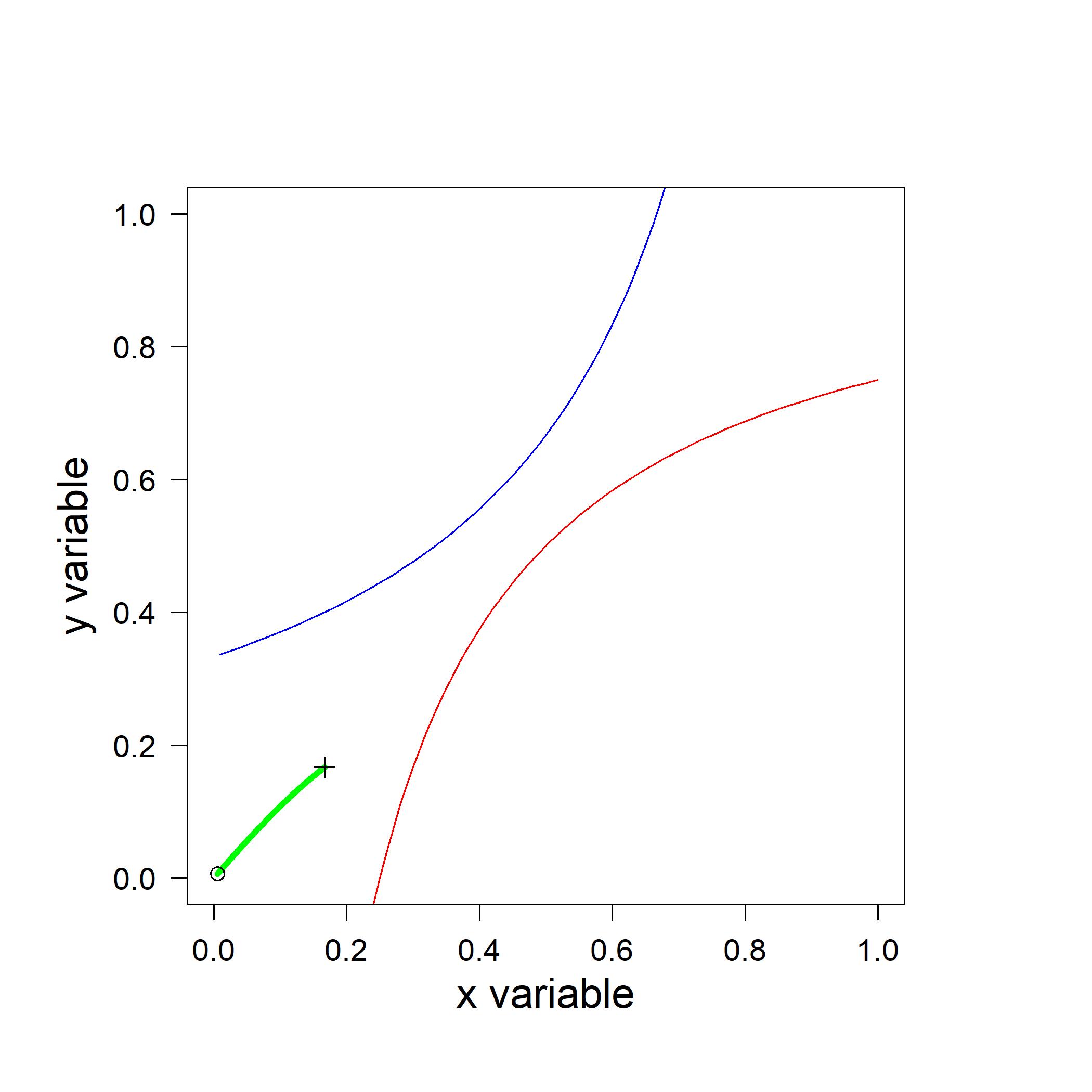}}
			\subfloat[]{\includegraphics[width=0.25\textwidth]{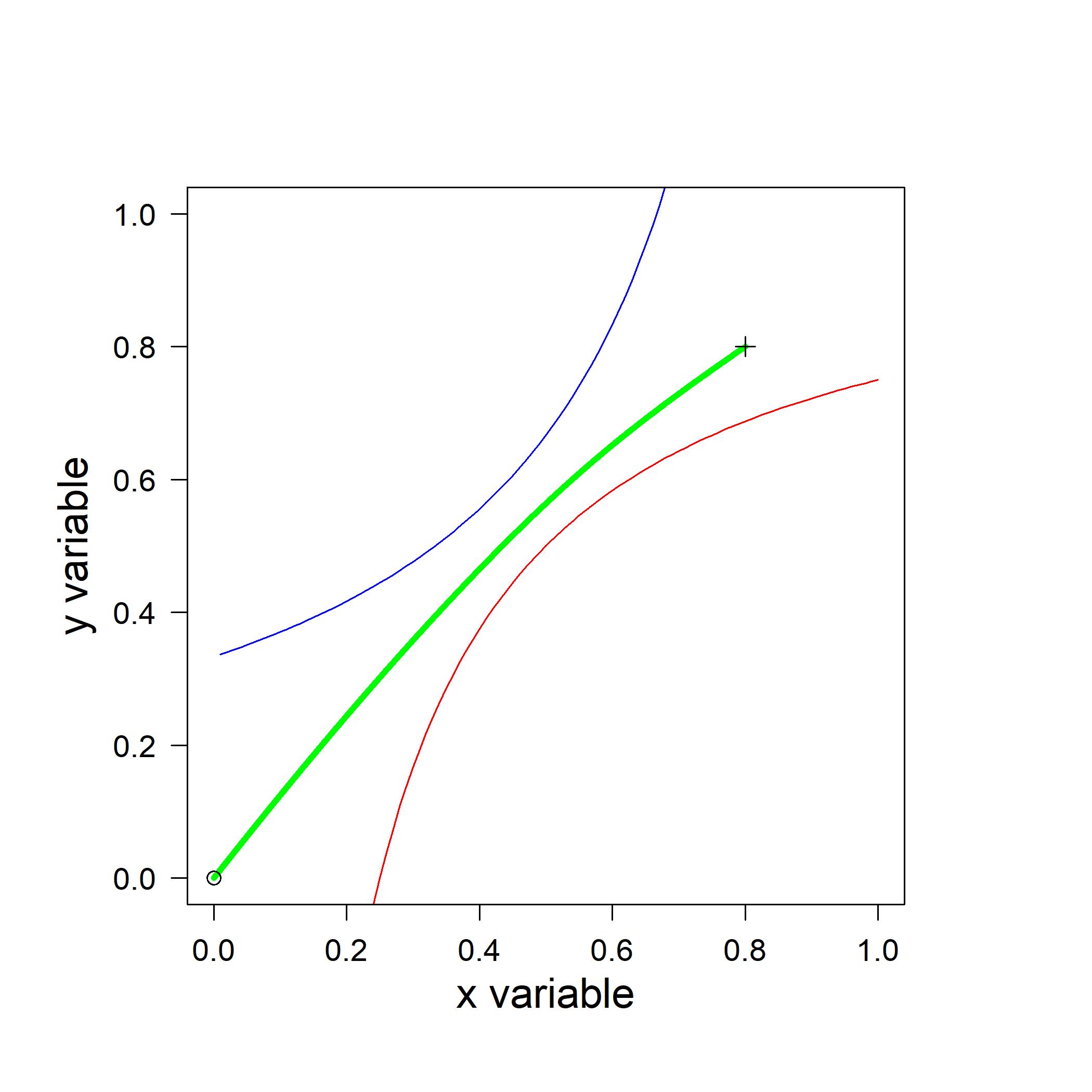}}\\  \vspace{-4mm}
			\subfloat[]{\includegraphics[width=0.25\textwidth]{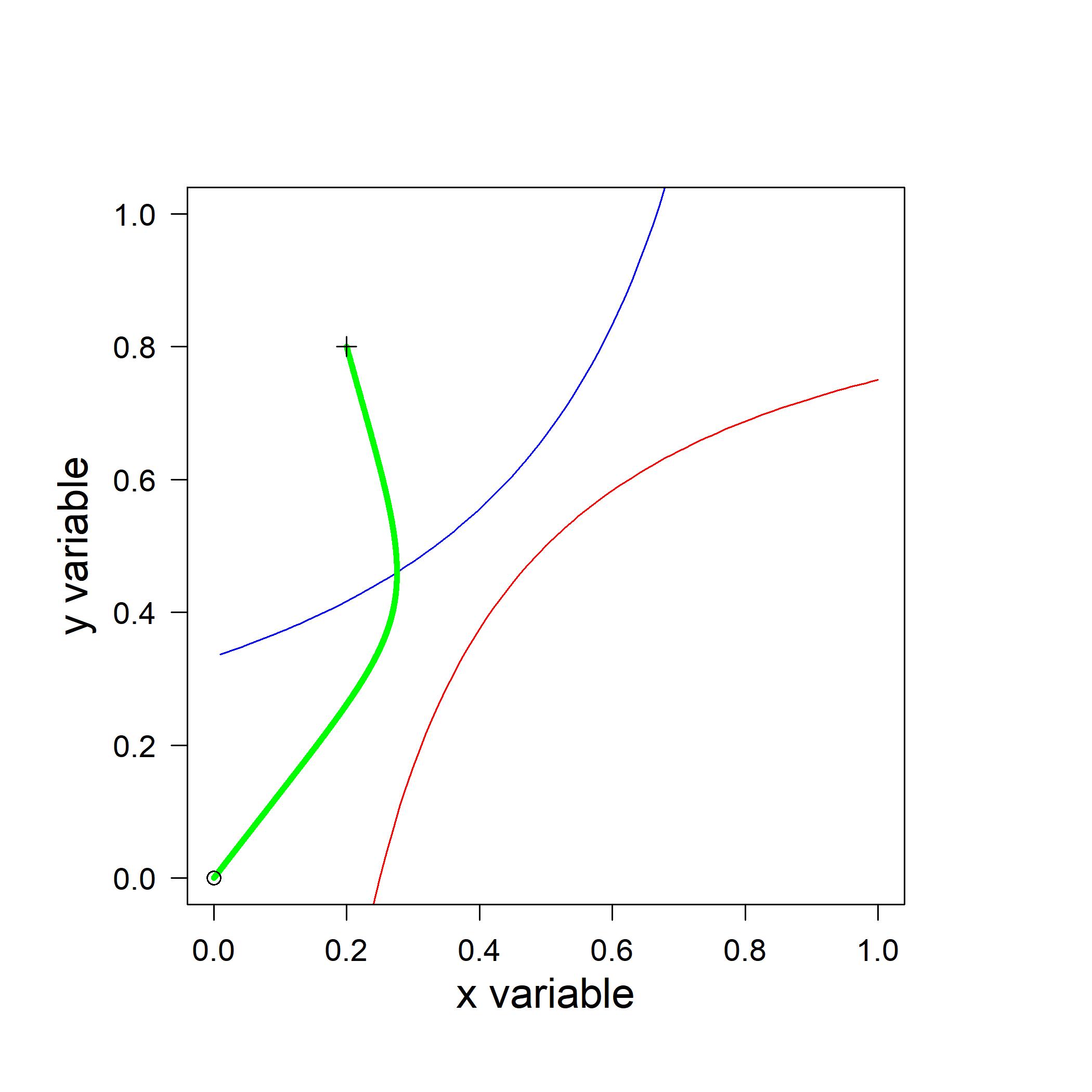}}
			\subfloat[]{\includegraphics[width=0.25\textwidth]{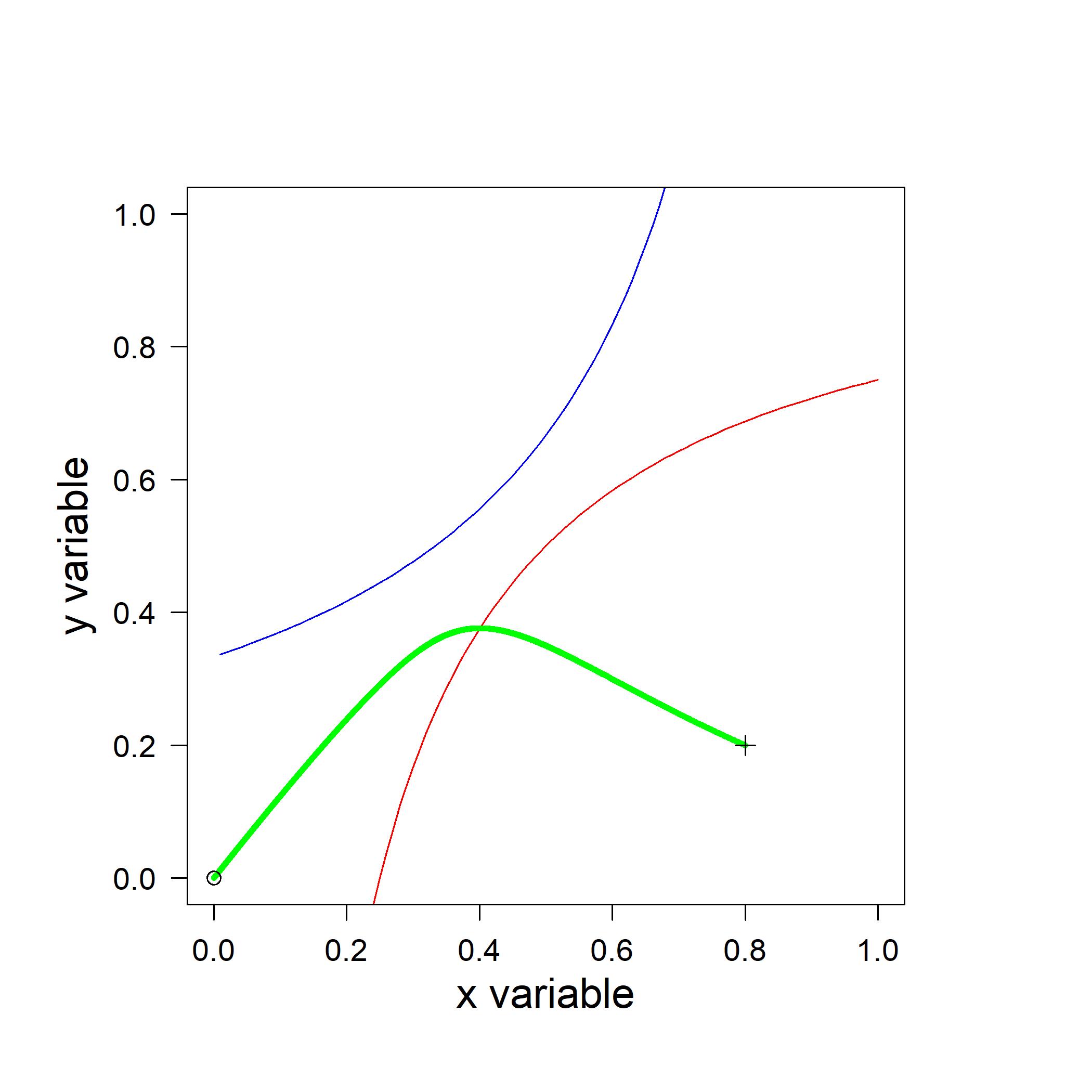}}			
			\caption{Panels (a)-(h): trajectories in the plane $x-y$ under the different conditions specified in the text.
			 }
			\label{fig18} 
		\end{figure}
		
\end{remark}

%

	\vspace{-11mm}

	\subsection{Star graphs $S_{n}$}
	Let us consider a star graph $S_{n}$ with $n$ nodes, $m=n-1$ edges, center in the node number $1$ with degree $n-1$, and adjacency matrix ${\bf B}$. The corresponding line graph is complete, has $m=n-1$ vertices, $\frac{1}{2}(n-1)(n-2)$ edges and constant degree $n-2$.
	The symmetries in the primary and dual processes ensure that the following apply: $x_{i}(t)=x(t),\ \forall i=2,\dots, n$ and $y_{j}(t)=y(t),\ \forall j=1,\dots, m$. Moreover: 
	${\rm diag}({\bf x}(t))={\rm diag}(x_{1}(t),x(t){\bf u}_{n-1})$,
	${\rm diag}({\bf y}(t))=y(t)\, {\rm diag}\,{\bf u}_{m}=y(t){\bf I}_{m}$,
	${\bf E}\, {\rm diag}\, {\bf u}_{m} {\bf E}^{T}-{\rm diag}({\bf E}{\bf u}_{m})={\bf B}_{P}\in {\mathbb R}^{n\times n}$,
	${\bf E}^{T}\, {\rm diag}\, {\bf x}\, {\bf E}-{\rm diag}({\bf E}^{T}{\bf x})=x_{1}{\bf B}_{D}\in {\mathbb R}^{m\times m}$,
	where ${\bf B}_{D}$ is the adjacency matrix of the complete graph with $m$ nodes. Therefore, Eq. (\ref{ApAd_continuos}) becomes
	\begin{equation}
		\resizebox{.30\hsize}{!}{$
		\left\{ 
		\begin{array}{l}
			{\bf A}_{P}(t)=y(t){\bf B}_{P}\\
			\hfill \\
			{\bf A}_{D}(t)=x_{1}(t){\bf B}_{D}\\
		\end{array}
		\right.$}
	\end{equation}
	and, by components, Eq. (\ref{continuos_eqs}) becomes
	\begin{equation}
		\label{discrete_solution_star}
		\resizebox{.85\hsize}{!}{$
		\left\{ 
		\begin{array}{l}
			\dot{x}_{1}(t)=\beta\left[1-x_{1}(t) \right]y(t)\sum_{h=1}^{n} ({\bf A}_{P})_{1h}\, x_{h}(t)-\gamma x_{1}(t) \\
			\hfill \\
			\dot{x}_{i}(t)=\beta\left[1-x_{i}(t) \right]y(t)\sum_{h=1}^{n} ({\bf A}_{P})_{ih}\, x_{h}(t)-\gamma x_{i}(t)\qquad \ i=2,\dots, n \\
			\hfill \\
			\dot{y}_{j}(t)=\beta\left[1-y_{j}(t) \right]x_{1}(t)\sum_{h=1}^{m} ({\bf A}_{D})_{jh}\, y_{h}(t)-\gamma y_{j}(t)\quad j=1,\dots, m\\
		\end{array}
		\right.$}
	\end{equation}
	The problem (\ref{discrete_solution_star}) is equivalent to
	\begin{equation}
		\label{solution_star}
		\resizebox{.65\hsize}{!}{$
		\left\{ 
		\begin{array}{l}
			\dot{x}_{1}(t)=\beta (n-1)\left[1-x_{1}(t) \right]y(t) x(t)-\gamma x_{1}(t) \\
			\hfill \\
			\dot{x}(t)=\beta\left[1-x(t) \right]y(t) x_{1}(t)-\gamma x(t)\\
			\hfill \\
			\dot{y}(t)=\beta (n-2) \left[1-y(t) \right]x_{1}(t)y(t)-\gamma y(t)
		\end{array}
		\right.$}
	\end{equation}
	The equilibrium points are, thus, given by
	\begin{equation}
		\label{solution_star_equilibrium}
		\resizebox{.45\hsize}{!}{$
		\left\{ 
		\begin{array}{l}
			\mathcal{R} (n-1)\left[1-x_{1} \right]y x-x_{1}=0 \\
			\hfill \\
			\mathcal{R}\left[1-x \right]y x_{1}-x=0\\
			\hfill \\
			\mathcal{R} (n-2) \left[1-y \right]x_{1}y-y=0
		\end{array}
		\right.$}
	\end{equation}
	The resolution of the previous system is very cumbersome, and a closed expression is not particularly useful. Nonetheless, we can get some information about the steady states. First, the relationship between the value of the asymptotic probability for the node $1$ and that for the other nodes in the network $G_P$, can be expressed as
	\begin{equation}
		x_{1}=\frac{\gamma}{\beta(n-2)} \frac{(n-3)x+1}{1-x}. 
	\end{equation}
	Since it must be $0<x_{1}<1$, we get an upper bound for $x$
	\begin{equation}
		x<\frac{ (n-2)\mathcal{R}-1}{(n-2)\mathcal{R}+(n-3)}<1.
	\end{equation}
	This value represents a worst-case scenario for the infection probability of the pendant nodes. For instance, for $n=6$, $\beta=0.005$ and $\gamma=0.001$, we get $x<0.826087$. Eq. (\ref{exacteqstar}) for $x$ can be used to compute the exact numerical solution for specific values of $n$ and $\mathcal{R}$:
	\begin{equation}
		\begin{split}
			& (n-1)(n-2)\mathcal{R} \left[(n-2)\mathcal{R}+(n-3) \right]x^3\\
			&+\left[ (n-3)^2+(n-1)(n-2)\mathcal{R}-(n-1)(n-2)^2\mathcal{R}^{2} \right ]x^2\\
			&+2(n-3)x+1=0.
		\end{split}
		\label{exacteqstar}
	\end{equation}
	With the same parameters as before, the exact solution of the previous equation is $x=0.818337$. By this result, we get also $x_{1}=0.9509388$ and $y=0.9474204$.

\nocite{*}
\bibliography{References}

\end{document}